\documentclass{CSML}
\pdfoutput=1

\usepackage{lastpage}
\newcommand{\dOi}{14(1:22)2018}
\lmcsheading{}{1--\pageref{LastPage}}{}{}%
{Jan.~16, 2017}{Mar.~20, 2018}{}

\usepackage{etoolbox} 

\usepackage[all]{xy}
\xyoption{2cell}
\xyoption{curve}
\UseTwocells
\SelectTips{cm}{}


\usepackage[utf8]{inputenc}
\usepackage{amsmath}
\usepackage{amssymb}
\usepackage{array}
\usepackage{multirow}
\usepackage{color}
\usepackage{xcolor}
\usepackage{mathpartir}

\usepackage{stmaryrd}

\usepackage{enumitem}

\newlength{\croutw}
\newlength{\crouth}
\newcommand{\crossout}[1]%
        {\settowidth{\croutw}{$#1$}\settoheight{\crouth}{$#1$}#1%
        \hspace{-1.0\croutw}\raisebox{0.3\crouth}{\rule{\croutw}{0.1ex}}}

\newcommand{\commentout}[1]{\ignorespaces}

\newcommand{\Set}{\mathsf{Set}}

\newcommand{\bit}{\begin{itemize}}
\newcommand{\eit}{\end{itemize}}

\newcommand{\comp}{\circ}

        {\begin{enumerate}}%
        {\end{enumerate}}
        {\begin{enumerate}}%
        {\end{enumerate}}
        {\begin{enumerate}}%
        {\end{enumerate}}
        {\begin{enumerate}}%
        {\end{enumerate}}

\newcommand{\Cls}{\mathcal}

\newcommand{\Op}{{op}}

\newcommand{\CE}{{\Cls E}}

\newcommand{\Sem}[1]{{[\![#1]\!]}}

\newcommand{\modimpl}{\to}

\newcommand{\Pfin}{\mathcal P_{\mathit{fin}}}

        {\smallskip\par\noindent\qquad$\begin{array}{@{}l@{{}={}}l}}%
        {\end{array}$\par\noindent}

\newlength{\myboxwidth}
	{\end{center}\end{figure}}

\newcommand{\Lang}{\mathcal{L}}

\newcommand{\gldiamond}[1]{\Diamond_{#1}}

\newcommand{\Nat}{{\mathbb{N}}}

\newcommand{\Rat}{{\mathbb{Q}}}

\newcommand{\SDist}{\mathcal{S}}
\newcommand{\SDistN}{\mathcal{S}^{\mathit{rc}}}
\newcommand{\SDistZ}{\mathcal{S}^Z}

\newcommand{\Bag}{\mathcal{B}}

\newcommand{\Prop}{\mathsf{Prop}}

\newcommand{\Rules}{\mathcal{R}}

\newcommand{\hearts}{\heartsuit}
\newcommand{\finsubset}{\subseteq_{\mathit{fin}}}

\newcommand{\FA}{\mathfrak{A}}
\newcommand{\Pow}{\mathcal{P}}
\newcommand{\Land}{\bigwedge}
\newcommand{\Lor}{\bigvee}

\newcommand{\Hilb}{\mathcal{H}}

\newcommand{\Bound}{\mathcal{B}}

\newcommand{\Con}{\mathsf{Con}}
\newcommand{\Cut}{\mathsf{Cut}}
\newcommand{\MCut}{\mathsf{MCut}}

\newcommand{\Seq}{\mathcal{S}}
\newcommand{\entails}{\vdash}

\newcommand{\To}{\Rightarrow}

\newcommand{\col}{\!:\!}
\newcommand{\lcomp}{\lceil}
\newcommand{\rcomp}{\rceil}

\newcommand{\lsem}{\llbracket}
\newcommand{\rsem}{\rrbracket}
\newcommand{\inv}{^{-1}}
\newcommand{\CPow}{\mathcal{Q}}

\newcommand{\Moni}{\rul{Mon}_i}
\newcommand{\Pastek}{\rul{Paste}^k_i}

\newcommand{\STcpl}{\mathit{ST}}			
\newcommand{\ST}[1]{\mathit{ST}_{#1}}			
\newcommand{\STwr}[1]{\mathit{STmod}_{#1}}
\newcommand{\HT}{\mathit{HT}}

\newcommand{\deq}{=}

\newcommand{\PrL}{\heartsuit}				
\newcommand{\PrLSet}{\mathrm{\Lambda}}			
\newcommand{\arty}{\mathsf{ar}}

\newcommand{\FoV}{\ensuremath{\mathsf{iVar}}}		
\newcommand{\SoV}{\Sigma}

\newcommand{\WoV}{\ensuremath{\mathsf{wVar}}}		

\newcommand{\eqF}{=}
\newcommand{\inF}[2]{#2(#1)}				
\newcommand{\botF}{\bot}				
\newcommand{\impF}{\to}					
\newcommand{\uqF}{\forall}

\newcommand{\eqvF}{\leftrightarrow}			
\newcommand{\orF}{\vee}					
\newcommand{\andF}{\wedge}				

\newcommand{\tsc}[1]{\textsc{#1}}		

\newcommand{\botM}{\bot}				
\newcommand{\impM}{\to}					
\newcommand{\gloM}{\mathsf{A}}				

\newcommand{\daa}{\downarrow\!}				

\newcommand{\FoF}{\mathrm{CPL}}				

\newcommand{\MoF}{\ensuremath{\mathrm{CML}_{\PrLSet}}}

\newcommand{\HyFA}{\ensuremath{\mathrm{H}_{\PrLSet}(@)}}
\newcommand{\HyFDA}{\ensuremath{\mathrm{H}_{\PrLSet}(\daa,@)}}
\newcommand{\HyFDG}{\ensuremath{\mathrm{H}_{\PrLSet}(\daa,\gloM)}}
\newcommand{\HyFUA}{\ensuremath{\mathrm{H}_{\PrLSet}(\forall,@)}}

\newcommand{\finin}{\finsubset}			

\newcommand{\sbst}[3]{#1[#3/#2]}		

\newcommand{\ntoc}{\mbox{ fresh for }} 

\newcommand{\axiom}[1]{\tsc{#1}}			
\newcommand{\axref}[1]{\tsc{#1}}			

\newcounter{ender}

\newcommand{\yprf}{\forall\vec{y}.}		

\newcommand{\lb}{\left(}				
\newcommand{\rb}{\right)}				

\newcommand{\card}[1]{|#1|}				

\newcommand{\Pvr}{P}
\newcommand{\Svr}{S}

\newcommand{\vmd}[2]{v[#1/#2]}
\newcommand{\lcmh}{\lceil}
\newcommand{\rcmh}{\rceil}
\newcommand{\cmh}[2]{\lcmh#1:#2\rcmh}

\newcommand{\ultra}{\mathfrak{U}}

\newcommand{\eat}[1]{}

\newcommand{\schv}[1]{\boldsymbol{#1}}
\newcommand{\sva}{\schv{p}}
\newcommand{\svb}{\schv{q}}
\newcommand{\svc}{\schv{r}}
\newcommand{\svd}{\schv{s}}
\newcommand{\svA}{\schv{A}}
\newcommand{\svB}{\schv{B}}

\newcommand{\svW}{\schv{W}}
\newcommand{\svX}{\schv{P}}

\newcommand{\RSch}{\mathsf{Rank1}}
\newcommand{\BSch}{\mathsf{Prop}}

\newcommand{\SchV}{\ensuremath{\mathsf{sVar}}}
\newcommand{\FV}{\ensuremath{\mathsf{FV}}}

\newcommand{\gM}{\mathfrak{M}}
\newcommand{\gN}{\mathfrak{N}}

\newcommand{\CPL}{\mathrm{CPL}}

\newcommand{\bel}{\flat\Lambda}
\newcommand{\bbe}{\flat\kern-1.4pt\flat}
\newcommand{\bbel}{\bbe\Lambda}

\newcommand{\gdash}{\vdash^{\Hilb\Rules}_{\bel}}
\newcommand{\bdpl}{\tsc{BdPL}${}_{\bel}$}

\newcommand{\colo}{\ensuremath{\mathfrak{c}}}

\newcommand{\Spa}{\mathcal{S}}

\newtheorem{claim2}{\textbf{Claim}}

\newcommand{\valu}{v}

\newcommand{\rul}[1]{\mathsf{#1}}

\newcommand{\bro}{\textup{(}\!}
\newcommand{\brc}{\textup{)}}

\newcommand{\nrule}[3]{\infer[#1]{#3}{#2}}

\newcommand{\refeq}[1]{(\ref{#1})}

\usepackage{hyperref}
\hypersetup{hidelinks}

\newcommand{\Dist}{\mathcal{D}}
\newcommand{\Sel}{\mathcal{S}}

\newcommand{\Atoms}{\mathsf{At}}

\newcommand{\fin}{\mathit{fin}}

\theoremstyle{plain}\newtheorem{defprop}[thm]{Definition and Proposition}

\ifdef{\clm}{}{ 
\newtheorem{clm}[thm]{Claim}
}

\usepackage{fixme} 
\FXRegisterAuthor{tlo}{atlo}{TL Old} 
\FXRegisterAuthor{tl}{atl}{TL Autumn'16} 
\FXRegisterAuthor{tln}{atln}{TL Autumn'17} 
\FXRegisterAuthor{dp}{adp}{DP} 
\FXRegisterAuthor{ks}{aks}{KS} 
\FXRegisterAuthor{ls}{als}{LS} 

\definecolor{cobalt}{rgb}{0.0, 0.28, 0.67}

\newcommand{\tlnew}[1]{\tlnnote[inline,marginclue]{\textcolor{cobalt}{#1}}}

\usepackage{proof}

\begin{document}
\title[Model Theory and Proof Theory of Coalgebraic Predicate Logic]{Model Theory and Proof Theory of \\Coalgebraic Predicate Logic}

\author{Tadeusz Litak}
\address{ Friedrich-Alexander-Universität Erlangen-N\"{u}rnberg }
\email{tadeusz.litak@fau.de}
 
\author{Dirk Pattinson}
\address{Australian National University}
\email{dirk.pattinson@anu.edu.au}

\author{Katsuhiko Sano}
\address{Graduate School of Letters, Hokkaido University}
\email{v-sano@let.hokudai.ac.jp}

\author{Lutz Schr\"{o}der}
\address{ Friedrich-Alexander-Universität Erlangen-N\"{u}rnberg }
\email{lutz.schroeder@fau.de}

\begin{abstract} 
  We propose a generalization of first-order logic originating in a
  neglected work by C.C. Chang: a natural and generic correspondence
  language for any types of structures which can be recast as
  Set-coalgebras. We discuss axiomatization and completeness results
  for several natural classes of such logics. Moreover, we show that an
  entirely general completeness result is not possible. We study the
  expressive power of our language, both in comparison with
  coalgebraic hybrid logics and with existing first-order proposals
  for special classes of Set-coalgebras (apart from relational
  structures, also neighbourhood frames and topological spaces).
  Basic model-theoretic constructions and results, in particular
  ultraproducts, obtain for the two classes that allow
  completeness---and in some cases beyond that. Finally, we discuss a
  basic sequent system, for which we establish a syntactic
  cut-elimination result.
\end{abstract}


\maketitle



\newcommand{\takeout}[1]{}

\begin{flushright}
\emph{Dedicated to Ji\v{r}\'i Ad\'amek on the occasion of his seventieth birthday}
\end{flushright}


\section{Introduction}\label{sec:intro}
\newcommand{\Neigh}{\mathcal{N}}

Modal logics are traditionally a core formalism in computer
science. Classically, their semantics is relational, i.e.\ a model
typically comes with a set of states and one or several binary
accessibility relations on the state set. However, non-relational
semantics of various descriptions have come to play an increasing
role, e.g.\ in concurrency, reasoning about knowledge and agency,
description logics and ontologies. Models may involve such diverse
features as concurrent games, as in coalition logic and
alternating-time temporal logic~\cite{AlurEA02,Pauly02};
probabilities~\cite{LarsenSkou91,FaginHalpern94,HeifetzMongin01};
integer weights as in the multigraph semantics of graded modal
logic~\cite{DAgostinoVisser02}; neighbourhoods~\cite{Chellas80}; and
selection functions or preference orderings as in the different
variants of conditional
logic~\cite{Lewis73,Chellas80}. \emph{Coalgebraic modal logic} serves
as a unifying framework for such non-relational modal
logics~\cite{CirsteaEA11}.

Relational modal logic  can be seen as a subset of first-order logic,
specifically as the bisimulation-invariant fragment as shown by van
Benthem for arbitrary models and later shown for finite models by
Rosen~\cite{BenthemThesis,Rosen97}. An analogous first-order
counterpart for coalgebraic modal logic has been introduced in
previous work by two of the authors~\cite{SchroderP10fossacs}. The
language described there does support a van Benthem/Rosen-style
theorem. It is quite expressive but has a fairly complex syntax with
three sorts, modelling states, sets of states, and composite states,
respectively, and is equipped with a carefully tuned Henkin-style
semantics. In the current work we develop \emph{coalgebraic predicate
  logic (CPL)}, a first-order correspondence language for coalgebraic
predicate logic that is slightly less expressive than the language
proposed originally but has a simpler syntax and a straightforward
semantics that does not require any design decisions. The naturality
of CPL is further corroborated by the fact that CPL is expressively
equivalent to hybrid logic (see the overview article by Areces and ten
Cate~\cite{ArecesTenCate07}) with satisfaction operators and universal
quantification (equivalently with the downarrow binder $\downarrow$
and a global modality). Thus, CPL not only serves as a correspondence
language for coalgebraic modal logic but also arises by adding a
standard set of desirable expressive features widely used in
specification and knowledge representation.

Our proposal originates in a largely forgotten paper by C.C. Chang
\cite{Chang73} who introduces a first-order logic of
\emph{Scott-Montague neighbourhood frames}, which in coalgebraic terms
can be seen as coalgebras for the doubly contravariant powerset
functor. Chang's original motivation was to simplify model theory for
what Montague called \emph{pragmatics} and to replace Montague's
many-sorted setting by a single-sorted one. Chang's contributions were
primarily of a model-theoretic nature. He provided
adaptations of the notions of (elementary) submodel/extension,
elementary chain of models and ultraproduct and established
a Tarski-Vaught theorem as well as downward and upward
L\"{o}wenheim-Skolem theorems.
Our syntax is a notational variant of Chang's syntax; semantically, we
generalize from neighbourhood frames to coalgebras for an underlying
set functor, thus capturing the full range of non-relational
modalities indicated above.

Our semantics naturally extends coalgebraic modal logic in that it is
parametrized over an interpretation of the modal operators as
predicate liftings~\cite{Pattinson03,Schroder08}. It can thus be
instantiated with modalities such as, for instance, the standard
relational $\Diamond$; with neighbourhood-based modalities as in
Chang's original setup; with probabilistic operators $L_p$ `with
probability at least $p$'; or with a binary conditional $\Rightarrow$
`if -- then normally'. We incorporate a unary $\hearts$ into a
first-order language by allowing formulas of the form 
\begin{equation*}
t \hearts \lcmh z: \phi \rcmh  
\end{equation*}
where $t$ is a term, $\phi$ is a formula of coalgebraic predicate
logic, $z$ is a (comprehension) variable.  Such a formula
stipulates that $t$ satisfies $\hearts$, applied to the set of all $z$
that satisfy $\phi$. For example, in standard modal logic over
relational semantics, the formula $x \Diamond \lcmh z: z = y \rcmh$
says that $x$ has $y$ as a (relational) successor. In the
probabilistic setting, the formula $x L_p \lcmh y: y \neq x \rcmh$
states that the probability of moving from $x$ to a different state is
at least~$p$.

As indicated above, CPL supports a van Benthem / Rosen type result
stating essentially that coalgebraic modal logic is the
bisimulation-invariant fragment of CPL both over the class of all
structures and over the class of finite structures; this result is
proved in a companion paper~\cite{SchroderPL15:jlc}, which also
establishes a Gaifman-type theorem for CPL. In the current paper, we
establish the following results on CPL:
\begin{itemize}
\item We give a Hilbert-style axiomatization that we prove strongly
  complete for two particular classes of coalgebraic structures, viz.\
  structures that are either neighbourhood-like or \emph{bounded},
  where the latter type includes the relational and the graded case as
  well as positive Presburger modalities. As usual in model theory, strong completeness can be supplemented with a suitable variant of the Omitting Types Theorem.
\item While boundedness is a rather strong condition on structures, we
  show that the condition is fairly essential for completeness in the
  sense that within a much broader type of \emph{$\omega$-bounded}
  structures, the bounded structures are the only ones that allow for
  strong completeness. 
\item As indicated above, we establish the equivalence of CPL and
  several natural variants of coalgebraic hybrid logic.
\item We prove some basic model-theoretic results. Specifically, we
  show that, under the same (alternative) assumptions as for our
  completeness result, ultraproducts exist and a downward
  L\"owenheim-Skolem theorem holds; in fact, it turns our that the
  latter is applicable more broadly, requiring as it does only
  $\omega$-boundedness in place of boundedness in its corresponding
  variant.
\item We give sequent systems complementing the above-mentioned
  Hilbert system, and establish completeness, under the same
  (alternative) assumptions as for the Hilbert system, and more interestingly, 
  syntactic cut-elimination for the ``neighbourhood-like'' case. 
\end{itemize}

\noindent The material is organized as follows.  In \S
\ref{sec:ss} we introduce the syntax and semantics of CPL and give a
number of intuitive examples. In \S~\ref{sec:completeness} we
discuss the Hilbert-style axiomatization and associated completeness
results. We proceed to clarify the relationship between CPL and
several variants of coalgebraic modal and hybrid logic in \S\ 
\ref{sec:cml}.
In \S~\ref{sec:modeltheory} we take first steps in the model theory of
CPL, and \S~\ref{sec:proof} deals with proof theory. While our
presentation of the basic definitions in coalgebraic logic is
self-contained in principle, we do import some of its basic
results. Additional information is found in work on coalgebraic finite
models~\cite{Schroder06}, one-step
rules~\cite{SchroderPattinson09a,SchroderPattinson10b}, and
modularization of coalgebraic logics~\cite{SchroderPattinson11}.





\subsection{Related Work.} \label{sec:other}
As already discussed, the syntax of our logic follows Chang's
first-order logic of neighbourhood frames~\cite{Chang73}. %
An alternative, two-sorted language for neighbourhood frames has been
proposed by Hansen et al.~\cite{HansenKP09lmcs}. Over neighbourhood
frames, the language studied in the present work is a fragment of the
two-sorted one; we give details in
\S~\ref{sec:ss}. 





First-order formalisms have also been considered for topological
spaces, which are particular instances of neighbourhood frames when
defined in terms of local neighbourhood bases. In particular,
Sgro~\cite{Sgro80ams} studies interior operator logic in topology with
interior modalities for finite topological powers of the
space. 
This language is the weakest one in the hierarchy of topological
languages considered in an early overview by
Ziegler~\cite{Ziegler85}. Makowsky and Marcja \cite{MakowskyM77:mlq}
prove a range of completeness theorems for topological logics,
including a completeness result for the Chang language itself, i.e., a
special version of our Theorem~\ref{th:completeness}. See also ten
Cate et al. \cite{CateGS09apal} for a more contemporary
reference.  
Despite the fact
that CPL combines quantifiers and modalities, it should not be
confused with what is usually termed quantified or first-order modal
logic; see Remark~\ref{rem:qml}.


As mentioned above, our logic is less expressive but more naturally
defined than the correspondence language used in the first van
Benthen/Rosen type characterization result for coalgebraic modal
logic~\cite{SchroderP10fossacs}. Axiomatizations and model-theoretic
results as we develop here are not currently available for the more
expressive language of~\cite{SchroderP10fossacs}. 


A different generic first-order logic largely concerned with the
Kleisli category of a monad rather than with coalgebras for a functor
is introduced and studied in~\cite{Jacobs10predicate}. Of all the
languages discussed above, this one seems least related to the present
one; indeed, the study of connections with languages like that of the
original, three-sorted variant \cite{SchroderP10fossacs} is mentioned
by Jacobs \cite{Jacobs10predicate} as a subject for future
research. 

This paper is based on results first announced in earlier conference
papers~\cite{LitakPSS12:icalp,LitakSP13:tbillc}. Compared to the
conference versions, it features full proofs and additional examples. Some results previously only mentioned such as the Omitting Types Theorem (Theorem \ref{th:omt}) are explicitly stated and proved here for the first time. We also corrected a number of errors and typos. Most notably, as reconstructing the proof of cut-elimination for the $\mathsf{G3c}$-style system proposed in \cite{LitakSP13:tbillc} proved problematic, we replaced it with a $\mathsf{G1c}$-style system in this version, with a different treatment of equality and provided all the proof details. 


\section{Syntax, Semantics and Examples} \label{sec:ss}

We proceed to give a formal definition of \emph{coalgebraic predicate
  logic} (CPL). We fix a set $\SoV$ of predicate symbols and a modal
similarity type $\PrLSet$, i.e.\ a set of modal operators. Modal
operators $\PrL \in \PrLSet$ and predicate symbols $P \in \SoV$ both
come with fixed \emph{arities} $\arty{\PrL},\arty{\Pvr}\in\Nat$. The
set $\CPL(\PrLSet,\SoV)$ of CPL \emph{formulas} over $\PrLSet$ and
$\SoV$ is given by the grammar
\begin{equation*}
 \CPL(\PrLSet,\SoV) \owns \phi, \psi \;  ::= \;  
 y_1 \eqF y_2\mid 
 \Pvr(\vec x)
 \mid \botF \mid \phi \impF \psi \mid \uqF x.\phi \mid 
 x \PrL \cmh{y_1}{\phi_1}\dots\cmh{y_n}{\phi_n}
\end{equation*}
where $\PrL \in \PrLSet$ is an $n$-ary modal operator and
$\Pvr \in \SoV$ a $k$-ary predicate symbol, $x,y_i$ are variables from
a fixed set $\FoV$ we keep implicit. We just write $\CPL(\PrLSet)$ for
$\CPL(\PrLSet,\emptyset)$ and sometimes we omit $\PrLSet$ and $\Sigma$ altogether. 
Booleans and the existential quantifier are defined in the standard
way.  We do not include function symbols, which can be added in a
standard way \cite{Chang73}. We adopt the usual convention that the
scope of a quantifier extends as far to the right as possible.
In the $\cmh{y_i}{\phi_i}$ component, $y_i$ is used as a comprehension
variable, i.e.,
$\cmh{y_i}{\phi_i}$ 
denotes a subset of the carrier of the model, to which modal operators
can be applied in the usual way. In
$x\PrL\cmh{y_1}{\phi_1}\dots\cmh{y_n}{\phi_n}$, $x$ is free and $y_i$
is bound in $\phi_i$, otherwise the notions of freeness and
boundedness are standard. We write $\FV(\phi)$ for the set of free
variables of a formula $\phi$. A variable is \emph{fresh} for a
formula if it does not have free occurrences in it; to save space, we
will also sometimes say that $x \in \FoV$ is fresh for $y \in \FoV$
whenever it is distinct from it. A
\emph{sentence}, 
as usual, is a formula without free
variables. 

As usual, some care is needed when defining substitution to avoid, on
the one hand, capture of newly substituted variables by quantifiers
and on the other hand, substituting for a bound
variable.  
We take as our model the discussion in Enderton's monograph
\cite[p.~112--113]{Enderton72}.  As we now have two ways in which a
variable can become bound and the binder $\PrL$ involves also a variable/term in a non-binding way, it is desirable to spell out
details. 
We thus define---prima facie not necessarily capture-avoiding---substitution $\alpha[t/x]$ with $t, x \in \FoV$ (had we allowed for
function symbols, $t$ could be any term) as replacing $x$ with $t$ in
atomic formulas and commuting with implication (and of course other
Boolean connectives, were they taken as primitives). For binders, the
clauses are:
\begin{gather*}
  \begin{aligned}
    (\forall x.\phi)[t/y] & = \begin{cases} \forall x. \phi & x =y \\ \forall x. \phi[t/y] & \text{ otherwise, } \end{cases} \\
    (x\PrL\cmh{z_1}{\phi_1}\dots\cmh{z_n}{\phi_n})[t/y] & =
    u\PrL\cmh{z_1}{\phi'_1}\dots\cmh{z_n}{\phi'_n}
  \end{aligned}
\\
 \text{ where }  \phi'_i   = \begin{cases} \phi_i & y =z_i \\  \phi'_i[t/y] & \text{ otherwise, } \end{cases}\text{ and }u=
 \begin{cases}
   t & x = y\\
   x & \text{otherwise.} 
 \end{cases}
\end{gather*}
This of course cannot work without restrictions, so we follow Enderton
in defining the notion of \emph{substitutability} of $t$ for $x$ in a
term. There are no restrictions on substitutability in atomic
formulas, and for implications, it is defined as substitutability in
the two argument formulas. Finally, $t$ is substitutable for $x$ in
$z\PrL\cmh{y_1}{\phi_1}\dots\cmh{y_n}{\phi_n}$ whenever for every $i$,
$t$ is either
\begin{itemize}
\item fresh for $\cmh{y_i}{\phi_i}$ (this includes the case $t = y_i$) or
\item different from $y_i$ and substitutable for $x$ in $\phi_i$.
\end{itemize}
Note that in the first alternative, substituting $t$ for $x$ has no
effect on $\cmh{y_i}{\phi_i}$. In a language with more general terms,
the second alternative would require that $y_i$ is fresh for $t$
rather than different from $t$. Substitutability of $t$ for $x$ in
$\forall y_1.\phi_1$ is defined similarly (and standardly).

We depart from Enderton's conventions by restricting, from now on, the
usage of the $\alpha[t/x]$ notation to the case where $t$ is
substitutable for $x$ in $\alpha$ (as usual, when this is not the case
the substitution can still be applied after suitably renaming bound
variables in $\alpha$). For example, the axiom scheme
$\forall \vec{y}. (\forall x. \phi \to \phi[z/x])$ denoted as
\axref{En2} in Table \ref{tab:folaxioms} has as its valid instances
only those formulas where $z$ is substitutable for $x$.



The semantics of CPL is parametrized over the choice of an endofunctor
$T$ on $\Set$ that determines the underlying system type: models are
based on \emph{$T$-coalgebras}, i.e.\ pairs $(C, \gamma: C \to TC)$
consisting of a carrier set $C$ of worlds or \emph{states} and a
\emph{transition function} $\gamma$. We think of the elements of $TC$
as being \emph{composite states}; e.g.\ if~$T$ is the identity functor
then a composite state is just a state, and if $T$ is powerset, then a
composite state is a set of states. Thus, the transition function
assigns to each state $c$ a composite state $\gamma(c)$ that
represents the successors of $c$ and that we correspondingly refer to
as the \emph{composite successor} of~$c$. E.g.\ in case $T$ is
powerset, a $T$-coalgebra assigns to each state a set of successor
states, and hence is essentially a Kripke frame.

To interpret the modal operators, we extend $T$ to a
\emph{$\Lambda$-structure}, i.e.\ we associate to every $n$-ary modal
operator $\hearts \in \Lambda$ a set-indexed family of mappings
\begin{equation*}
  \lsem \hearts \rsem_C: (\CPow C)^n \to \CPow TC
\end{equation*}
where $\CPow$ denotes the contravariant powerset functor, subject to
\emph{naturality}, i.e.
\begin{equation*}
(Tf)\inv \circ \lsem \hearts \rsem_C = \lsem \hearts \rsem_D \circ
(f\inv)^n
\end{equation*}
for every set-theoretic function $f: C \to D$. In categorical
parlance, this means that $\lsem\hearts\rsem$ is a natural
transformation
  $\CPow^n\to\CPow\comp T^\Op$;
  we recall that the contravariant powerset functor $\CPow$ maps a set
  $X$ to the powerset of $X$ and a map $X\to Y$ to the preimage map
  $\CPow f:\CPow Y\to\CPow X$, i.e.\ $\CPow f(A)=f^{-1}[A]$ for
  $A\subseteq Y$. Each such $\lsem \hearts \rsem$ is called a \emph{predicate lifting}~\cite{Pattinson03,Schroder08}. Formally speaking, we should define a
  $\Lambda$-structure as a pair
  $(T, (\lsem \PrL \rsem)_{\PrL \in \Lambda})$, but to avoid
  cumbersome notation and terminology, we will speak about a
  $\Lambda$-structure \emph{based} on $T$ (or a $\Lambda$-structure
  \emph{over} $T$) and suppress the second component of the pair
  whenever $(\lsem \PrL \rsem)_{\PrL \in \Lambda}$ is clear from the
  context.

A triple $\gM = (C, \gamma, I)$ consisting of a coalgebra
$\gamma: C \to TC$ and a predicate interpretation
$I: \SoV \to \bigcup\limits_{n \in \omega}\CPow(C^n)$ respecting
arities of symbols will be called a \emph{(coalgebraic) model}. In
other words, a coalgebraic model consists simply of a $\Set$-coalgebra
and an ordinary first-order model whose universe coincides with the
carrier of the coalgebra.  Given a model $\gM = (C, \gamma,
I)$ 
and a valuation $v: \FoV \to C$, we define
satisfaction $\gM, v \models \phi$ in the standard way for
first-order connectives and for $\hearts$ by the clause
\begin{equation*}
\gM, v \models x \hearts \lcmh y_1 : \phi_1 \rcmh
\dots \lcmh y_n: \phi_n \rcmh \iff 
\gamma(v(x)) \in \lsem \hearts \rsem_C( \lsem \phi_1 \rsem^{y_1}_{\gM},
\dots, \lsem \phi_n \rsem^{y_n}_{\gM}) 
\end{equation*}
where 
\begin{equation} \label{eq:deny}
\lsem \phi \rsem^y_{\gM} \deq \lbrace c \in C \mid \gM, \vmd{c}{y}
\models \phi \rbrace
\end{equation}
and $\vmd{c}{y}$ is $v$ modified by mapping $y$ to $c$.

\begin{rem} \label{rem:qml}
  \emph{Quantified} or \emph{first-order modal logic} in the sense
  used widely in the literature~(see, e.g., \cite{Garson01}) combines
  quantification and modalities in a two-sorted and, effectively,
  two-dimensional semantics: One has an underlying set of worlds as
  well as an underlying set of individuals, with modalities
  interpreted as moving between worlds and quantification interpreted
  as ranging over individuals in the current world. We emphasize that
  although CPL also combines modalities and quantifiers, it is not a
  quantified modal logic in this sense: it is interpreted over a
  single set of individuals, and both the modalities and the
  quantifiers move within this set. In particular, the instance of CPL
  induced by the standard modalities equipped with their usual
  predicate liftings is standard first-order logic rather than a
  quantified modal logic, as we discuss below in some detail.
\end{rem}
\smallskip\noindent In a companion paper on the van Benthem-Rosen
theorem for CPL \cite{SchroderPL15:jlc} and in conference papers the
present work is based upon \cite{LitakPSS12:icalp,LitakSP13:tbillc},
we have focused on Chang's original motivation for this
language~\cite{Chang73}. Namely, Chang saw his setup as a modification
of Montague's account of \emph{pragmatics}, tailored to reasoning
about social situations and relationships between an individual and
sets of individuals. We have proposed a series of examples kept in the
same spirit, utilizing Facebook, Twitter and social networks.  In the
present paper, we offer examples based on, so to say, networking of a
more low-level character, especially \emph{delay-} or
\emph{disruption-tolerant} networking
(\emph{DTN}). 
We do not claim to be very accurate with respect to specifications of
concrete protocols; our examples are of purely inspirational and
illustrative character. It is worth mentioning, though, that such
routing and forwarding protocols can be backed by social insights
\cite{HuiCY08}, so in a sense we are still following the spirit of our
original examples.\footnote{In particular, Hui et al. \cite{HuiCY08}
  gave us the idea of using \emph{subcommunities} in this context.}

\smallskip
\subsection*{Neighbourhood Frames.}
Scott/Montague \emph{neighbourhood semantics} is captured
coalgebraically using $\Lambda \deq \lbrace \Box \rbrace$ and putting
$TC \deq \CPow\CPow C$ (the doubly contravariant powerset functor),
which extends to a $\Lambda$-structure by 
\begin{equation*}
\lsem \Box \rsem_C (A) \deq \lbrace \sigma \in TC \mid A \in \sigma \rbrace.
\end{equation*}
A $T$-coalgebra then associates to each state a set of sets of states,
i.e.\ a system of neighbourhoods; thus, $T$-coalgebras are just
neighbourhood frames. In the presence of a binary relation
$\Svr(x, y)$ that we read as `node/router $y$ is in the forwarding
table of $x$' and interpreting $\Box$ as `is a recognized
subcommunity', the formula
\begin{equation*}
  \exists y. x \Box \lcmh z: \Svr(z, y) \rcmh 
\end{equation*}
reads as `there exists a certain $y$ such that amongst the
subcommunities recognized by $x$, there is one formed exactly by those
having $y$ in its forwarding table'.

The instance of CPL that we obtain in this way is, up to quite minor
syntactic differences, Chang's original language~\cite{Chang73}. As
mentioned in \S~\ref{sec:intro}, it embeds as a fragment into
Hansen et al.'s two-sorted correspondence
language~\cite{HansenKP09lmcs}. We refrain from giving full syntactic
details; roughly, the setup is as follows. The two-sorted language has
sorts $s$ for states and $n$ for neighbourhoods, and features binary
infix predicates $\mathsf{N}$ and $\mathsf{E}$ respectively modelling
the neighbourhood relation between states and neighbourhoods, and the
inverse elementhood relation $\owns$ between neighbourhoods and
states. Then our $x\PrL\cmh{y}{\phi(y)}$ can be translated as
$\exists u.(x\mathsf{N}u\wedge \forall y.(u\mathsf{E}y \eqvF
\phi(y)))$. 


\subsection*{Relational first-order logic.}
Instantiating CPL with the usual modalities of relational modal logic,
specifically the logic $K$,
we obtain a notational variant of ordinary FOL over relational
structures, that is, of the usual correspondence language. The main
idea has already been indicated in the introduction: we encode the
successor relation in formulas of the form $x\Diamond\lcmh
z: y=z\rcmh$, which states that $y$ is a successor of
$x$.
Formally, we capture the standard modality and the propositional atoms
of the relational modal logic $K$ in the similarity type
\begin{equation*}
  \Lambda=\{\Diamond\}\cup\Atoms
\end{equation*}
where $\Atoms$
is a set of propositional atoms; as expected, $\Diamond$
is unary, and $a\in\Atoms$
is nullary. We interpret these operators over the functor $T$
given on objects by 
\begin{equation*}
  TX=\Pow X \times\Pow\Atoms
\end{equation*}
where $\Pow$
denotes the covariant powerset functor. That is, a coalgebra
$\gamma:C\to
TC$ assigns to each state $c\in
C$ a set of successors as well as a set of propositional atoms valid
in~$c$. The interpretation is defined by means of predicate liftings
\begin{align*}
  \Sem{\Diamond}_X(A)&=\{(Y,U)\in\Pow X\times\Pow\Atoms\mid A\cap Y\neq\emptyset\}\\
  \Sem{a}_X & =\{(Y,U)\in\Pow X\times\Pow\Atoms\mid a\in U\} 
\end{align*}
where, corresponding to the arity of the modal operators, the
predicate lifting for $\Diamond$ is unary and the predicate liftings
for the $a\in\Atoms$ are nullary. These predicate liftings capture
precisely the standard semantics of both $\Diamond$ and the
propositional atoms. In particular, the above-mentioned formula
$x\Diamond\lcmh z: y=z\rcmh$ really does say that $y$ is a successor
of $x$. (Notice that in the nullary case, our syntax instantiates to
formulas $x\,a$ saying that $x$ satisfies the propositional atom~$a$).

The standard first-order correspondence language of modal logic has
unary predicates $a$ for the atoms $a\in\Atoms$ and a binary predicate
$R$ to represent the successor relation. We translate $\CPL(\PrLSet)$
as defined above into the standard correspondence language by just
extending the standard translation of modal logic to CPL, with the
modification that the current state is represented by an explicit
variable in CPL so that it is no longer necessary to index the
standard translation with a variable name. That is, our translation
$\STcpl$ is defined in the modal cases (which by our conventions
include the case of propositional atoms) by
\begin{align*}
  \STcpl(x\Diamond\lcmh y:\phi\rcmh) & = \exists y.\,R(x,y)\land\STcpl(\phi)\\
  \STcpl(x\,a) & = a(x) && (a\in\Atoms)
\end{align*}
and by commutation with all other constructs. In the converse
direction, we translate $R(x,y)$ into $x\Diamond\lcmh z:z=y\rcmh$ and
$a(x)$ into $x\,a$. In summary, \emph{CPL over
$\Lambda=\{\Diamond\}\cup\Atoms$ with the above semantics is
expressively equivalent to the standard first-order correspondence
language of modal logic}.

\subsection*{Graded Modal Logic.} 
We obtain a variant of graded modal logic \cite{Fine72} if we consider
the similarity type
$\Lambda = \lbrace \langle k \rangle \mid k \geq 0 \rbrace$ where
$\langle k \rangle$ reads as `more than $k$ successors satisfy
\dots'. We interpret the ensuing logic over
\emph{multigraphs}~\cite{DAgostinoVisser02}, which are coalgebras for
the \emph{multiset functor} $\Bag$ given on objects by
\begin{equation*}
\Bag X \deq \lbrace \mu: X \to \Nat\cup\{\infty\} \mid f\text{ a map}\rbrace.
\end{equation*} 
We see such a map $\mu: X \to \Nat\cup\{\infty\}$ as an
integer-valued discrete measure on $X$, i.e.\ we write
$\mu(A)=\sum_{x\in A}\mu(x)$ for $A\subseteq X$. Then, $\Bag$ acts on
maps $f:X\to Y$ by taking image measures; i.e.\
$\Bag f(\mu)(B)=\mu(f^{-1}[B])$ for $B\subseteq Y$. We extend $\Bag$
to a $\PrLSet$-structure by stipulating
\begin{equation*}\textstyle
\lsem \langle k \rangle \rsem_X(A) = \lbrace \mu \in \Bag X \mid
\mu(A) > k \rbrace
\end{equation*}
to express that more than $k$ successors (counted with multiplicities)
have property $A$. Note that over Kripke frame, graded operators can
be coded into standard first order logic; the difference with standard
first-order logic arises through the multigraph semantics, for which
the requisite expressive means arise only through the graded
operators.

Continuing our line of routing examples, we can, given a
$\Bag$-coalgebra $\gamma:C \to \Bag C$, think
of 
$\gamma(c)(c')$ as the number of packets forwarded from $c$ to $c'$ in
the past hour.\lsnote{Changed this slightly to motivate use of
  integers}
In the presence of a binary relation $\Svr(x, y)$ interpreted as
above, the formula
\begin{equation*}
\neg\exists y. (x \langle k \rangle \lcmh z: \Svr(y, z) \rcmh)
\end{equation*}
then expresses that there is no router $y$ s.t. the total number of
packets sent by $x$ to nodes in $y$'s forwarding table in the past
hour exceeds $k$.




\medskip

\subsection*{Presburger modal logic and arithmetic.}
A more general set of operators than graded modal logic is that of
positive Presburger modal logic~\cite{DemriLugiez06}, which admits
integer linear inequalities
$\sum_{i=1}^n a_{i} \cdot \# (\phi_{i}) > k$ among formulas where
$a_{i} \geq 0$ for all~$i$. We see such a formula as an application of
an $n$-ary modality $L_k(a_1,\dots,a_n)$ to formulas
$\phi_1,\dots,\phi_n$, and interpret this modality over the multiset
functor $\Bag$ as introduced above by the $n$-ary predicate lifting
\begin{equation*}\textstyle
  \Sem{L_k(a_1,\dots,a_n)}_X (A_1, \dots, A_n) =\{\mu\in\Bag X\mid 
  \sum_{i=1}^n a_{i} \cdot \mu (A_{i}) > k\}. 
\end{equation*}
In addition to the binary predicate $S$, let us also introduce unary
predicate $O(x)$ expressing that $x$ is an overloaded node. The
formula
\begin{equation*}
\forall x. ((xL_{10{,}000}(1,3)\lcmh y : S(x,y) \rcmh  \lcmh y : O(y) \rcmh) \to O(x))
\end{equation*}
means that, if the weighted number of packets sent by $x$ to
overloaded nodes combined with packets $x$ sends to all nodes in its
forwarding table exceeds $10{,}000$, then $x$ itself is overloaded.
 

\subsection*{Combination of Frame Classes.} 
Frame classes can be combined: 
  we can take $T \deq \Bag \times \CPow\CPow$ and combine operators for packet counting and subcommunity recognition. A formula 
\begin{equation*}
\neg x\Box\cmh{y}{y \langle 30 \rangle\cmh{u}{u \neq z}}
\end{equation*}
 expresses then that the collection of those nodes  which have forwarded  more than 30 packages to servers different than $z$ in the past hour is not a subcommunity recognized by $x$.


\subsection*{Probabilistic Modal Logic.}
The \emph{discrete distribution functor} $\Dist$ is defined on objects
by
\begin{equation*}\textstyle
\Dist X = \lbrace \mu: X \to [0, 1] \mid \sum_x \mu(x) = 1 \rbrace,
\end{equation*}
and on morphisms by taking image measures exactly as for the multiset
functor $\Bag$ discussed earlier. Coalgebras for $\Dist$ thus
associate to every state a probability distribution over successor
states; such structures are known as \emph{Markov chains}, or 
\emph{probabilistic transition systems}, or \emph{type spaces}. Taking the
similarity type
$\Lambda = \lbrace \langle p \rangle \mid p \in [0, 1] \cap \Rat\}$,
with $\langle p \rangle$ read as `with probability more than $p$'
(thus departing from the choice of operators $L_p$ ` with probability
at least $p$' that we used in the introduction), we formally interpret
$\langle p \rangle$ using the predicate lifting
\begin{equation*}
  \lsem \langle p \rangle \rsem_X (A) = \lbrace \mu \in \Dist X \mid
  \mu(A) > p \rbrace
\end{equation*}
over $\Dist$. We thus obtain a form of probabilistic first-order logic
for probabilistic transition systems that extends probabilistic modal
logic~\cite{LarsenSkou91,FaginHalpern94,HeifetzMongin01}. Continuing
our line of routing examples, if we interpret the transition
probabilities as the likelihood of a server forwarding any given
packet to another, then the formula
\begin{equation*}
  \forall x,y.( x \langle 1/2\rangle\lcmh z: z=y\rcmh \to  y \langle 1/2
  \rangle \lcmh z: z =
  x \rcmh )
\end{equation*}
expresses a partial form of symmetric connectivity: whenever a server
$x$ prefers the connection to $y$ in the sense that it will more
likely than not route any given packet through $y$, then the same will
hold in the other direction.

We obtain a version of this logic with finitely many modal operators
in situations where all possible probabilities are contained in some
finite set of rationals (such as when rolling a fair die). We then
consider substructures of the form
\begin{equation*}
\Dist_k X = \lbrace
\mu\in\Dist(X)\mid \mu(x)\in\{i/k\mid i=0,\dots,k\} \rbrace,
\end{equation*}
restricting the modal operators to come from
$\Lambda_k = \lbrace \langle n/k \rangle \mid n = 0, \dots, k
\rbrace$.

\newcommand{\CTo}{\Rightarrow} 
\newcommand{\CToDual}{>} \newcommand{\ff}{\mathsf{ff}}
\newcommand{\tf}{\mathsf{tf}}
\subsection*{Non-Monotonic Conditionals.}
An example of a binary modality is provided by (conditional)
implication $\CToDual$. Such operators are
interpreted over a variety of semantic structures; one of these
involves \emph{selection function frames}, which in our terminology
can be defined as coalgebras for the \emph{selection function functor}
$\Sel$. The latter acts on objects by
\begin{equation*}
  \Sel X = \lbrace f: \CPow X \to \Pow X \rbrace
\end{equation*}
and, correspondingly, on maps $f:X\to Y$ by
$\Sel f = \CPow f\to \Pow f:\Sel X\to\Sel Y$, i.e.\
$\Sel f(g)(A)=f[g(f^{-1}[A])]$ for $A\subseteq X$ (recall that $\CPow$
denotes the contravariant powerset functor). We think of $f_x \in\Sel X$
as selecting the set $f_x(A)$ of worlds  which $x$ sees as `most typical' given a condition
$A\subseteq X$. 
 Over this functor, we interpret the conditional $\CToDual$ 
by the predicate lifting
\begin{equation*}
  \lsem \CToDual \rsem_X(A, B)  = \lbrace f \in
  \Sel X \mid
  f(A) \cap B \neq \emptyset \rbrace.
\end{equation*}
The formula $\phi \CToDual \psi$ expresses that $\psi$ is typically
possible under condition $\phi$. This presentation of conditional
logic is dual to the standard presentation \cite{Chellas80} in terms
of a binary operator~$\CTo$ `if -- then normally', related to
$\CToDual$ by $a\CToDual b\equiv \neg(a\CTo\neg b)$. For our purposes,
$\CToDual$ has the technical advantage of being bounded in the second
argument in a sense that we will introduce in \S~
\ref{sec:completeness}. 

We continue to interpret our examples in the context of routing: Given
an $\Sel$-coalgebra (i.e.\ selection function frame)
$\gamma:C\to\Sel C$, we may read $\gamma(c)(A)$ as the set of those
servers through which server $c$ will normally route an incoming
packet if $c$ is currently active in the subcommunity $A\subseteq C$.
Then a formula of the form
\begin{equation*}
  \forall x,u.(\phi(u)\to x\CToDual \lceil y:\phi(y)\rceil \lceil z:z=u \rceil)
\end{equation*}
says that if $x$ is currently active in a subcommunity delineated by
the formula $\phi(y)$ and $u$ belongs to that subcommunity, then $u$
is normally a possible target for packets forwarded
by~$x$.

We have not mentioned propositional atoms other than in the example on
relational first-order logic. We introduce an explicit notion of
propositional atom as part of a modal signature:
\begin{defi}\label{def:atom}
  A nullary modality $p\in\Lambda$ is a \emph{propositional atom} if
  $T$ decomposes as $T=T'\times 2$ and under this decomposition,
  $\Sem{p}_X= T'X\times\{\top\}$.
\end{defi}
\begin{rem} \label{rem:atomsafe}
  Propositional atoms can easily be added to all our examples by just
  extending the functor with a component for their interpretation, as
  indicated in the above definition. Explicitly, if $\Atoms$ is a set
  of propositional atoms and $T'$ is a functor, then the atoms
  $p\in\Atoms$ give rise to nullary modalities $p$, interpreted over
  $T=T'\times\Pow(\Atoms)$ by
  $\Sem{p}_X=\{(t,U)\in T'X\times\Pow(\Atoms)\mid p\in U\}$. In fact,
  this is an instance of what we have called \emph{combination of
    frame classes} above: Systems featuring only propositional atoms
  with pointwise valuations, and no further transition structure, are
  captured as coalgebras for the constant functor $\Pow(\Atoms)$, and
  this system type can be freely combined with others. A coalgebraic
  framework for such combinations of system types and modalities is
  afforded by multisorted coalgebra, essentially following the
  principle of converting components of the coalgebraic type functor
  into sorts~\cite{SchroderPattinson11}. General results in this
  framework imply that completeness properties transfer smoothly to
  such logic combinations (essentially, to \emph{fusions} in standard
  terminology). In particular, completeness will always extend without
  further ado under adding propositional atoms to a logic, so we will
  continue to largely elide them in the discussion of examples.
\end{rem}


\section{Completeness} \label{sec:completeness}

In \S~ \ref{sec:hilb-axiom} below, we propose an axiom system for CPL,  sound wrt arbitrary structures  (Theorem \ref{thm:hilb-soundness}) and in \S~\ref{sec:proofco} we show its completeness wrt structures s.t. each operator on every coordinate is either ``neighbourhood-like'' or ``Kripke-like'' (Theorem \ref{thm:hilb-complete}). As discussed in \S~\ref{sec:omega}, even a mild relaxation of these conditions makes a generic completeness result impossible.

However, not only for the proof, but even for the statements of our completeness result, or of the axiomatization itself, we need some spadework.\lsnote{@Tadeusz: I strongly suggest to reduce the technical development to unary modalities from here on. Section 5 already assumes this anyway.}


\subsection{S1SC and Boundedness} \label{sec:sone}

In order to state our axiomatization and completeness results, we need several notions from coalgebraic model theory. The first of them, central to the entire edifice, is that of  \emph{one-step satisfiability}. 

\begin{defi}[One-step logic] 
\begin{itemize}\
\item Given a supply of primitive symbols $D$ (which can be any set),
  define the set $\BSch(D)$ of \emph{Boolean $D$-formulas} (or
  \emph{propositions}) as
  $$\svA, \svB ::=  d \mid \svA \to \svB \mid \bot$$ 
  where $d \in D$, and the set $\Lambda (D)$ of \emph{modalized
    $D$-formulas} as
  $$\Lambda (D) = \{ \PrL d_1\dots d_n \,|\,d_{1},
  \dots, d_{n} \in D \text{ and } \heartsuit \in \Lambda \mbox{ is
    $n$-ary}\}.$$
  Then the set $\RSch(D)$ of \emph{rank-1 $D$-formulas} is defined as
  \begin{equation*}
    \RSch(D) = \BSch(\Lambda(\BSch(D)));
  \end{equation*}
  in other words, a rank-1 formula is a Boolean combination of
  formulas consisting of a modality from $\Lambda$ applied to Boolean
  combinations of atoms from $D$.
  
\item Given a set $C$ and a valuation $\tau: D \to \Pow(C)$, we extend
  $\tau$ to $\BSch(D)$ using the Boolean algebra structure of
  $\Pow(C)$, and then write $C, \tau \models \svA$ if
  $\tau(\svA) = C$, for $A\in\BSch(D)$.
\item Given the same data, we define the extension
  $\Sem{\phi}_{TC,\tau}\subseteq TC$ of $\phi\in\RSch(D)$ by extending
  the assignment 
  \begin{equation*}
    \lsem \PrL\svA_1\dots\svA_n \rsem_{TX, \tau} = \lsem \hearts
  \rsem_C( \tau(\svA_1), \dots, \tau(\svA_n))
  \end{equation*}
  using the Boolean algebra structure of $\Pow(TC)$.
 \item  We then write $TC, \tau \models \phi$ if
  $\lsem \phi \rsem_{TC, \tau} = TC$, and $t\models_{TC,\tau}\phi$ if
  $t\in\lsem \phi \rsem_{TC, \tau}$. 
\item If $D \subseteq \Pow(C)$ and $\tau$ is just the inclusion, we will
  usually drop it from the notation; in particular, for subsets
  $Y_1,\dots,Y_n\subseteq C$ and $\hearts\in\Lambda$ $n$-ary, we write
  $t\models\hearts(Y_1,\dots,Y_n)$ to mean
  $t\in\Sem{\hearts}_C(Y_1,\dots,Y_n)$.
\item A set $\Xi \subseteq \RSch$ is \emph{one-step satisfiable}
  w.r.t.\ $\tau$ if
  $\bigcap_{\phi \in \Xi} \lsem \phi \rsem_{TC, \tau} \neq \emptyset$.
\end{itemize}
\end{defi}


\noindent Just like in the case of coalgebraic modal logic (see \S~\ref{sec:cml} below),  proof systems for CPL are best described in terms of rank-1 rules -- or, more precisely, rule schemes --, which describe the geometry  of the $\Lambda$-structure under
consideration~\cite{Schroder06}. 

\begin{defi}[One-Step Rules]\label{def:hilonestep} 
  Fix a collection $\SchV$ of schematic variables
  $\sva, \svb, \svc \dots$
\begin{itemize}
\item A  \emph{one-step rule} is of the form $\svA / \svX$ where 
 $\svA \in \BSch(\SchV)$ and 
 $\svX \in \RSch(\SchV)$ is a disjunctive clause over
 $\Lambda(\SchV)$, i.e.\ a finite disjunction of formulas that are
 either in $\Lambda(\SchV)$ or negations of formulas in
 $\Lambda(\SchV)$. As usual, we impose moreover that every schematic
 variable is mentioned at most once in $\svX$, and every schematic
 variable occurring in $\svA$ occurs also in $\svX$.
 \item A rule $\svA / \svX$ is \emph{one-step sound} if

\begin{center}
$TC, \tau \models \svX$
whenever $C, \tau \models \svA$ for a valuation $\tau: \SchV \to \Pow(C)$.
\end{center}

\item Given a set $\Rules$ of one-step rules and a valuation $\tau: \SchV \to \Pow(C)$, a set $\Xi \subseteq \RSch(\SchV)$ is \emph{one-step consistent \bro with respect to $\tau$\brc}~\cite{SchroderPattinson10b} if the set 
\begin{equation*}
\Xi\cup\lbrace \svX
\sigma \mid\sigma:\SchV \to\BSch(\SchV) \text{ and } \svA/ \svX \mbox{ is a rule in }
\Rules \text{ s.t. } C, \tau \models \svA \sigma\rbrace
\end{equation*}
is propositionally consistent, where $\svA\mapsto\svA\sigma$ and
$\svX\mapsto\svX\sigma$ denote the obvious inductive extensions of
$\sigma$ to $\BSch(\SchV)$ and $\RSch(\SchV)$, respectively (in other
words, we see $\sigma$ as a substitution, and use postfix notation to
denote application of substitutions).
\end{itemize}
\end{defi}

\begin{asm} \label{conv:1ss} For purposes of the technical development
  \bro excluding the examples\brc, we fix from now on a set $\Rules$ of
  one-step sound one-step rules.
\end{asm}

\noindent We next introduce the two variants of \emph{one-step}
completeness that we need for our global completeness proof. By
\emph{one-step}, we mean that the completeness assumption is only made
for a very simple logic that precludes nesting of modal operators, and
hence can be interpreted over single elements of $TC$ rather than full
coalgebraic models.
\begin{defi}[Strong One-Step Completeness~\cite{SchroderPattinson10b}] \label{def:sosc}
  The rule set $\Rules$ is \emph{strongly one-step complete \bro S1SC, neighbourhood-like\brc} for
  a $\Lambda$-structure if
  \begin{center}
    for every $C \in \Set$, $\Xi \subseteq \RSch(\SchV)$, and
    $\tau: \SchV \to \Pow(C)$, \\ $\Xi$ is one-step satisfiable wrt
    $\tau$ whenever $\Xi$ is one-step consistent wrt $\tau$.
  \end{center}
  Similarly, $\Rules$ is \emph{finitary S1SC} if the same condition
  holds for $\tau$ restricted to be of type $\SchV \to \Pfin(C)$.
\end{defi}
\noindent By the usual argument, both forms of strong one-step
completeness imply corresponding forms of compactness:
\begin{defi}[One-step compactness]\label{def:os-compact}
  $\Lambda$-structure is \emph{one-step compact} if for every set $X$,
  every finitely satisfiable set
  $\Phi\subseteq\BSch(\PrLSet(\Pow X))$ of one-step formulas is
  satisfiable.  Similarly, a $\Lambda$-structure is \emph{finitary
    one-step compact} if for every set $X$, every finitely satisfiable
  set $\Phi\subseteq\BSch(\PrLSet(\Pfin X))$ of one-step formulas is
  satisfiable.
\end{defi}
\begin{lem}
  Every \bro finitary\brc\ S1SC $\Lambda$-structure is  \bro finitary\brc\ one-step
  compact. \qed
\end{lem}


\begin{rem}\label{rem:cong}
  Every (finitary) S1SC rule set derives the \emph{congruence rule}
  \begin{equation*}
    \frac{\sva_1\eqvF\svb_1\quad\dots\quad\sva_{\arty{\PrL}}\eqvF\svb_{\arty{\PrL}}}{\PrL(\sva_1,\dots,\sva_{\arty{\PrL}})\to\PrL(\svb_1,\dots,\svb_{\arty{\PrL}})}
  \end{equation*}
  for every modal operator $\PrL\in\Lambda$.
  The reason is that the congruence rule is clearly one-step sound,
  and even under the much weaker asssumption of \emph{one-step
    completeness} (obtained by restricting $\Xi$
  to be finite in Definition~\ref{def:sosc}), all one-step sound rules
  are derivable~\cite{SchroderVenema17}.
\end{rem}
\renewcommand{\modimpl}{\To}

\begin{exa}\label{ex:s1sc}\
\begin{itemize}
\item Modal neighbourhood semantics is axiomatised by
\begin{equation*} 
\nrule{\rul{C}}{\sva \eqvF \svb}{\Box \sva \to \Box \svb}
\end{equation*}
which expresses that $\Box$ is a congruential operator (where we write $\Box\sva\to\Box\svb$ in place of $\neg\Box\sva\orF\Box\svb$ for readability). This is the paradigmatic example of S1SC; see the discussion in Remark \ref{rem:s1sc} below. 
\item The rule set for the normal modal logic $\mathsf{K}$  consists of the rules
\begin{equation*} 
\nrule{\rul{K_n}}{\sva \to \svb_1\orF \dots\orF \svb_n}{\Diamond \sva \to \Diamond \svb_1\orF \dots\orF
\Diamond \svb_n}
\end{equation*}
for all $n \geq 0$. As we are going to see in Lemma \ref{lem:kr1sc}, $\Rules$ is finitary S1SC. On the other hand, for reasons detailed in Remark \ref{rem:s1sc}, $\Rules$ is not S1SC.  An  S1SC semantics for this particular rule set would be provided by \emph{normal neighbourhood \bro i.e., filter\brc\ frames}.
\item For graded modal logic, the proof of
  one-step completeness of the rule set  
        \begin{gather*}
     \nrule{\rul{RG1}}{\sva\to \svb}{\gldiamond{n+1}\sva\to\gldiamond{n}\svb}
    \qquad
    \nrule{\rul{A1}}
	  {\svc\to \sva\orF \svb}
	  {\gldiamond{n_1+n_2}\svc\to\gldiamond{n_1}\sva\orF\gldiamond{n_2}\svb}\\[2mm]
    \nrule{\rul{A2}}
	  {	\sva \to \svc \quad \svb\to \svc \quad
	\sva\andF \svb\to \svd}
	  {\gldiamond{n_1}\sva\orF\gldiamond{n_2}\svb\to
	    \gldiamond{n_1+n_2+1}\svc\orF\gldiamond{0}\svd}
	  \qquad
    \nrule{\rul{RN}}{\neg\sva}{\neg\Diamond_0 \sva} 
  \end{gather*}
  
  \cite[Lemma~3.10]{SchroderVenema17} upon inspection effectively
  establishes finitary S1SC. 
  \item For positive Presburger modal logic, we
  similarly have that the proof of one-step completeness of a natural
  rule system for full Presburger modal
  logic~\cite[Lemma~4.5]{KupkePattinson10} straightforwardly adapts to
  a) positive Presburger modal logic, and b) finitary S1SC, provided
  that one generalizes the semantics to infinite multisets as we do
  here.
    \item Conditional logic provides a curious mixed case, due to its being neighbourhood-like in one coordinate and Kripke-like in another. We will introduce  adequate apparatus and analyse it further in Example \ref{ex:cs1sc}.
  \item In Remark \ref{rem:pauly}, we discuss the issue of S1SC of  coalition logic interpreted over effectivity functions. 
\end{itemize}
\end{exa}

\begin{lem} \label{lem:kr1sc} The $\Lambda$-structure for relational
  modal logic \bro Section~\ref{sec:ss}\brc\ is finitary S1SC.
\end{lem}

\begin{proof}
Consider any consistent $\Xi \subseteq \BSch(\Lambda(\Pfin X))$, where $\Lambda = \{\Diamond\}$; w.l.o.g.,  $\Xi$ is maximally consistent. We need to find $A \in \Pow X$ s.t. $A \vDash_{\Pow X} \Xi$. Let us choose 
\[
A = \{ x \mid \Diamond\{x\} \in \Xi \}. 
\]
One shows by induction on formulas that 
\begin{center}
$A \vDash_{\Pow X}  \phi$ iff $\phi \in \Xi$.
\end{center}
The only non-trivial case is the modal one. For any $B = \{b_1, \dots, b_n\} \in \Pfin X$, we have that
\begin{align*}
A \vDash_{\Pow X} \Diamond B & \text{ iff } \exists i \leq n.A \vDash_{\Pow X}\Diamond\{b_i\}  \\
&  \text{ iff } \exists i \leq n. \Diamond\{b_i\} \in \Xi \\
&  \text{ iff } \Diamond B = \Diamond\{b_1\} \cup \dots \cup   \Diamond\{b_n\} \in \Xi.
\tag*{\qEd}
\end{align*}
\def\popQED{} 
\end{proof}

\begin{rem}\label{rem:s1sc}
  As noted in~\cite[Remark~55]{SchroderPattinson10b}, we can give a
  more abstract characterization of S1SC, recognizable also to readers
  familiar with more categorical presentations of coalgebraic modal
  logic (cf., e.g., \cite{KurzRosicky}). 
  Every signature $\Lambda$
  together with a given set of one-step axiom schemes (equivalently,
  one-step rules) can be encoded disregarding concrete syntax by its
  \emph{functorial presentation}~\cite{KupkeEA04} (cf.  also
  \cite[Definition 28]{SchroderPattinson10b}) as an endofunctor
  $L_\Lambda$
  on the category $\mathsf{BA}$
  of Boolean algebras. $\mathsf{BA}$
  is dually adjoint to $\Set$,
  with the adjunction given by the contravariant powerset
  functor\footnote{We write here $\overline{\CPow}$
    to stress that we change the target category.}
  $\overline{\CPow}$
  and the functor $\Spa$
  taking a Boolean algebra to the set of its ultrafilters:
  \begin{equation}\label{equ:LPT}
    \vcenter{
      \xymatrix@C=15pt{ 
        {\mathsf{BA}} \POS!R(-.5) \ar@(dl,ul)[]^{L_\Lambda} 
        \ar@/_/[rr]_{\Spa}  
        & &  
        {\Set}\ar@/_/[ll]_{\overline{\CPow}} \POS!R(.5)\ar@(dr,ur)[]_{T}  
      }
    }
  \end{equation} 
  The information contained in each $\Lambda$-structure
  can be then more abstractly encoded by $\delta:
  L_\Lambda\overline{\CPow} \to
  \overline{\CPow}T$~\cite{KupkeEA04} and the \emph{canonical}
  structure for $\Lambda$
  is given by $M_\Lambda
  = \Spa
  L_\Lambda\overline{\CPow}$. Coalgebras for the canonical structure
  can equivalently be described as generalized neighbourhood frames
  (where by \emph{generalized} we mean that for every $n$-ary
  modality in $\Lambda$,
  we have $n$-ary
  neighbourhoods, i.e.\ subsets of the $n$-th
  power of the state set), subject to satisfaction of the frame
  conditions embodied in the given
  one-step rules~\cite[Remark~34]{SchroderPattinson10b}. For every
  $\Lambda$-structure,
  we can define a canonical \emph{structure morphism}
  \cite[p. 1121]{SchroderPattinson10b} to $M_\Lambda$
  by composing the counit of the above adjunction with
  $\Spa\delta$,
  and S1SC effectively requires that this structure morphism is
  surjective. In other words, a $\Lambda$-structure
  is S1SC iff its functor surjects onto the canonical neighbourhood
  semantics; it is for this reason that we refer to the S1SC case as
  ``neighbourhood-like''. In fact, as we explain in the next
  remark, we do not currently have an in-the-wild example of an S1SC
  structure that is not actually isomorphic to the canonical
  neighbourhood semantics. 

\end{rem}

\begin{rem}\label{rem:pauly}
  \newcommand{\Gam}{\mathcal{G}} Coalition logic \cite{Pauly02} and,
  essentially equivalently, the next-step fragment of alternating-time
  temporal logic \cite{AlurEA02}, have modalities $[Q]$ indexed over
  \emph{coalitions} $Q$, which are subsets of a fixed finite set $N$
  of \emph{agents}; the operator $[Q]$ reads `the coalition $Q$ of
  players can enforce \dots in the next step'. The semantics is
  formulated over structures called \emph{game frames} or
  \emph{concurrent game structures}, i.e.,
  coalgebras for the functor
  \begin{equation*}\textstyle
    \Gam X = \lbrace ((S_i)_{i \in N}, f: \big(\prod_{i \in N} S_i\big) \to X)
    \mid \emptyset \neq S_i \subseteq \Nat
    \rbrace
  \end{equation*}
  where $S_i$ is thought of as the set of moves available to agent
  $i \in N$ and $f$ is an \emph{outcome function} that determines the
  next state of the game, depending on the moves chosen by the agents
  (we restrict to finitely many moves per agent as in alternating-time
  temporal logic). For notational convenience, given a coalition
  $Q = \{q_1, \dots, q_k\} \subseteq N$ and moves
  $s_{q_1} \in S_{q_1}$, \dots, $s_{q_k} \in S_{q_k}$, we write
  $s_Q \deq (s_q)_{q \in Q}$ and
  $S_Q=S_{q_1}\times\dots\times S_{q_k}$ (so that $s_Q\in S_Q$). Given
  $s_Q\in S_Q$ and $s_{N\setminus Q}\in S_{N\setminus Q}$, we write
  $(s_Q,s_{N\setminus Q})$ for the evident induced element of $S_N$.

  An alternative semantics of the coalitional operators is provided by
  \emph{effectivity functions}. These are functions $E$ assigning to
  each coalition $Q$ a set $E(Q)$ of properties that $Q$ can
  enforce. Explicitly, a concurrent game
  $G=((S_i)_{i \in P}, f)\in\Gam X$ induces an effectivity function
  $E_G$ by
  \begin{equation*}
    E_G(Q)=\{A\subseteq X\mid \exists s_Q\in S_Q.\,\forall s_{N\setminus Q}\in S_{N\setminus Q}.
    \, f(s_Q,s_{N\setminus Q})\in A\}.
  \end{equation*}
  Effectivity functions congregate into a functor $\CE$, a subfunctor
  of a product of neighbourhood functors. The modal operators $[Q]$
  are interpreted over effectivity functions in the usual style of
  neighbourhood semantics, i.e.\ by
  \begin{equation*}
    \Sem{[Q]}_X(A)=\{E\in\CE X\mid A\in E(Q)\}.
  \end{equation*}
  Composing this semantics with the above-defined projection from
  concurrent games to effectivity functions yields the interpretation
  of the coalitional modalities $[Q]$ over $\Gam$; this reproduces the
  standard semantics of coalition logic and alternating time temporal
  logic.  

  Now Theorem~3.2 in~\cite{Pauly02} states that an effectivity
  function $E\in\CE(X)$ is of the form $E_G$ for some $G \in\Gam X$ iff it is
  \emph{playable}, i.e.\ satisfies the following properties:
  \begin{itemize}
  \item For all $Q$, $\emptyset\notin E(Q)\owns X$
  \item $E$ is \emph{outcome-mononotic}, i.e.\ each $E(Q)$ is upwards
    closed under set inclusion.
  \item $E$ is \emph{$N$-maximal}, i.e.\ for all $A\subseteq X$,
    either $X\setminus A\in E(\emptyset)$ or $A\in E(N)$.
  \item $E$ is \emph{superadditive}, i.e.\ whenever $A_1\in E(Q_1)$
    and $A_2\in E(Q_2)$ for disjoint coalitions $Q_1,Q_2$, then
    $A_1\cap A_2\in E(Q_1\cup Q_2)$. 
  \end{itemize}
  If this were the case, then coalition logic interpreted over either
  concurrent games or playable effectivity functions would be S1SC,
  and we claimed as much in the conference
  version~\cite{LitakPSS12:icalp}: The above conditions amount to
  playable effectivity functions being just neighbourhood systems that
  satisfy a set of one-step rules:
  \begin{gather*}
    \neg[Q]\bot\qquad [Q]\top\qquad \frac{\sva\to\svb}{[Q] \sva\to [Q]\svb}\\[1ex]
    \frac{\sva\orF\svb}{[\emptyset]\sva\lor[N]\svb}\qquad
    \frac{\sva\land\svb\to\svc}{[Q_1]\sva\land[Q_2]\svb\to[Q_1\cup Q_2]\svc}\;\text{ for
    } Q_1\cap Q_2=\emptyset,
  \end{gather*}
  and such structures are always S1SC~(this
  is by the development in \cite{SchroderPattinson10b}, see in
  particular Remarks~34 and~55 in op.\ cit., and also
  Remark~\ref{rem:s1sc} above). However, it turns out that Theorem~3.2
  in~\cite{Pauly02} is not in fact entirely correct, and once fixed no
  longer implies that coalition logic is S1SC. To see this, note that
  for every effectivity function of the form $E_G$, $E_G(\emptyset)$
  must have a least element, equivalently be closed under
  intersections: every element of $E_G(\emptyset)$ must contain the
  set
  \begin{equation*}
    A=\{f(s_N)\mid s_N\in S_N\},
  \end{equation*}
  and this set is itself in $E_G(\emptyset)$. This condition is
  however not satisfied by all playable effectivity functions in the
  above sense: take $X$ to be some infinite set, pick a non-principal
  ultrafilter $U$ on $X$, and put $E(Q)=U$ for all coalitions
  $Q$. This defines a playable effectivity function but
  $E(\emptyset)=U$ has no least element. Adding the condition that
  $E(\emptyset)$ has a least element to the definition of playability
  does fix the theorem, but this condition is not expressible by a
  finitary one-step axiom and hence we do not obtain S1SC for 
  coalition logic as a corollary.
\end{rem}


\noindent As indicated above, we have alternative conditions that
ensure completeness~\cite{SchroderPattinson10}:
\begin{defi}\label{def:bounded}\
 A modal operator $\PrL$ is \emph{$k$-bounded} in the $i$-th argument
  for $k \in \Nat$ and with respect to a $\Lambda$-structure $T$ if
  for every $C \in \Set$ and every $\vec{A} \subseteq C$,
\begin{equation}\label{eq:bounded}
    \Sem{\PrL}_C(A_1,\dots,A_n) = 
\bigcup_{B\subseteq A_i, \#B\le k}\Sem{\PrL}_C(A_1,\dots,A_{i-1},B,A_{i+1},\dots,A_n).
\end{equation}
\end{defi}
\noindent Boundedness of $\PrL$ in the $i$-th argument implies in
particular that $\PrL$ is monotonic in the $i$-th argument.  
 We can replace the
assumption that the rule set $\Rules$ is S1SC with the weaker
assumption that $\Rules$ is finitary S1SC, provided that  modal
operators are bounded on respective coordinates. The technical details of
a suitably general setup are as follows:

\takeout{\begin{exa}
  Let us examine some examples of bounded structures. Under
  monotonicity, the right-to-left inclusion in~\eqref{eq:bounded} is
  automatic, so we only discuss the left-to right inclusion:
  \begin{itemize}
  \item    The
  standard Kripke modality $\Diamond$ is $1$-bounded: If
  $B\in\Sem{\Diamond}_C(A)$, then $B\cap A\neq\emptyset$ by
  definition. So fix $c\in B\cap A$; then
  $B\in\Sem{\Diamond}_C(\{c\})$. 
 \item Similarly, the non-monotonic
  conditional $>$ is $1$-bounded in its second argument. 
  \item The graded
  modality $\langle k\rangle$ is $k+1$-bounded: If
  $\mu\in\Sem{\langle k\rangle}_C(A)$, then $\mu(A)>k$. Hence there
  exist (not necessarily distinct) elements $c_1,\dots,c_{k+1}\in C$
  such that $\mu(\{c_1,\dots,c_{k+1}\})>k$, i.e.\
  $\mu\in\Sem{\langle k\rangle}_C(\{c_1,\dots,c_{k+1}\})$. 
  \item Similarly,
  in positive Presburger modal logic, the operator
  $L_k(a_1,\dots,a_n)$ is bounded in all its arguments, more
  precisely, $((k+1)\text{ div }a_i)$-bounded in the $i$-th argument.
  \end{itemize}

  \end{exa}}
\begin{defi}\
  \begin{itemize}
\item A \emph{boundedness signature} for $\Lambda$ is a function $\bel$ assigning to every $\PrL \in \Lambda$ a vector of elements of $\Nat \cup \{\infty\}$ of length $\arty{\PrL}$, i.e. an element of $(\Nat \cup \{\infty\})^{\arty{\PrL}}$. 
\item Being \emph{$\infty$-bounded} is a condition trivially satisfied by all operators, i.e., every operator is ``$\infty$-bounded'' in each coordinate.
\item  We say that $\bel$ is \emph{adequate}  for  a $\Lambda$-structure
 over  $T$  if every modal operator $\PrL \in \Lambda$ is $\bel(\PrL)(i)$-bounded in $i$ every $i \leq \arty{\PrL}$. 
\item We say that
  $\Lambda$ is $\bel$-\emph{bounded} w.r.t $T$ if $\bel$ is adequate for  the structure in question and the codomain of $\bel$ does not contain $\infty$. 
  \item We say that
  $\Lambda$ is \emph{bounded} w.r.t. $T$ if it is $\bel$-\emph{bounded} w.r.t. $T$ for \emph{some} $\bel$. That is, 
  every modal operator $\PrL \in \Lambda$ for every $i \leq \arty{\PrL}$  is
  $k_{\PrL,i}$-bounded in $i$ for some $k_{\PrL,i} < \infty$.  
\end{itemize}
\end{defi}


\begin{exa}\ \label{ex:bounded}
Here are some examples of boundedness signatures adequate for structures under consideration:
\begin{itemize}
\item for the neighbourhood case,  $\bel(\Box) = (\infty)$,
\item for the Kripke case,  $\bel(\Diamond) = (1)$: \newline
if
  $B\in\Sem{\Diamond}_C(A)$, then $B\cap A\neq\emptyset$ by
  definition. So fix $c\in B\cap A$; then
  $B\in\Sem{\Diamond}_C(\{c\})$,
\item for  graded modalities, $\bel(\langle k \rangle) = (k+1)$: \newline
 if
  $\mu\in\Sem{\langle k\rangle}_C(A)$, then $\mu(A)>k$. Hence there
  exist (not necessarily distinct) elements $c_1,\dots,c_{k+1}\in C$
  such that $\mu(\{c_1,\dots,c_{k+1}\})>k$, i.e.\
  $\mu\in\Sem{\langle k\rangle}_C(\{c_1,\dots,c_{k+1}\})$, 
\item for positive Presburger logic, 
$$\bel(L_k(a_1,\dots,a_n)) = ((k + 1) \;\mathrm{div}\; a_1 + 1, \dots, (k + 1) \;\mathrm{div}\; a_n + 1),$$
\item for the discrete distribution functor $\Dist$, $\bel(\langle p \rangle) = (\infty)$,
\item for its finite variant $\Dist_k$, $\bel(\langle n/k \rangle) = (n)$,
\item for non-monotonic conditionals,  $\bel(\CToDual) = (\infty,1)$.
\end{itemize}
\end{exa}

\noindent Note that, e.g., the neighbourhood modality clearly
fails to be bounded. 
Boundedness allows us to broaden the scope of our completeness results to setups where full S1SC would be too much to ask, i.e., to leave the  neighbourhood-like setting.  This is done by requiring S1SC  on suitable coordinates only for  valuations of schematic variables in \emph{finite} sets. Modalities which are both finitary S1SC and bounded will be called \emph{Kripke-like}. In order to make this precise so that we can cover mixed cases, such as those of non-monotonic conditionals, some care is needed.

\begin{defi}\label{def:finitarysc}\
\begin{itemize}
\item The \emph{colouring} function $\bbe: \Nat \cup \{\infty\} \to \{\fin, \infty\}$ assigns $\fin$ to elements of $\Nat$ and   $\bbe(\infty) = \infty$. It is extended pointwise to $(\Nat \cup \{\infty\})^{\arty{\PrL}}$ and $\bbel$ is defined as the composition of $\bel$ with this pointwise extension. 
\item Let $\colo: \SchV \to \{\fin, \infty\}$ be a colouring of the set of schematic variables. Define the set of \emph{$\bbel,\colo$-coloured
    modalities} as
  $$\bbel_\colo  = \{ \PrL \sva_1\dots \sva_{\arty{\Lambda}} \,|\, \sva_1\dots \sva_{\arty{\Lambda}} \in \SchV \text{ and }  (\colo\sva_1,\dots,\colo\sva_{\arty{\Lambda}}) = \bbel(\PrL) \}.$$
\item A valuation $\tau: \SchV \to \Pow(C)$ respects \colo\ iff
  $\tau(\sva_i) \in \Pow_{\colo\sva_i}$, where we recall that
  $\Pow_\fin$ is finite powerset, and $\Pow_\infty$ is simply $\Pow$.
\item  A Gentzen-style one-step rule $R$ is  \emph{$\bbel,\colo$-compatible} if it is of the form 
 
\begin{equation*}
  \frac{\Gamma_1 \To \Delta_1 \qquad \cdots \qquad \Gamma_k \To
	\Delta_k}{\Gamma_{R} \To \Delta_{R}}
	\hspace*{\fill}(R)
\end{equation*}
where 
\begin{itemize}
\item $\Gamma_1, \dots, \Gamma_k, \Delta_1, \dots, \Delta_k$  are multisets of elements of $\SchV$,
\item $\Gamma_{R}$ and $\Delta_{R}$ are multisets of elements of $\bbel_\colo$. 
\end{itemize}
\item For a Gentzen-style rule, write
$\bbel_\colo(R)$ for the set of  $\bbel,\colo$-compatible variants of  $R$ obtained by renaming of schematic variables. For a Hilbert-style rule, $\bbel_\colo(R)$ is obtained via its Gentzen-style counterpart. 
Finally, $\bbel_\colo(\Rules) = \{ \bbel_\colo(R) \mid R \in \Rules\}$.
\item A set $\Xi \subseteq \RSch(\SchV)$ is \emph{$\colo$-consistent wrt $\tau$} if its union with the set 
\begin{equation*}
\lbrace \svX
\sigma \mid\sigma:\SchV \to\BSch(\SchV) \text{ and } \svA/ \svX \in
\bbel_\colo(\Rules) \text{ s.t. } C, \tau \models \svA \sigma\rbrace
\end{equation*}
 is
propositionally consistent.

\item We say that a set of rules $\Rules$ is \emph{$\bel$-S1SC} if  
\begin{center}
 for every $C \in \Set$, any $\Xi \subseteq \RSch(\SchV)$, any colouring $\colo$ and any $\tau$ respecting $\colo$, \\ $\Xi$ is one-step satisfiable wrt $\tau$ whenever it is $\colo$-consistent wrt $\tau$. 
 \end{center}
\end{itemize}
\end{defi}
\noindent That is, the notion of $\bel$-S1SC instantiates to finitary
S1SC in bounded arguments, and to S1SC in unbounded ones. 
\begin{exa} \label{ex:cs1sc}
In the
case of the modal signature $\{\CToDual\}$ of non-monotonic
conditionals with $\bel(\CToDual) = (\infty,1)$, the associated
one-step completeness condition may be called \emph{\bro S1SC, finitary
  S1SC\brc}~\cite{SchroderPattinson10}. A suitable axiomatization can be extracted from existing references \cite{Chellas80,SchroderPattinson10b,PattinsonSchroder10,SchroderPattinson10}:
  \takeout{\begin{itemize}
  \item $\neg(\sva > \bot)$
  \item $(\sva > (\svb \wedge \svc)) \vee ((\sva > \neg\svb) \wedge (\sva > \svc))$
  \end{itemize}}
  
  \[
  \nrule{\rul{RCK}}{\sva \to \svb_1\orF \dots\orF \svb_n}{\svc > \sva \to \svc > \svb_1\orF \dots\orF \svc > \svb_n} \qquad \nrule{\rul{RE}}{\sva \eqvF \svb}{\sva > \svc \to \svb > \svc}
  \]
  A proof that this axiomatization is indeed $\bel$-S1SC  combines the neighbourhood argument in the first coordinate (cf.\ Remark~\ref{rem:s1sc}) with the Kripke argument in the second (cf.\ Lemma~\ref{lem:kr1sc}).
  \end{exa}




\subsection{Hilbert-style Calculus} \label{sec:hilb-axiom}


\def\tbskip{0.7mm}
 
\begin{table*}[t]
\hrule
\vspace{\tbskip}

\caption{\label{tab:folaxioms}\label{tab:fo-axioms}\strut  Hilbert-style Calculus  $\Hilb\Rules$} 

\small

The axioms are modelled after those of Enderton \cite{Enderton72}.

\vspace{\tbskip}

Everywhere below, $\yprf$ denotes a sequence of universal quantifiers of arbitrary length, possibly empty.

\setlist[description]{font=\normalfont,}

\begin{description}
\item[\axref{En1}] all propositional tautologies. These can be axiomatized, e.g., by

\medskip

\begin{itemize}
\item $\yprf\lb\phi \impF \lb\psi\impF\phi\rb\rb$
\item $\yprf\lb\lb\phi \impF \lb\psi \impF \chi\rb\rb \impF \lb\lb\phi \impF \psi\rb \impF \lb\phi \impF \chi\rb\rb\rb$
\item $\yprf\lb\bot \impF\phi\rb$
\item $\yprf\lb\lb\lb \phi \impF \bot\rb \impF \bot\rb \impF \phi\rb$
\end{itemize}

\medskip

\noindent For \axref{En2} and \axref{Alpha} below, recall that whenever we write a substitution we implicitly impose the assumption that the substituted term is actually substitutable.
 
\medskip

\item[\axref{En2}] $\forall \vec{y}. (\forall x. \phi \to \phi[z/x])$
\item[\axref{En3}] $\forall \vec{y}. (\forall x. (\phi \to \psi) \to (\forall
x. \phi \to \forall x. \psi))$
\item[\axref{En4}] $\forall \vec{y}. (\phi \to \forall x. \phi)$ if $x$ is
fresh for $\phi$
\item[\axref{En5}] $\forall \vec{y}.(x = x)$
\item [\axref{En6}.1] $\forall \vec{y}. (x = z \to P(\vec{u}, x, \vec{v}) \to
P(\vec{u}, z, \vec{v}))$ for $P \in \Sigma \cup \lbrace = \rbrace$
\item[\axref{En6}.2] $\forall \vec{y}. (x = z \to x \hearts \lcomp y_1 \col
\phi_1 \rcomp \dots \lcomp y_n \col \phi_n \rcomp \to  z
\hearts \lcomp y_1 \col
\phi_1 \rcomp \dots \lcomp y_n \col \phi_n \rcomp)$
\item[\axref{Alpha}] $\forall \vec{y} ((x \PrL  \dots \lcomp z : \phi \rcomp \dots) \to (x \PrL  \dots \lcomp u : \phi[u/z] \rcomp \dots)) $
\item[\axref{Onestep($\Rules$)}] $\forall \vec{y}. \forall z. 
  (\forall x. \svA\sigma \to [\sigma, x, z]\svX)$ where
\begin{itemize}
\item $\svA/\svX$ ranges over the one-step rules in $\Rules$, and 
\item $\sigma$ is a substitution sending each $\schv{p}_{i}$ to a formula of $\Lang$, and $[\sigma, x, z]$ is 
the inductive extension of the map sending each $\hearts_i \vec{\schv{p}_i}$ to 
$z \hearts_i \lcomp x : \sigma(\schv{p}_{i}^{1}) \rcomp \cdots \lcomp x : \sigma(\schv{p}_{i}^{a(i)}) \rcomp)$
\end{itemize}
\end{description}


\bigskip

\begin{description}
\item[\bdpl] An additional axiom scheme when $\bel(\PrL)(i) \neq \infty$, with $\vec{z}\ntoc y_i,\vec{\phi}$ 
\[
\yprf (x\PrL\dots\cmh{y_i}{\phi_i} \dots
\eqvF
\exists z_1\dots z_{\bel(i)}.(\bigwedge\limits_{j \leq \bel(\PrL)(i)}\sbst{\phi_i}{y_i}{z_j} \wedge
 x\PrL \dots \cmh{y_i}{\bigvee\limits_{j \leq \bel(\PrL)(i)}y_i \eqF z_j}\dots )) 
 \] 
\end{description}



\vspace{1mm}
\hrule
\end{table*}%


We are finally ready to present our axioms for CPL in
Table~\ref{tab:folaxioms}. Axioms \axref{En1}--\axref{En6} are just
those of Enderton, with \axref{En6}.2 an additional clause to cover
the case of modal formulas. The $\alpha$-renaming axiom \axref{Alpha}
is needed because our syntax features separate comprehension
variables. 
Our \axref{Onestep($\Rules$)} axiom scheme generalizes what was
originally just the congruence rule (Remark~\ref{rem:cong}) in Chang's
formalism, corresponding to the fact that the S1SC rule system for
neighbourhood semantics consists of only the congruence
rule. Axiom \axref{\bdpl} applies to operators that are bounded in
suitable coordinates. It is important to notice that boundedness is
not expressible as a \emph{sentence} or \emph{formula} in weak
frameworks; in languages like $\HyFA$, it can only be expressed by a
non-standard rule \cite{SchroderPattinson10}.
 

Let $\Gamma,\Delta \subseteq \CPL(\PrLSet,\SoV)$, let $\Rules$ be a set of one-step rules and $\phi \in \CPL(\PrLSet,\SoV)$. Write 
$\Gamma \gdash \phi$ if  
 there are $\gamma_1, \dots, \gamma_n \in \Gamma$ s.t. $\gamma_1 \impF \dots \impF \gamma_n \impF \phi$
 can be deduced from 
 \axref{En1}--\axref{En6}, \axref{Alpha}, \axref{Onestep($\Rules$)} and \axref{\bdpl} in Table \ref{tab:folaxioms} using \textbf{only Modus Ponens}. 
  This clearly defines \emph{a finitary deducibility relation} in the sense of Goldblatt \cite[Sec. 8.1]{goldblatt1993:abstract} and being \emph{$\gdash$-consistent} is equivalent with being \emph{finitely $\gdash$-consistent} in his sense, that is, $\Gamma \gdash \bot$ iff there is $\Gamma_0 \finin \Gamma$ s.t. $\Gamma_0 \gdash \bot$. 

\begin{thm}[Soundness] \label{thm:hilb-soundness}
Whenever a $\Lambda$-structure
 over  $T$ is adequate for $\bel$, 
 all the axioms in Table \ref{tab:folaxioms} hold in every coalgebraic $\Lambda$-model  and the set of formulas valid in such a model is closed under $\gdash$.
\end{thm}

\noindent Recall that by Convention \ref{conv:1ss}, one-step soundness of $\Rules$ is not mentioned explicitly.





\begin{defi}
For 
 any $\Lambda$, $\Rules$ and $\bel$, we say that the inference system given by $\gdash$ is \emph{strongly complete} 
  for a given $\Lambda$-structure based on $T$ if for any set of sentences $\Gamma \in \CPL(\PrLSet,\SoV)$, $\Gamma \not\gdash \bot$   
holds \textbf{if and only if} there is a coalgebraic $\Lambda$-model for $\Gamma$.
\end{defi}





\begin{thm}[Strong
  Completeness] \label{th:completeness}\label{thm:hilb-complete}
  Whenever a set of rules $\Rules$ is \emph{$\bel$-S1SC} for a
  $\Lambda$-structure over $T$ that is adequate for $\bel$, then
  $\gdash$ is strongly complete for this structure.
 \end{thm}
 


\begin{exa}\label{expl:completeness}
For concrete instances of this completeness result, combine Examples \ref{ex:s1sc}, \ref{ex:bounded} and \ref{ex:cs1sc}.
\takeout{For the examples presented in \S~\ref{sec:ss}, the situation is
as follows. Completeness holds for neighbourhood models as they have
a strongly one-step complete axiomatisation. For all others, but
excluding non-monotonic conditionals, finitary one-step complete
axiomatisations exist, cf. Example \ref{ex:s1sc}. \tlnew{This is either obsolete or requires updating. Are we consistently using finite probabilites, for example?} As discussed above (cf. Example \ref{ex:bounded}), boundedness holds for relational models,
graded modal logic and the logic of finite probabilities
(interpreted over $\mathcal{D}_k$-coalgebras) whereas conditional logic is covered as a mixed case. }
\end{exa}

\subsection{Proof of The Completeness Theorem} 

\label{sec:proofco}

First, we introduce  machinery proposed in \cite{goldblatt1993:abstract}. Consider any $Fr \subseteq \CPL(\PrLSet,\SoV)$ closed under propositional connectives. $Fr$ can be, for example, the set of all formulas whose free variables are contained in a fixed subset of $\FoV$, the set of all sentences and the entire $\CPL(\PrLSet,\SoV)$ itself being the two borderline cases.  Any set $\emph{Inf} \subseteq \Pow(Fr) \times Fr$ will be called, following Goldblatt, \emph{a set of inferences}. For any $\mathit{inf} \deq (\Pi,\chi) \in \mathit{Inf}$ and any $\Gamma \subseteq Fr$, we say that 
\begin{itemize}
\item $\Gamma$ \emph{respects} $\mathit{inf}$ if $\Gamma \gdash \chi$ whenever $\Gamma \gdash \phi$ for all $\phi \in \Pi$,
\item $\Gamma$ \emph{is closed under} $\mathit{inf}$ if $\chi \in \Gamma$ whenever $\Pi \subseteq \Gamma$,
\item  $\Gamma$ respects  $\mathit{Inf}$ iff it respects each member of $\mathit{Inf}$,
\item $\Gamma$ is closed under $\mathit{Inf}$ iff it is closed under each member of $\mathit{Inf}$.
\end{itemize}

\begin{thm}[Goldblatt's Abstract Henkin Principle \cite{goldblatt1993:abstract}] \label{th:goldhenkin}
If $\mathit{Inf}$ is a set of inferences in $Fr$ of an infinite cardinality $\kappa$ and $\Gamma$ is a $\gdash$-consistent subset of $Fr$ satisfying in addition:
\begin{equation} \label{eq:finrespect}
\forall X \subseteq Fr. \card{X} < \kappa \text{ implies that } \Gamma \cup X \text{ respects } \mathit{Inf} 
\end{equation}
(i.e., \emph{every $\kappa$-finite extension of $\Gamma$ respects $\mathit{Inf}$}), then $\Gamma$ has a maximally $\gdash$-consistent extension in $Fr$ which is closed under $\mathit{Inf}$.
\end{thm}


\begin{rem} \label{rem:noinfinitary} We emphasize that speaking about
  \emph{inferences} in Goldblatt's sense being infinite sets does
  \emph{not} mean that \emph{deductions} in the axiom system for for
  CPL use infinitary rules.  As stated above, the only inference rule
  in our system is ordinary Modus Ponens. Even the one-step rules
  defined above (which are not infinitary anyway) can be written as
  sentence schemes \axref{Onestep} thanks to the use of quantifiers.
  
  We further point out that an Enderton-style axiomatization does not
  involve the generalization rule: if $x$ is a free variable in
  $\phi$, it is not necessarily the case that
  $\phi \gdash \forall x.\phi$ (this is not in contradiction to
  completeness: the rule is sound in the sense that validity of the
  premise implies validity of the conclusion, but its conclusion is
  not a logical consequence of its premise). This makes it enjoy a
  rather rare property for an axiomatization of FOL: a deduction
  theorem in exactly the same form as propositional logic, i.e.,
  $\Gamma \cup \{\phi\} \gdash \psi$ iff
  $\Gamma \gdash \phi \impF \psi$
  (cf. \cite[p. 118]{Enderton72}). This will also allow us to give our
  Henkin-style proofs \emph{without introducing additional
    constants}---the role of Henkin constants for existentially
  quantified variables will be played by the variables themselves.\footnote{Recall that we have been working in a setup without functions symbols (including $0$-ary ones) anyway; extending our original syntax with constants just for the sake of this particular proof does not even seem particularly hygienic.} The only disadvantage of this approach would be that if we consider uncountable $\PrLSet$ or $\SoV$, we would also need to allow uncountably many elements of $\FoV$, something we highlight in the statement of several lemmas and claims below.
\end{rem}

\noindent Let us  recall the crucial ingredient in Henkin-style completeness proofs: the notion of quasi-Henkin model and its associated Truth Lemma. This is inspired by previously announced completeness proofs for coalgebraic hybrid logic  \cite{SchroderPattinson10}; we discuss the relationship in detail in Remark \ref{rem:hyfa} and \S~\ref{sec:cml} below.

\begin{defi} \label{def:quasih}
Let $\Gamma$ be a maximal consistent set (MCS) of formulas. Define $C_{\Gamma}$ = $\lbrace |x| : x \text{ is a variable} \rbrace$, where $|x|$ = $\lbrace z : x \eqF z \in \Gamma \rbrace$, and put $I_{\Gamma}(P)$ \deq $\lbrace (|x_{1}|,\dots,|x_{n}|) : P(x_{1},\dots,x_{n}) \in \Gamma \rbrace$. Set $\widehat{\phi}^{y_{i}}$ \deq $\lbrace |z| : \phi[z/y_{i}] \in \Gamma \rbrace$, to be thought of as \emph{the set of variables satisfying $\phi$ according to $\Gamma$} (when $y_i$ is taken to be the \emph{argument variable} or the \emph{context hole}). Given a $T$-coalgebra structure $\gamma: C_{\Gamma} \to TC_{\Gamma}$, we say that $(C_{\Gamma},\gamma,I_{\Gamma})$ is a {\em quasi-Henkin coalgebraic model} if, for any variables $x$, $y_{1}$, \dots, $y_{n}$ and any formulas $\psi$, $\phi_{1}$, \dots, $\phi_{n}$, 
\begin{equation} \label{eq:henkinqu}
\exists x.\psi \in \Gamma \Longrightarrow \text{ for some } y_i, y_i \in \widehat{\psi}^{x}.
\end{equation}
(note that the converse implication holds for any MCS) and 
\begin{equation} \label{eq:henkinmo}
x \heartsuit \lceil y_{1} : \phi_{1} \rceil \cdots \lceil y_{n} : \phi_{n} \rceil \in \Gamma \Longleftrightarrow \gamma(|x|) \in  \lsem \heartsuit \rsem_{C_{\Gamma}}(\widehat{\phi_{1}}^{y_{1}},\dots,\widehat{\phi_{n}}^{y_{n}}).
\end{equation}
\end{defi}

In a quasi-Henkin model, define the \emph{canonical variable assignment}  $v_{\Gamma}$ by $v_{\Gamma}(x)$ = $|x|$.

\begin{lem}[Truth Lemma]
Let $\Gamma$ be a maximal consistent set of formulas and $\mathfrak{M}_{\Gamma}$ = $(C_{\Gamma},\gamma,I_{\Gamma})$ a quasi-Henkin coalgebraic model.  Then, for every formula $\phi$, 
\begin{equation}\label{eq:truth}
\mathfrak{M}_{\Gamma},v_{\Gamma} \models \phi 
\Longleftrightarrow 
\phi \in \Gamma.
\end{equation}
\end{lem}

\begin{proof}
By induction on $\phi$. An auxiliary fact we need is that whenever
$\phi$ satisfies the inductive claim~\eqref{eq:truth}, then
\begin{equation} \label{eq:semfa}
\lsem \phi \rsem^{y}_{\gM_\Gamma} = \widehat{\phi}^{y}
\end{equation}
(recall $\lsem \phi \rsem^{y}$ is defined in \refeq{eq:deny}), which can be shown in the following way. Let
$|z| \in C_{\Gamma}$. Then we have
\begin{align*}
  \mathfrak{M}_{\Gamma},v_{\Gamma}[|z|/y] \models \phi 
  &\Longleftrightarrow 
    \mathfrak{M}_{\Gamma},v_{\Gamma} \models \phi[z/y]\\
  &\Longleftrightarrow \phi[z/y] \in \Gamma \qquad \text{by \eqref{eq:truth},}
\end{align*}
as desired. 

The base case of induction for atomic formulas follows now from the
definitions of $C_\Gamma$ and $I_\Gamma$, the Boolean cases from
the fact that we are dealing with a MCS, and the case for quantifiers
directly from Condition~\ref{eq:henkinqu}. For the modal case, 
 where
$\phi \equiv  x \hearts \lcmh y_1 : \phi_1 \rcmh
\dots \lcmh y_n: \phi_n \rcmh$: 
\begin{align*}
\mathfrak{M}_{\Gamma},v_{\Gamma} \models x \hearts \lcmh y_1 : \phi_1 \rcmh
\dots \lcmh y_n: \phi_n \rcmh
&\Longleftrightarrow \gamma(v_{\Gamma}(x)) \in  \lsem \heartsuit \rsem_{\gM_{\Gamma}} ( \lsem \phi_1 \rsem^{y_1}_{\gM_\Gamma},
\dots, \lsem \phi_n \rsem^{y_n}_{\gM_\Gamma}) & \text{by def.} \\
&\Longleftrightarrow \gamma(|x|) \in  \lsem \heartsuit \rsem_{\gM_{\Gamma}} ({\widehat{\phi_1}}^{y_1},\dots,{\widehat{\phi_n}}^{y_n}) & \text{by \refeq{eq:semfa}}\\
&\Longleftrightarrow  x\heartsuit \lceil y_{1} : \phi_{1} \rceil \cdots \lceil y_{n} : \phi_{n} \rceil \in \Gamma & \text{by \refeq{eq:henkinmo}.}
\end{align*}
\end{proof}

\tlnew{Referee complains, but qedhere does not work}

\noindent Next, we need to find a suitable candidate for an MCS from
which to build our quasi-Henkin model. Consider the following sets of
inferences:

\begin{align*}
\mathit{Inf}_{\tsc{name}a} \deq & \{ \langle\{\sbst{\phi}{x}{z} \mid z \in \FoV\}, \forall x.\phi \rangle \mid \phi \in \CPL(\PrLSet,\SoV), x \in \FoV \} \\
\mathit{Inf}_{\tsc{name}b} \deq & \{ \langle\{\sbst{\lb\phi_1 \eqvF \psi_1\rb}{x}{z}, \dots, \sbst{\lb\phi_n \eqvF \psi_n\rb}{x}{z} \mid z \in \FoV\}, \\
\forall x.&\lb x\PrL\cmh{x}{\phi_1}\dots\cmh{x}{\phi_n}\eqvF x\PrL\cmh{x}{\psi_1}\dots\cmh{x}{\psi_n}\rb \rangle \mid \vec{\phi}, \vec{\psi} \in \CPL(\PrLSet,\SoV), x \in \FoV \} \\
\mathit{Inf}_{\tsc{name}} \deq & \mathit{Inf}_{\tsc{name}a} \cup \mathit{Inf}_{\tsc{name}b} \\
\mathit{Inf}_{\bel} \deq & \{ \langle\{\bigwedge\limits_{j \leq \bel(\PrL)(i)}\sbst{\phi_i}{y_i}{z_j} \impF \neg x\PrL\dots\cmh{y_i}{\bigvee\limits_{j \leq \bel(\PrL)(i)}y_i \eqF z_j}\dots   \mid \vec{z} \in \FoV\}, \\
& \qquad\qquad \neg x\PrL\dots\cmh{y_i}{\phi_i}\dots \rangle \mid \vec{\phi} \in \CPL(\PrLSet,\SoV), x \in \FoV, \PrL \in \Lambda, \bel(\PrL)(i) \neq \infty\} \\
 \mathit{Inf} \deq & \mathit{Inf}_{\tsc{name}} \cup \mathit{Inf}_{\bel}
\end{align*}


Let us begin with

\begin{clm} 
Assume $\card{\FoV} = \kappa \geq \card{\PrLSet \cup \SoV \cup \omega}$. Then any $\gdash$-consistent set of formulas $\Gamma$ s.t. $\card{\{x \in \FoV \mid x \ntoc \Gamma\}} = \kappa$ (in particular, any consistent set of sentences) satisfies condition \ref{eq:finrespect} of Theorem \ref{th:goldhenkin} for $\mathit{Inf}_{\tsc{name}}$.
\end{clm}

\begin{proof}

We begin by observing that 
\begin{center}
  (a) If $\Gamma' \gdash \sbst{\phi}{x}{z}$ and $z$ is fresh for
  $\Gamma', x, \phi$, then $\Gamma' \gdash \forall x.\phi$
\end{center}
The proof of this fact is perfectly standard, but working with an Enderton-style axiomatization is particularly convenient for such reasoning: We have a finite $\Gamma'_0 \finsubset \Gamma'$ s.t. $\Gamma'_0 \gdash \sbst{\phi}{x}{z}$. Then one uses the Deduction Theorem 
(cf. Remark \ref{rem:noinfinitary}) to obtain $\gdash \bigwedge \Gamma'_0 \to  \sbst{\phi}{x}{z}$. However, even with an Enderton-style axiomatization it is still the  case\footnote{In fact, a variant of the Generalization Theorem is available even for non-empty contexts as long as the quantified variable does not occur freely therein, cf. \cite[p. 117]{Enderton72}.} that $\gdash \chi$ implies $\gdash \forall z.\chi$, hence $\gdash \forall z. (\bigwedge \Gamma'_0 \to  \sbst{\phi}{x}{z})$. The rest is an easy exercise using \axref{En3}, \axref{En4} and renaming of bound variables thanks to \axref{En2}.

The condition (a) tells us that $\Gamma$ itself does respect $\mathit{Inf}_{\tsc{name}a}$ by assumption. But if $\card{X} < \kappa$, then there are $\kappa$-many $z \in \FoV$ that are fresh for $\Gamma \cup X \cup \{\phi\}$. For any such $z$, (a) would hold also for $\Gamma' = \Gamma \cup X$. This gives condition \ref{eq:finrespect} for $\mathit{Inf}_{\tsc{name}a}$.

For $\mathit{Inf}_{\tsc{name}b}$, let us observe that (a) allows to infer that 
\begin{center}
If $\Gamma' \gdash \sbst{\lb\phi_1 \eqvF \psi_1\rb}{x}{z} \wedge \dots \wedge \sbst{\lb\phi_n \eqvF \psi_n\rb}{x}{z}$ and $z \ntoc \Gamma', \vec{\phi},\vec{\psi},x$, then
$\Gamma \gdash \forall x.\lb\lb \phi_1 \eqvF \psi_1 \rb \wedge \dots \wedge \lb \phi_n \eqvF \psi_n \rb\rb$.
\end{center}
Now an application of the congruence rule (Remark~\ref{rem:cong})
completes the proof of the claim.
\end{proof}

\begin{clm} 
Assume $\card{\FoV} = \kappa \geq \card{\PrLSet \cup \SoV \cup \omega}$. Then any $\gdash$-consistent set of formulas $\Gamma$ s.t. $\card{\{x \in \FoV \mid x \ntoc \Gamma\}} = \kappa$ (in particular, a consistent set of \emph{sentences}) satisfies condition \ref{eq:finrespect} of Theorem \ref{th:goldhenkin} for $\mathit{Inf}_{\bel}$.
\end{clm}

\begin{proof}

We begin by observing that 

\begin{center}
(b) If $\Gamma' \gdash \neg(\bigwedge\limits_{j \leq \bel(\PrL)(i)}\sbst{\phi_i}{y_i}{z_j} \wedge x\PrL\dots\cmh{y_i}{\bigvee\limits_{j \leq \bel(\PrL)(i)}y_i \eqF z_j}\dots)$    for some $\vec{z}\ntoc \Gamma', x, y_i, \vec{\phi}$,  then $\Gamma' \gdash \neg x\PrL\dots\cmh{y_i}{\phi_i}\dots$
\end{center}

This is shown by first following the proof of (a) and finding  a finite $\Gamma'_0 \finsubset \Gamma'$ s.t. $$\gdash \Gamma'_0  \to \neg\exists z_1, \dots, z_{\bel(\PrL)(i)}.(\bigwedge\limits_{j \leq \bel(\PrL)(i)}\sbst{\phi_i}{y_i}{z_j} \wedge x\PrL\dots\cmh{y_i}{\bigvee\limits_{j \leq \bel(\PrL)(i)}y_i \eqF z_j}\dots).$$ Applying \axref{\bdpl}\ proves (b).

The condition (b) tells us that $\Gamma$ itself does respect $\mathit{Inf}_{\bel}$ by assumption. But if $\card{X} < \kappa$, then there are $\kappa$-many $z \in \FoV$ which are fresh for $\Gamma \cup X \cup \{\phi_1, \dots, \phi_{\arty{\PrL}}\}$ and distinct from $x$ and $\vec{y}$. For any tuple of such $z$'s, (b) would hold also for $\Gamma' = \Gamma \cup X$. This gives condition \ref{eq:finrespect} for $\mathit{Inf}_{\bel}$. \qedhere
\end{proof}

\begin{clm} \label{cl:named_and_pasted}
Assume $\card{\FoV} = \kappa \geq \card{\PrLSet \cup \SoV \cup \omega}$. Then any $\gdash$-consistent set of formulas $\Gamma$ s.t. $\card{\{x \in \FoV \mid x \ntoc \Gamma\}} = \kappa$ (in particular, a consistent set of \emph{sentences}) can be extended to a maximally $\gdash$-consistent set of formulas $\Gamma'$ s.t. 
\begin{itemize}
\item whenever $\exists x.\phi \in \Gamma'$, then $\sbst{\phi}{x}{z} \in \Gamma'$ for some $z \in \FoV$
\item whenever $\exists x.\lb x\PrL\cmh{x}{\phi_1}\dots\cmh{x}{\phi_n}\wedge \neg x\PrL\cmh{x}{\psi_1}\dots\cmh{x}{\psi_n} \rb \in \Gamma'$, then there is $z~\in~\FoV$ and $i\leq n$ s.t. $\neg\sbst{\lb\phi_i \eqvF \psi_i\rb}{x}{z} \in \Gamma'$.
\item whenever $x\PrL\dots\cmh{y_i}{\phi_i}\dots \in \Gamma'$ and $\bel(\PrL)(i) \neq \infty$, then there are $z_1, \dots, z_{\bel(\PrL)(i)}$ s.t. $x\PrL\dots\cmh{y_i}{\bigvee\limits_{j \leq \bel(\PrL)(i)}y_i \eqF z_j}\dots  \in \Gamma'$ and moreover $\sbst{\phi_i}{y_i}{z_j}  \in \Gamma'$ for each $j \leq \bel(\PrL)(i)$. 
\end{itemize}
\end{clm}

\begin{proof}
This immediately follows from the preceding Claim, Theorem \ref{th:goldhenkin} and the fact $\Gamma'$ is a  MCS.
\end{proof}

\begin{proof}[Proof of Theorem \ref{thm:hilb-complete}]
Recall Definition \ref{def:quasih}. We will build our quasi-Henkin model using $\Gamma'$. Satisfaction of condition \ref{eq:henkinqu} follows then directly from the first item in Claim \ref{cl:named_and_pasted}, i.e., from being closed under $\mathit{Inf}_{\tsc{name}a}$.  Hence, we just need to define a transition structure $\gamma$ on $C_{\Gamma'}$ and for this purpose, we need to find for each $|x|$ a suitable $t \in TC_{\Gamma'}$ s.t. when $\gamma(|x|)$ is defined as $t$, the condition \ref{eq:henkinmo} is satisfied. 

Assume then $\RSch$ has enough schematic variables to name all elements of $\CPL(\PrLSet,\SoV)$; let $\sva_\phi$ be the schematic variable corresponding to $\phi$ under some fixed assignment. For each $x$, we can define an evaluation $\tau_x(\sva_\psi) = \widehat{\psi}^{x}$. Note that for each pair of distinct $x$ and $y$ we have that  $\tau_x(\sva_{x= y})$ is a singleton, thanks to the definition of $C_{\Gamma'}$.

Thus, let us define for each $x \in \FoV$ the set 
$$
\Psi_x := \{ \epsilon\PrL\sva_{\psi^\circ_1}\dots\sva_{\psi^\circ_n} \mid \psi_1, \dots, \psi_n \in \CPL(\PrLSet,\SoV) \text{ and } \epsilon(x\PrL\lcomp x \col \psi^\circ_1 \rcomp\dots\lcomp x \col \psi^\circ_n \rcomp) \in \Gamma'\},
$$

where $\epsilon$ is either nothing or negation and for each $i \leq \arty{\PrL}$, $\psi^\circ_i$ is either:
\begin{itemize}
\item $\psi_i$ itself, if $\bel(\PrL)(i) = \infty$ or 
\item $\bigvee\limits_{j \leq \bel(\PrL)(i)}x \eqF z_j$ otherwise, where $z_1, \dots, z_{\bel(\PrL)(i)}$ are s.t. 
\begin{itemize}
\item $x\PrL\dots\cmh{x}{\bigvee\limits_{j \leq \bel(\PrL)(i)}x \eqF z_j}\dots  \in \Gamma'$ and moreover 
\item $\sbst{\psi_i}{x}{z_j}  \in \Gamma'$ for each $j \leq \bel(\PrL)(i)$. 
\end{itemize}
\end{itemize}

Furthermore, let us define a colouring $\colo_x$ of schematic variables which assigns  $fin$ to every $\sva_{\bigvee\limits_{j \leq m}x \eqF z_m}$, where $z_1, \dots, z_m$ is any finite sequence of variables and $\infty$ to every other $\sva_{\psi}$. It is clear that
 $\tau_x$ respects $\colo_x$. We have: 
 \begin{clm}  \label{clm:psix}
  $\Psi_x$ is  $\colo_x$-consistent wrt $\tau_x$, i.e., 
 \begin{equation*}
\Psi_x \cup \lbrace \svX
\sigma \mid\sigma:\SchV \to\BSch(\SchV) \text{ and } \svA/ \svX \in
\bbel_\colo(\Rules) \text{ s.t. } C_{\Gamma'}, \tau_x \models \svA \sigma\rbrace.
\end{equation*}
is propositionally consistent.
\end{clm}

\begin{proof}
Assume it is not. By compactness of the classical propositional calculus, a contradiction can be then derived already from a certain finite subset of $\Psi_x$ and finitely many $\svX
\sigma$ s.t.  $\svA/ \svX \in
\bbel_\colo(\Rules)$ and $C_{\Gamma'}, \tau_x \models \svA \sigma$. Given the definition of $\tau_x$ and $\Psi_x$, this directly contradicts the fact that $\Gamma'$ is supposed to be a MCS closed under all instances of axioms  \axref{Onestep($\Rules$)} and \bdpl\ in Table \ref{tab:folaxioms}. 
\end{proof}
%
By Definition \ref{def:finitarysc} of $\bel$-S1SC, Claim \ref{clm:psix} implies that $\Psi_x$ is one-step satisfiable wrt $\tau_x$, i.e., 
\begin{multline*}\textstyle
\bigcap_{\phi \in \Psi_x} \lsem \phi \rsem_{TC_{\Gamma'}, \tau_x}  =  \{ \epsilon\Sem{\PrL}_{C_\Gamma'}\widehat{\psi^\circ_1}^x\dots\widehat{\psi^\circ_n}^x \mid  \psi_1, \dots, \psi_n \in \CPL(\PrLSet,\SoV) \text{ and } \\
  \epsilon(x\PrL\lcomp x \col \psi^\circ_1 \rcomp\dots\lcomp x \col \psi^\circ_n \rcomp) \in \Gamma'\} 
\neq \emptyset. 
\end{multline*}
Now use the Axiom of Choice to define $\gamma(|x|)$ to be a representative of this non-empty intersection for every $|x|$ (strictly speaking, for an arbitrarily chosen representative of this equivalence relation), automatically yielding the condition \ref{eq:henkinmo} of Definition \ref{def:quasih}.
\end{proof}




%


\begin{rem} \label{rem:hyfa}
The similarities and differences between CPL and languages like $\HyFA$ and its extensions to be discussed in \S~\ref{sec:cml} are best appreciated by comparing the proof of Theorem \ref{th:completeness} with earlier hybrid ones \cite{SchroderPattinson10}. In the predicate case:
\begin{itemize} 
\item not only one-step rules, but also non-standard naming and pasting rules of \cite{SchroderPattinson10} can be expressed as ordinary first-order axioms. 
\item As we are going to discuss now, the Henkin-style completeness proof directly leads to the Omitting Types theorem. 
It is not clear how to obtain such a result for a language like $\HyFA$ studied in \cite{SchroderPattinson10}; the presence of a binding and/or quantification mechanism seems essential in the proof. Recall again that the presence of such a mechanism also allowed us to reuse (equivalence classes of) variables as building block of models instead of Henkin-style constants. 
\end{itemize}
\end{rem}



\subsection{Omitting Types Theorem}

The Omitting Types Theorem is a standard result of model theory. Goldblatt \cite[\S~8.2]{goldblatt1993:abstract} shows how to 
establish it using  the Abstract Henkin Principle. 
 Here is a more detailed description how to obtain it in our setting. In this section, we assume that the entire $\CPL(\PrLSet,\SoV)$ is countable and we keep these countable $\PrLSet$ and $\SoV$ fixed and implicit.
 
Fix a finite subset of $\FoV$ $\{x_1, \dots, x_k\}$ and denote the set of all formulas whose free variables are contained in $\{x_1, \dots, x_k\}$ as $\CPL(k)$. Thus, the set of sentences can be written as $\CPL(0)$. Recall that a $k$-\emph{type} (sometimes called a \emph{complete type}) is a maximal consistent subset of $\CPL(k)$; sometimes, one also uses the term \emph{partial type} for consistent yet not maximal subsets of $\CPL(k)$. For any given $\Gamma \subseteq \CPL(0)$ and any (partial or total) $k$-type $\Delta$, say that $\Delta$ is \emph{principal over $\Gamma$} if there is $\phi \in \CPL(k)$ consistent with $\Gamma$ s.t. $\forall \psi \in \Delta. \Gamma \gdash \phi \to \psi$. Note here that for complete types, we can assume that $\phi \in \Delta$. Say that a model $\gM = (C, \gamma, I)$ \emph{realizes} a (partial or total) $k$-type $\Delta$ if $\bigcap\limits_{\psi\in\Delta}\lsem \psi \rsem^{x_1,\dots,x_k}_C \neq \emptyset$, where as before 
$$\lsem \phi \rsem^{x_1,\dots,x_k}_C \deq \lbrace (c_1,\dots,c_k) \in C \mid \gM, \vmd{c_1}{x_1}\dots[c_k/x_k]
\models \phi \rbrace;$$ 
a $k$-type is \emph{omitted} by $\gM$ if it is not realized by it.

Note that for complete types, one consequence of being non-principal is that $\Delta$ is neither entailed by $\Gamma$ nor inconsistent with it (maximal consistent sets are closed under finite conjuctions).

\begin{thm}[Omitting Types] \label{th:omt}
  Whenever a set of rules $\Rules$ is \emph{$\bel$-S1SC} for a
  $\Lambda$-structure over $T$ that is adequate for $\bel$, $\Gamma$ is a  consistent set of sentences and $\Delta$ is a \bro complete or partial\brc\ $k$-type non-principal over $\Gamma$, $\Gamma$ has a model omitting $\Delta$.
   \end{thm}
   
\begin{proof}
We only need to refine somewhat the proof of the completeness theorem by using a richer set of inferences than $Inf$. Consider $$Inf_\Delta = Inf \cup  \{ \langle \{\sigma[z_1\dots z_k/x_1\dots x_k] \mid \sigma \in  \Delta \} , \bot \rangle \mid z_1, \dots, z_k \text{ distinct els. of } \FoV \}$$ 
(we have not formally defined simultaneous substitution, but it should be clear how to extend conventions introduced in \S~\ref{sec:ss}). We claim that the condition \ref{eq:finrespect} of Theorem \ref{th:goldhenkin} is satisfied with $\kappa = \omega$. For assume it is not. Then there exists a finite set $\Delta \subseteq \CPL$   and a finite tuple of distinct variables $z_1, \dots, z_k$ s.t. (*) $\Gamma \cup \Delta$ is consistent but  
\begin{center}
 $\Gamma  \gdash \bigwedge \Delta \to \sigma[z_1\dots z_k/x_1\dots x_k] $ for every $\sigma \in \Delta$.
\end{center}
Let $\vec{z'}$ be a sequence containing all the variables in $\Delta$ different from $z_1\dots z_k$ and $\delta \deq \exists\vec{z'}.\bigwedge \Delta$. Then we have
 \begin{center}
(**) $\Gamma  \gdash \delta \to \sigma[z_1\dots z_k/x_1\dots x_k] $ for every $\sigma \in \Delta$
\end{center}
and consequently, setting $\delta'$ to be $\delta[x_1\dots x_k/z_1\dots z_k]$
 \begin{center}
 (***) $\Gamma  \gdash  \delta' \to \sigma$ for every $\sigma \in \Delta$
 
(in deriving (**) and (***) we obviously use the fact that $\Gamma$ is a set of \emph{sentences}).
 \end{center}

 \takeout{As $\Delta$ is not principal over $\Gamma$ and $\delta' \in \CPL(k)$, we have that $\delta' \not\in \Delta$, hence 
$\neg\delta' \in \Delta$. By (***) this means that $\Gamma  \gdash \neg\delta'$, and using again renaming and the fact that $\Gamma$ is a set of sentences, we obtain  a contradiction with  (*).} 
At the same time, $\Gamma$ being a set of sentences yields that $\delta'$ is consistent with $\Gamma$ by virtue of (*). This entails a contradiction with non-principality of $\Delta$ over $\Gamma$. The rest proceeds as in the completeness proof.
\end{proof}   

\begin{rem}
 Goldblatt \cite[\S~8.2]{goldblatt1993:abstract} points out this can be extended to simultaneously omitting a countable set of non-principal types. 
 \end{rem}

\newcommand{\is}[2]{#1\,\mathsf{is}\,#2}
\newcommand{\isna}[1]{\is{#1}{\mathsf{nat}}}
\newcommand{\isnu}[2]{\is{#1}{\underline{#2}}}

  \noindent
 Here is a typical application adapted from the monograph of Chang and Keisler \cite[Ch 2.2, p. 83]{ChangK90}: $\omega$-logic, $\omega$-rule and $\omega$-completeness. 
 
 Given any $T$, $\Lambda$ and $\Rules$ within the scope of our completeness result, extend $\Lambda$ with a countable family of propositional atoms (cf. Definition \ref{def:atom}) $\{\isnu{}{n} \mid n \in \Nat \} \cup \{ \isna{} \}$. By similar considerations as in Remark \ref{rem:atomsafe}, the completeness result is not affected by such extensions. 
 
 \begin{rem}
Of course, we could alternatively extended $\Sigma$ with corresponding predicate symbols, which would be even easier from the point of view of directly applying Theorems~\ref{th:goldhenkin} and~\ref{th:omt}, but would also have a less coalgebraic flavour.
 \end{rem}
 
 \noindent
 Consider now the following set of sentences in the extended language:
 
 \begin{align*}
 \Gamma_N = & \, \{\forall x. \isnu{x}{n} \to (\isna{x} \wedge \neg \isnu{x}{m})  \mid n, m \in \Nat, n \neq m\}  \, \cup \\
 & \, \{\forall x, y. \isnu{x}{n} \wedge  \isnu{y}{n} \to x = y \mid n \in \Nat\} \, \cup \\
 & \{ \exists x. \isnu{x}{n} \mid n \in \Nat \}.
 \end{align*}
 
 \noindent
 An $\omega$-\emph{model} is any model of $\Gamma_N$ where, moreover, the denotation of $\isna{}$ is the set-theoretical sum of all ``$\isnu{}{n}$''. A theory $\Gamma$ is \takeout{$\omega$-\emph{consistent} if \takeout{$\Gamma \cup \Gamma_N$ has an $\omega$-model, i.e., (assuming we work with structures within the scope of our completeness result) if} there is no formula $\phi(x)$ s.t. $\Gamma \cup \Gamma_N$ entail the following set of sentences:

\begin{equation} \label{eq:omco}
 \{\forall x. \isnu{x}{n} \to \phi(x) \mid n \in \Nat\} \cup \{\exists x. \isna{x} \wedge \neg\phi(x)\}. 
\end{equation}

\noindent
$\Gamma$ is} $\omega$-\emph{complete} if $\Gamma \cup \Gamma_N$ is closed under the $\omega$-\emph{rule}\takeout{formalizing this notion of consistency}:

\begin{equation*} 
\nrule{\rul{\omega}}{\forall x. \isnu{x}{n} \to \phi(x) \quad n \in \Nat}{\forall x. \isna{x} \to \phi(x)}
\end{equation*}

\noindent
One can use the Omitting Types theorem to show the following:

\begin{cor}
Assume that a set of rules $\Rules$ is \emph{$\bel$-S1SC} for a
  $\Lambda$-structure over $T$ is adequate for $\bel$, and $\Gamma$ is a set of sentences s.t. $\Gamma \cup \Gamma_N$ is consistent.
  \takeout{
  \begin{itemize}
  \item} If $\Gamma$ is $\omega$-complete, then $\Gamma$ has an $\omega$-model.
  \takeout{
  \item If $\Gamma$ has $\omega$-model, then $\Gamma$ is $\omega$-consistent.
  \end{itemize}}
\end{cor}

\begin{proof}[Sketch]
\takeout{Only the first part requires proving.} Consider 
\[
\Delta(x) = \{ \neg(\isnu{x}{n}) \mid n \in \Nat\} \cup \{\isna{x}\}.
\]
\takeout{If $\Delta(x)$ is inconsi}
Pick any $\phi(x) \in \CPL(1)$ s.t. $\isna{x} \wedge \phi(x)$ is  consistent with $\Gamma \cup \Gamma_N$. Then 
\begin{align*}
& \Gamma \cup \Gamma_N \not\gdash \forall x. (\isna{x} \to \neg\phi(x)) & \\
\Longrightarrow\qquad & \Gamma \cup \Gamma_N \not\gdash \forall x. (\isnu{x}{n} \to \neg\phi(x))  \text{ for some } n \in \Nat & \text{(by } \omega-\text{completeness),} \\
 \Longleftrightarrow\qquad & \Gamma \cup \Gamma_N \not\gdash  \phi(x) \to \neg(\isnu{x}{n}) \text{ for some } n \in \Nat.
\end{align*}

\noindent
On the other hand, if
\[
 \Gamma \cup \Gamma_N \gdash \isna{x} \to \neg\phi(x) \]
 and at the same time
 \[
  \Gamma \cup \Gamma_N \gdash \phi(x) \to \isna{x},
 \]
  then $\phi(x)$ is inconsistent with $\Gamma \cup \Gamma_N$. Thus,  $\Delta(x)$ \takeout{containing $\isna{x}$} is non-principal over $\Gamma \cup \Gamma_N$ and hence, $\Gamma \cup \Gamma_N$ has a model omitting it. Such a model is an $\omega$-model.
\end{proof}

\noindent
As discussed by Chang and Keisler \cite[Ch 2.2]{ChangK90}, it follows that we can extend our deductive apparatus with $\Gamma_N$ as additional axioms and $\omega$-rule as an additional rule of proof and consistency in this extended system is equivalent to the existence of an $\omega$-model.


 
 
 Such examples are worth contrasting with  the incompleteness result (Theorem \ref{th:noncompact}) we are going to present next. 







\subsection{\texorpdfstring{$\omega$}{Omega}-boundedness and Failure of Completeness} \label{sec:omega}

\noindent
In this subsection, we show that there is a substantial gap between
S1SC and finitary S1SC as conditions allowing for strong completeness,
by proving that within a larger class of \emph{$\omega$-bounded}
structures, the bounded structures are the only ones that satisfy
compactness. Here, $\omega$-boundedness of an operator means
informally that its satisfaction can always be established by looking
only at a finite subset of the successors, without however requiring a
fixed bound on their number. In examples for this property, we
concentrate on cases additionally satisfying finitary one-step
compactness (Definition~\ref{def:os-compact}), a condition 
 essentially necessary for overall compactness and
that will moreover become important in our forays into model theory
(\S~\ref{sec:modeltheory}).  In the whole subsection, to keep things
simple we work with unary $\hearts \in \PrLSet$.

\begin{defi}[$\omega$-Bounded operators]
A modal operator $\hearts$ is \emph{$\omega$-bounded} if for each
set $X$ and each $A\subseteq X$,
\begin{equation*}
  \Sem{\hearts}_X(A)=\bigcup_{B\finsubset A}\Sem{\hearts}_X(B).
\end{equation*}
\end{defi}

\begin{exa}[Nonstandard subdistributions] \label{ex:hyperprob} We
  generally write $\SDist$ for the discrete subdistribution functor,
  i.e.\ $\SDist(X)$ consists of real-valued discrete measures $\mu$ on
  $X$ such that $\mu(X)\le 1$, and for maps $f$, $\mu(f)$ takes image
  measures. As a variant of this functor, we consider the the discrete
  subdistributions functor $\SDistN$ where measures take values in
  real-closed fields. Explicitly: we intend to model Markov chains
  with non-standard probabilities; these consist of a set $X$ of
  states, and at each state $x$ an $R_x$-valued transition
  distribution $\mu_x$, where $R_x$ is a real-closed field (i.e.\ a
  model of the first-order theory of the
  reals). 
  These structures are coalgebras for the functor $T$ which maps a set
  $X$ to the set of pairs $(R,\mu)$ where $R$ is a real-closed field
  and $\mu$ is an $R$-valued discrete subdistribution on~$X$ (again
  meaning that $\mu(X)\le 1$). This functor is in fact class-valued,
  which however does not affect the applicability of our coalgebraic
  analysis (which never requires iterated application of the
  coalgebraic type functor, e.g.\ it does not use the terminal
  sequence).  We take the modal signature $\Lambda$ to consist of the
  operators $\langle p\rangle$ (`with probability more than $p$') for
  $p\in[0,1]\cap\Rat$.

  We show that the $\langle p\rangle$ are $\omega$-bounded and that the arising
  logic $\Lang$ is finitary one-step compact. To see the former, let
  $(R,\mu)\in TX$ and let $A\subseteq X$ such that $\mu\models \langle p\rangle A$,
  i.e.\ $\sum_{x\in A}\mu(x)>p$. Then there exists $B\finsubset A$
  such that $\sum_{x\in B}\mu(x)>p$, i.e.\ $\mu\models \langle p\rangle B$. Since
  $\langle p\rangle$ is clearly monotone, this implies that
  $\Sem{\langle p\rangle}_X(A)=\bigcup_{B\finsubset A}\Sem{\langle p\rangle}_X(B)$, as required.

  To show that $\Lang$ is finitary one-step compact, let
  $\Phi\subseteq\Prop(\Lambda(\Pfin(X)))$ be finitely satisfiable.
  Extend the standard language of real arithmetic with a constant
  symbol $c_x$ for each element of $X$, obtaining a language $L$. Then
  satisfaction of a formula in $\Prop(\Lambda(\Pfin(X)))$ by
  $\mu\in TX$ translates into a first-order formula over $L$ with
  $c_x$ representing $\mu(x)$; specifically, the translation $t$
  commutes with the Boolean connectives and translates formulas
  $\langle p\rangle A$ with $A\in\Pfin(X)$ into $\sum_{x\in A}c_x>p$. Applying $t$
  to $\Phi$ and introducing additional formulas $c_x\ge 0$ for all
  $x\in X$ and $\sum_{x\in A}c_x\le 1$ for all $A\in\Pfin(X)$ thus
  produces a finitely satisfiable, and hence satisfiable, set of
  first-order formulas over $L$. A model of this set consists of a
  real-closed field $R$ and interpretations $\hat c_x\in R$ of the
  constants $c_x$ such that putting $\mu(x)=\hat c_x$ defines a
  discrete subdistribution (note that $\sum_{x\in A}\hat c_x\le 1$ for
  all $A\in\Pfin(X)$ implies $\sum_{x\in X}\hat c_x\le 1$), which
  then yields a model $(R,\mu)$ of $\Phi$.
\end{exa}

\begin{exa}[Zero-dimensional subdistributions] \label{ex:zerodis} Fix
  a zero-dimensional closed (hence compact) subset $Z\subseteq[0,1]$,
  e.g.\ a discrete set or the Cantor space, and let $\SDistZ$ be the
  associated \emph{zero-dimensional discrete subdistributions
    functor}, i.e.\ the subfunctor of the subdistribution functor
  $\SDist$ where probabilities of finite sets of states are restricted
  to take values in $Z$:
  \begin{equation*}
    \SDistZ(X)=\{\mu\in\SDist(X)\mid\forall A\in\Pfin(X).\,\mu(A)\in Z\}.
  \end{equation*}
  Moreover, we restrict the probabilities~$p$ in operators $\langle p\rangle$ to be
  such that $(p,1]\cap Z$ is clopen in $Z$; since $Z$ is
  zero-dimensional, there exist enough such $p$ to separate all values
  in~$Z$. As before, all these operators are $\omega$-bounded.  It
  remains to show that the logic is finitary one-step compact. So let
  $\Phi\subseteq\Prop(\Lambda(\Pfin(X)))$ be finitely
  satisfiable. Note that the space $Z^X$, equipped with the product
  topology, is compact. We equip $\SDistZ(X)$ with the subspace
  topology in $Z^X$. 
  Observe that the condition $\forall A\in\Pfin(X).\,\mu(A)\in Z$
  already implies $\mu(X)\le 1$; since for $A\in\Pfin(X)$, the
  summation map $Z^A\to Z$ is continuous (this would fail for
  infinite~$A$), it follows that $\SDistZ(X)$ is closed in $Z^X$,
  hence compact.

  By the restriction placed on the indices $p$ in modal operators
  $\langle p\rangle$, and again using continuity of finite summation, we have that
  for every formula $\langle p\rangle A$ with $A\in\Pfin(X)$, the extension
  \begin{equation*}
    \Sem{\langle p\rangle A}=\{\mu\in\SDistZ(X)\mid \mu(A)>p\}
  \end{equation*}
  is clopen in $\SDistZ(X)$. As clopen sets are closed under Boolean
  combinations, we thus have that the extension of every formula in
  $\Prop(\Lambda(\Pfin(X)))$ is clopen in $\SDistZ(X)$.  Let $\FA$
  denote the family of clopens induced in this way by formulas in
  $\Phi$. Finite satisfiability of $\Phi$ implies that~$\FA$ has the
  finite intersection property, and hence has non-empty intersection
  by compactness of $\SDistZ(X)$. It follows that $\Phi$ is
  satisfiable.
\end{exa}
\noindent \takeout{Structures which are $\omega$-bounded without being $k$-bounded behave very differently to those within the scope of our completeness results. As an immediate corollary of the Omitting Types Theorem, we have:
 
\begin{cor} \label{cor:omega}
Assume that a set of rules $\Rules$ is \emph{$\bel$-S1SC} for a
  $\Lambda$-structure over $T$ is adequate for $\bel$. If $\hearts$ is $\omega$-bounded, then it is $k$-bounded for some $k \in \Nat$.
\end{cor}

\begin{proof}
Of course, whenever $\bel(\hearts) \neq \infty$, there is nothing to prove. Thus, assume that $\bel(\hearts) = \infty$, $\hearts$ is $\omega$-bounded, but not $k$-bounded for any $k \in \Nat$. Consider $\Sigma = \{ P \}$ and
\begin{align*}
  \Delta(x) = & \{\neg x\hearts\cmh{y}{P(y)}\} \; \cup \\
           & \{\exists y_1,...,y_k. (P(y_1) \wedge \dots \wedge P(y_k) \wedge x\hearts\cmh{y}{y=y_1 \vee ... \vee y=y_k}) \mid k \in \omega\}.
\end{align*}
Consider any $\phi(x)$ s.t. $\not\gdash \phi(x) \to x\hearts\cmh{y}{P(y)}$.

\takeout{By assumption, for every $k \in \Nat$ there is $C_k$ and $B_k$ s.t.
\begin{equation}
\bigcup_{B\subseteq A, \#B\le k}\Sem{\PrL}_{C_k}(B) \subsetneq \Sem{\PrL}_C(A_k)
\end{equation}
(we get monotonicity by $\omega$-boundedness, hence we can rule out the strict inequality in the opposite direction).}
\end{proof}

There is also a natural converse to this result: a kind of
incompleteness theorem.}  Structures which are $\omega$-bounded
without being $k$-bounded fail strong completeness. To state this
observation in full generality, we require the notion of
\emph{propositional atom} as defined previously
(Definition~\ref{def:atom}).

\begin{thm} \label{th:noncompact} Whenever a $\Lambda$-structure makes
  some $\PrL \in \PrLSet$ $\omega$-bounded without being $k$-bounded
  for any $k \in \omega$, strong completeness fails whenever either
  $\Sigma$ contains a predicate symbol of positive arity or $\PrLSet$
  contains a propositional atom.
\end{thm}

\begin{proof}
  Assume $P \in \SoV$ is a predicate symbol of positive arity,
  w.l.o.g.\ unary, and $\PrL \in \PrLSet$ is as in the statement of
  theorem.
  Consider 
\begin{align*}
  \Delta(x) = & \{\neg x\hearts\cmh{y}{P(y)}\} \; \cup \\
           & \{\forall y_1,...,y_k. (P(y_1) \wedge \dots \wedge P(y_k) \to \neg x\hearts\cmh{y}{y=y_1 \vee ... \vee y=y_k}) \mid k \in \omega\}.
\end{align*}  
Clearly, every finite subset of $\Delta(x)$ is satisfiable in a model
based on a coalgebra witnessing the failure of $k$-boundedness for a
suitably large $k$. However, a coalgebraic model satisfying the whole
$\Delta(x)$ would witness the failure of $\omega$-boundedness. This means
that~$\Delta(x)$ is a counterexample to compactness, and hence no
finitary deduction system can be strongly complete. The proof for the
case where $\PrLSet$ contains a propositional atom is entirely
analogous.
\end{proof}
\begin{exa}
  The probabilistic instances of CPL given by interpreting the
  probabilistic modalities $\langle p\rangle$ over nonstandard or
  zero-dimensional subdistributions, respectively, and are
  $\omega$-bounded but fail to be $k$-bounded for any $k$. Hence they
  fail to be compact by Theorem~\ref{th:noncompact} (once equipped
  with propositional atoms) although they satisfy finitary one-step
  compactness (Examples~\ref{ex:hyperprob} and~\ref{ex:zerodis}).
\end{exa}

\section{Correspondence with Coalgebraic Modal Logic} \label{sec:cml}

\noindent We next compare the expressivity of CPL with that of
various coalgebraic modal and hybrid logics. 

\subsection{Coalgebraic Standard Translation for CML}

\noindent
The formulas $\MoF\SoV$ of pure (coalgebraic) modal logic in the modal signature $\PrLSet$ over $\SoV$ (now all elements of $\SoV$ are assumed to be of arity 1) are given by the grammar:

$$ \MoF\SoV \quad \phi, \psi ::=   P \mid \botF \mid \phi \impF \psi \mid 
\heartsuit (\phi_1, \dots, \phi_n),$$
where $P \in \SoV$. 

Satisfaction is defined with respect to $\gM = (C,\gamma, I)$ 
and a specific point $c \in C$ in a standard way, see e.g. \cite{SchroderP10fossacs,SchroderPattinson10}.
\begin{defprop}
\label{defprop:st}
Define the \emph{coalgebraic standard translation} as 
\begin{align*}
\ST{x}(P) & \deq \inF{x}{P}, \\ 
\ST{x}(\heartsuit (\phi_1, \dots, \phi_n)) & \deq x \heartsuit \lcomp  x : \ST{x}(\phi_1) \rcomp \dots \lcomp x : \ST{x}(\phi_n) \rcomp,\\ 
\ST{x}(\bot) & \deq \bot,\\ 
\ST{x}(\phi\to\psi) & \deq \ST{x}(\phi)\to\ST{x}(\psi).
\end{align*} 
%
Then for any $\phi \in \MoF\SoV$ and any $\gM = (C, \gamma, I), \valu, c$, 
we have $\gM, c \vDash \phi$ iff $\gM, \valu[ x \mapsto c] \vDash \ST{x}(\phi)$.
\end{defprop}

For example, $\ST{x}(\PrL\PrL P) = x\PrL\cmh{x}{x\PrL\cmh{x}{P(x)}}$. This definition is more straightforward than the standard translation into FOL of modal logic over ordinary Kripke frames. Moreover, $\ST{x}$ uses only one variable from $\FoV$, namely $x$ itself. 
In fact, we can immediately observe that
\begin{prop}
\label{prop:modalsyntactic}
Whenever $\SoV$ consists entirely of \emph{unary} predicate symbols, the subset of $\phi \in \FoF(\SoV)$ obtained as the image of $\ST{x}$ for a fixed $x \in \FoV$ consists precisely of equality-free and quantifier-free formulas in the variable $x$.
\end{prop}
%


\subsection{Hybrid Languages} \label{sec:hybrid}

\noindent In this section, we establish the equivalence of CPL with the hybrid languages
$\HyFDG$ and $\HyFUA$. Both correspondences also
hold for ordinary predicate logic over relational structures
(FOL) and extend to CPL. We take this
as yet another indication that CPL is natural and
well-designed both as a generalization of FOL and ``the'' predicate
logic
cousin of existing coalgebraic formalisms.

\noindent This is our main, but not the only motivation. We progress towards this result step-by-step, extending the modal  language gradually
with new hybrid constructs. In this way, we reveal that a similar
correspondence exists between natural fragments of CPL and weaker
hybrid languages, most importantly between quantifier-free CPL and
$\HyFDA$. 

\noindent Again, obviously the correspondence between fragments of CPL and extensions of CML is tighter than in the case of FOL and ML only due to the modal flavour of CPL. However, results such as Corollary \ref{cor:hyfda} are useful spadework: any model-theoretic tool to be developed---say, a variant of E-F games---would be adequate for an extended coalgebraic modal formalism (e.g., $\HyFDA$) \textbf{iff} it is adequate for the corresponding fragment of CPL (e.g., the variable-free fragment), so we are free to work with whichever formalism we find more convenient at a given moment. 
 The straightforward correspondence also provides a good starting point for an extension of research programme sketched in \cite{tencate2005:model}---see Remark \ref{rem:tencate} at the end of this section.

\noindent Given a supply of \emph{world variables} $\WoV$ that we are going to keep fixed and implicit---in fact, as stated below, \emph{near} identical to $\FoV$---we  define the following \emph{coalgebraic hybrid languages}

\medskip 

\begin{tabular}{>{$}r<{$}@{ \qquad $\phi, \psi ::=$ \; }>{$}l<{$}@{\hspace{1ex}$\mid$\hspace{1ex}}>{$}l<{$}} 

\rule{0pt}{3ex}

\HyFDA &    z \mid P \mid \botF \mid \phi \impF \psi \mid 
\heartsuit (\phi_1, \dots, \phi_n) \mid @_z\phi & \daa z.\phi\\
\HyFDG &    z \mid P \mid \botF \mid \phi \impF \psi \mid 
\heartsuit (\phi_1, \dots, \phi_n) \mid \gloM\phi & \daa z.\phi \\
\HyFUA &    z \mid P \mid \botF \mid \phi \impF \psi \mid 
\heartsuit (\phi_1, \dots, \phi_n) \mid @_z\phi & \forall z.\phi  \\
\end{tabular}

\medskip

\medskip\noindent
where $z \in \WoV$. We refer the reader to, e.g, \cite{SchroderPattinson10,BlackburnC06sl,tencate2005:model} 
for the semantics. The extension of the standard translation to
these formalism is unproblematic in some cases, just like in the case of ordinary hybrid logic over Kripke frames:
%
$$
\ST{x}(z)  \deq x \eqF z, \quad 
\ST{x}(\gloM\phi)  \deq \forall x.\ST{x}(\phi), \quad
\ST{x}(\forall z.\phi)  \deq \forall z.\ST{x}(\phi).
$$
%
\noindent
One is tempted to put forward also
%
$$
\ST{x}(@_z\phi)  \deq \ST{x}(\phi)[z/x], \quad
\ST{x}(\daa z.\phi) \deq \ST{x}(\phi)[x/z].
$$
%
%
However, with other clauses remaining the same, this would
violate our convention that $[z/x]$
 is used only when $z$ is \emph{substitutable} for $x$; we would need
 to interpret it as \emph{capture-avoiding} substitution. Sadly, this in turn would entail
forsaking the luxury of using just one designated variable
for comprehension. Guillame Malod
(see \cite{Cate:2005:CHL}) observed that if we restrict the supply of variables, a translation along the above lines---indeed first proposed in the literature, which also goes to show that the present discussion is less trivial than it might seem---would fail even when embedding the hybrid logic over Kripke frames in the two-variable fragment of FOL. Malod's counterexample used nesting of modalities
of level two, but as our translation uses just one designated
variable,  $\ST{}$ would go wrong already on formulas of depth one. Just
consider $\ST{x}(\daa z.\Diamond z)$: were we careless about capture of bound variables, we would obtain $x \Diamond \lcomp x : x = x \rcomp$,
which is a formula with a completely different meaning. 
 There are two ways out. First is to redefine

\begin{align}
\STwr{x}(@_z\phi) & \deq \forall x. (x = z \to \ST{x}(\phi)),\\
\STwr{x}(\daa z.\phi) & \deq \forall z. (x = z \to \ST{x}(\phi)).\label{eq:stwr}
\end{align}
%


The second is to keep $\ST{}$ for hybrid formulas as defined above and change the modal clause instead: 

\begin{equation} \label{eq:stcompr}
\ST{x}(\heartsuit (\phi_1, \dots, \phi_n)) \deq x \heartsuit \lcomp  y: \ST{y}(\phi_1) \rcomp \dots \lcomp y: \ST{y}(\phi_n) \rcomp,\\ 
\end{equation}

\noindent
where $y$ is the first (in some fixed enumeration)  variable \emph{not used} in $\ST{x}(\phi_1), \dots, \ST{x}(\phi_n)$; by \emph{not used} here we mean both free and bound usage. Furthermore, to ensure that the translation works correctly, we have to assume that \emph{neither $x$ nor $y$} appears in $\WoV$. While the requirement to use more bound variables can be cumbersome---particularly for infinite sets of formulas---we prefer this option, as it makes it easier to characterize weaker hybrid languages as suitable syntactic fragments of CPL.

We can now state a generalization of both Proposition \ref{defprop:st} and corresponding results from the hybrid logic literature---see, e.g., \cite{BlackburnC06sl} for references:

\begin{prop} \label{prop:sthybrid} For any hybrid formula $\phi$ and any
  $\gM = (C, \gamma, I), \valu, c$, we have
  $\gM, \valu, c \vDash \phi$ iff
  $\gM, \valu[ x \mapsto c] \vDash \ST{x}(\phi)$.
\end{prop}
\noindent
As is well-known in the hybrid logic community---see again
\cite{BlackburnC06sl} for references---there is also a translation in
the reverse direction for sufficiently expressive hybrid
languages. This also generalizes to our setting, see Table
\ref{tab:ht}.

\begin{table}
\hrule
\caption{\label{tab:ht}Coalgebraic Hybrid Translation from quantifier-free CPL to $\HyFDA$}

\vspace{1mm}

\begin{tabular}{>{$}r<{$}@{ \deq \; }>{$}l<{$}@{\hspace{1cm}}>{$}r<{$}@{ \deq \; }>{$}l<{$}}
\HT(\inF{x}{P}) & @_xP & \HT(x \eqF y) & @_xy\\[0.7mm] 
\HT(\botF) & \botM & \HT(\phi \impF \psi) & \HT(\phi) \impM \HT(\psi) \\[0.7mm] 
\multicolumn{4}{c}{$\HT(x\PrL\lcomp y_1: \phi_1\rcomp \dots\lcomp y_n : \phi_n \rcomp) \quad \deq \quad @_x\PrL (\daa y_1.\HT(\phi_1),\dots,\daa y_n.\HT(\phi_n))$
} 
\end{tabular}

\vspace{1mm}
\hrule
\end{table}

\begin{prop} \label{prop:htda} For any $\phi \in \FoF$ and any
  $\gM = (C, \gamma, I), \valu, c$, we have
\begin{center}
  $\gM, \valu, c \vDash \HT(\phi)$ iff
  $\gM, \valu[ x \mapsto c] \vDash \phi$.
\end{center}
\end{prop}
\noindent
Combining Propositions \ref{prop:htda} and \ref{prop:sthybrid}, we
get:
\begin{cor} \label{cor:hyfda} Whenever $\SoV$ consists purely of unary
  predicates \bro and no function symbols\brc, $\HyFDA$ is expressively
  equivalent to the quantifier-free fragment of $\FoF$, assuming
  $\FoV$ contains $\WoV$ plus a disjoint infinite supply of additional
  individual variables \bro used for comprehension\brc.
\end{cor}

\begin{rem}[Quantifier-free CPL as the bounded fragment of FOL]
  In the case of ordinary FOL, the fragment equivalent to $\HyFDA$ is
  characterized as the \emph{bounded fragment}, see, e.g.,
  \cite{ArecesTenCate07}. In fact, our formula
  $x \hearts \lcomp y: \phi \rcomp$, despite being quantifier-free on
  the surface, can be described as a form of bounded
  quantification. This can be formalized as a result stating that over
  coalgebras for the covariant powerset functor (Kripke frames),
  quantifier-free $\FoF$ is equivalent to the bounded-fragment of
  ordinary FOL, where the role of $\hearts$ in $\FoF$ is played by the
  binary relation symbol $R$ in FOL; details are left to the reader.
\end{rem}

\begin{rem}[Chang's original syntax]
  As already mentioned, our syntax is slightly different to the
  original one proposed by Chang \cite{Chang73}. In that paper,
  there were no explicit comprehension variables and even in the
  enriched syntax which allowed constants and function terms, the term
  on the left-hand side of $\hearts$ had to be a variable. This
  variable was reused then on the right side of $\hearts$ as the
  comprehension variable. In other words, Chang's $x \hearts \phi(x)$
  was equivalent to ours $x \hearts \lcomp x: \phi(x) \rcomp$. In
  presence of quantifiers, which can be used to simulate the effect of
  capture-avoiding substitution as in $\STwr{}$ (this trick in fact
  stems back to Alfred Tarski), the two languages are obviously
  equivalent. But when considering fragments, as we do here, the
  equivalence breaks down; without quantifiers, Chang's syntax does
  not allow (\ref{eq:stwr}) and simple renaming of the comprehension
  variable on the right-hand side of $\hearts$ as in
  (\ref{eq:stcompr}) is not possible either.
\end{rem}
\noindent
There are two usual routes in hybrid logic to achieve full first-order
expressivity. One is to add universal quantifiers over $\WoV$ in
presence of the satisfaction operator $@$. The other is to add the
global modality $\gloM$ in presence of the downarrow binder
$\daa$\,. The hybrid translation is extended then as follows:
\begin{align*}
\HT_{\forall@}(\forall x.\phi) & \deq \forall x.\HT(\phi) \\
\HT_{\gloM\daa}(\forall x.\phi) & \deq \daa y.\gloM\daa x.\gloM(y \to \phi)
\end{align*}
In $\HT_{\gloM\daa}$ we need the proviso that $y$ is not occurring in the whole formula. 

\begin{thm} \label{thm:expequivalence}
$\HyFDG$, $\HyFUA$ and $\FoF$ are expressively equivalent.
\end{thm}
\noindent
As we can use $\STwr{x}$ now and keep reusing $x$ as the comprehension variable, it is enough to assume that $\FoV = \WoV \cup \{x\}$.
Since $@_z\phi$ is definable in presence of $\gloM$ (as $\gloM(z \to \phi)$), $\daa$ is definable by the universal quantifier over $\WoV$ (as $\forall z.(z \to \phi)$) and $\gloM$ is definable by combination of $\forall$ and $@$ (as $\forall y.@_y\phi$, where $y$ is not used in $\phi$), we get in fact seven equivalent languages: $\FoF$, Chang's original language,  $\HyFDG$, $\HyFUA$, $\HyFDG$ with $@$, $\HyFUA$ with $\daa$ and the 
jumbo hybrid language with all connectives introduced above.

\begin{rem} \label{rem:tencate}
  The equivalences stated here extend to the case of hybrid languages and $\FoF$ enriched with quantification over predicates (i.e., second-order languages). It would be interesting to follow more thoroughly the program of \emph{coalgebraic abstract model theory} both above and below $\FoF$. See Ten Cate's PhD Thesis \cite{tencate2005:model} for spadework in abstract model theory below first-order logic.
\end{rem}

\medskip







\subsection{Semantic Correspondence: The Van Benthem-Rosen Theorem} \label{sec:vbr}


Our Proposition \ref{prop:modalsyntactic} provides a \emph{syntactic} characterization of the modal fragment of our language. In a companion paper \cite{SchroderPL15:jlc}, we develop a \emph{semantic}, Van Benthem-Rosen style characterization. To compare these two characterizations, let us briefly recall the details. 

In the context of standard Kripke models, expressiveness of modal logic is characterized by van Benthem's theorem: modal logic is the bisimulation invariant fragment of
first-order logic 
in the corresponding signature. 
 The finitary analogue of this theorem \cite{Rosen97} states that every
formula that is bisimulation invariant \emph{over finite models} is
equivalent \emph{over finite models} to a modal formula. In the
coalgebraic context, replace bisimilarity with behavioural equivalence
\cite{Staton11:lmcs}. Moreover, we need to assume that the language
has `enough' expressive power; e.g., we cannot expect that
bisimulation invariant formulas are equivalent to CML formulas over
the empty similarity type. 
This is made precise as follows:

\begin{defi}
A $\Lambda$-structure is \emph{separating} if, for every set
$X$, every element $t \in TX$ is uniquely determined by the set
$\lbrace (\hearts, A) \mid \hearts \in \Lambda \mbox{ $n$-ary}, A
\in \Pow(X)^n, t \in \lsem \hearts \rsem_X(A) \rbrace$.
\end{defi}

\noindent
Separation is in general a less restrictive condition than those we needed for completeness proofs. In particular, separation automatically obtains for Kripke semantics. 
It was first used to establish the Hennessy-Milner property
for coalgebraic logics \cite{Pattinson04,Schroder08}. 



\begin{thm}[\cite{SchroderPL15:jlc}] \label{th:vanbenthem}
Suppose that the structure is separating and $\phi(x)$ is a
CPL formula with one free variable. Then $\phi$ is invariant
under behavioural equivalence \bro over finite models\brc\ iff it is equivalent to an infinitary CML formula with finite modal rank \bro over finite models\brc.
\end{thm}


\noindent
If we deal with finite similarity types only, the conclusion can be
strengthened:

\begin{thm}[\cite{SchroderPL15:jlc}] \label{th:vbsep}
Suppose that the structure is separating, $\Lambda$ is finite  and $\phi(x)$ is a
CPL formula with one free variable. Then
$\phi$ is invariant under behavioural equivalence \bro over finite
models\brc\ iff $\phi$ is equivalent to a \textbf{finite} CML formula \bro over finite
models\brc. 
\end{thm}

In fact, we can combine Theorem \ref{th:vbsep} with the syntactic characterization of Proposition \ref{prop:modalsyntactic} to obtain

\begin{cor}
Whenever $\SoV$ consists entirely of \emph{unary} predicate symbols 
 and the structure is separating,  the behaviourally-invariant \bro over finite structures\brc\ formulas of $\FoF$ in one-free variable are up to equivalence \bro over finite structures\brc\ precisely the equality-free and quantifier-free formulas in the single-variable fragment of $\FoF$.
\end{cor}


\section{First Steps in Coalgebraic Model Theory} \label{sec:modeltheory}

We proceed to outline the beginning of coalgebraic model theory,
taking a look at ultraproducts and the downwards L\"owenheim-Skolem
theorem. In the course of the technical development, we will import a
result on one-step cutfree complete rule sets established in earlier
work~\cite{Schroder06,PattinsonSchroder10}.

Recall that if $\ultra$ is an ultrafilter on an index set~$I$ and
$(X_i)$ is an $I$-indexed family of nonempty sets, then the \emph{ultraproduct}
$\prod_\ultra X_i$ is defined as
\begin{equation*}\textstyle
  \prod_\ultra X_i=\big(\prod_{i\in I} X_i\big)/\sim
\end{equation*}
where $\sim$ is the equivalence relation on $\prod_{i\in I}X_i$ defined by
\begin{equation*}
  (x_i)\sim(y_i)\iff \{i\in I\mid x_i=y_i\}\in\ultra.
\end{equation*}
One may regard $\ultra$ as a $\{0,1\}$-valued measure on $I$; under
this reading, the above definition says that $(x_i)$ and $(y_i)$ are
identified under $\sim$ if they are almost everywhere equal. We write
elements of $\prod_\ultra X_i$ and $\prod_{i\in I} X_i$ just as $x$,
omitting notation for equivalence classes and accessing the $i$-th
component as $x_i$.

Observe that if $X=\prod_\ultra X_i$ is an ultraproduct of sets and
$(A_i)$ is a family of subsets $A_i\subseteq X_i$, then
\begin{equation}\label{eq:ultraprod-pred}
  A=\textstyle\prod_\ultra A_i:=\{x\mid \{i\mid x_i\in A_i\}\in\ultra\}
\end{equation}
is a well-defined subset of $X$ (this is in fact just the way unary
predicates are standardly extended from the components to the
ultraproduct). Subsets of ultraproducts that are of this form are
called \emph{admissible}.
\begin{lem}\label{lem:fin-admissible}
 All finite subsets of ultraproducts are admissible. 
\end{lem}
\noindent Ultraproducts of coalgebras will not be determined uniquely;
instead, we give a property-oriented definition and later show
existence.

\begin{defi}[Quasi-Ultraproducts of Coalgebras]\label{def:ultra} 
  Let $(C_i)=(X_i,\xi_i)_{i\in I}$ be a family of $T$-coalgebras, and
  let $\ultra$ be an ultrafilter on $I$. A coalgebra $\xi$ on the set-ultraproduct $X=\prod_\ultra X_i$ is called a
  \emph{quasi-ultraproduct} of the $C_i$ if for every family $(A_i)$
  of subsets $A_i\subseteq X_i$, every $x\in\prod_\ultra X_i$, and
  every $\hearts\in\Lambda$,
  \begin{equation}\label{eq:ultra}
    \xi(x)\in\Sem{\hearts}_X\textstyle\prod_\ultra A_i \iff 
\{i\in I\mid \xi_{i}(x_{i}) \in \lsem \PrL \rsem_{C_{i}} (A_{i})\}
\in\ultra.
  \end{equation}
\end{defi}
\noindent The notion of quasi-ultraproduct extends naturally to
coalgebraic models using the standard definition to extend the
interpretation of predicates (as indicated above,
Equation~\eqref{eq:ultraprod-pred} recalls the case of unary
predicates).

The definition of quasi-ultraproducts is designed in such a way that
\L{}o\'{s}'s theorem, which in the view of ultrafilters as
$\{0,1\}$-valued measures states that the ultraproduct satisfies
exactly those formulas that hold in almost all its components, extends
to coalgebras:
\begin{thm}[Coalgebraic \L{}o\'{s}'s Theorem] \label{th:clos}
If $\gM=(C,\gamma,V)$ is a
  quasi-ultraproduct of $\gM_i=(C_i,\gamma_i,V_i)$ for the ultrafilter $\ultra$,
  then for every tuple $(a^1,\dots,a^n)$ of states in $C$, where
  $a^k=(a^k_i)_{i\in I}$, and for every CPL formula $\phi(x_1,\dots,x_n)$,
    $C\models\phi(a^1,\dots,a^n) \iff \{i\mid C_i\models\phi(a^1_i,\dots,a^k_i)\}\in\ultra$.
\end{thm}



\begin{proof}
  Induction over formulas. The cases for Boolean operators and
  quantifiers are as in the classical case, and the case for modal
  operators is exactly by the quasi-ultraproduct property. 
\end{proof}

\noindent
From this theorem, we obtain the usual applications, in particular
compactness (the latter with literally the same proof as in the
classical case).
The question is, of course, when quasi-ultraproducts
exist. A core observation is
\begin{lem}\label{lem:loc-finsat-full}
  In the notation of Definition~\ref{def:ultra}, the demands placed on
  $\xi(x)$ by Condition~\eqref{eq:ultra} constitute a finitely
  satisfiable set of one-step formulas.
\end{lem}
\noindent In the proof of this lemma, and on several further
occasions, we will need the fact that the set of all one-step sound
one-step rules is \emph{one-step cutfree
  complete}~\cite{PattinsonSchroder10}. Instead of repeating the
definition of this term, we state the relevant property directly.  
\begin{lem}\label{lem:os-cutfree-complete}
  \cite[Theorem~18 with proof]{Schroder06} Let $\Phi$ be a finite
  subset of $\Lambda(\Pow(C))\cup\neg\Lambda(\Pow(C))$, where
  $\neg\Lambda(\Pow(C))=\{\neg\hearts A\mid \hearts
  A\in\Lambda(\Pow(C))\}$.
  If $\Phi$ is one-step unsatisfiable, then there exists a sound
  one-step rule $\svA/\svX$ and a valuation $\tau:\SchV\to\Pow(C)$
  such that $C\models\svA\tau$ and $\svX\tau=\neg\Land\Phi$.
\end{lem}

\noindent We then proceed as follows with the open proof of
Lemma~\ref{lem:loc-finsat-full}:
\begin{proof}[Proof (Lemma~\ref{lem:loc-finsat-full})]
  Fix finitely many instances of (\ref{eq:ultra}) (keeping the same
  notation) for families of sets $(A^j_i)_{i\in I}$ and sets
  $A^j=\prod_\ultra A^j_i$, $j=1,\dots,k$. We regard these sets as
  extensions of unary predicates $P^j$ over the $X_i$ and over $X$,
  respectively. If the corresponding instances of (\ref{eq:ultra}) do
  not have a solution $\xi(x)$ in $TX$, then by
  Lemma~\ref{lem:os-cutfree-complete} we have a sound one-step rule
  $\svA/\svX$ and a valuation $\tau:\SchV\to\Pow(X)$ such that
  $X\models\svA\tau$ but the instances of (\ref{eq:ultra}) for
  $A^1,\dots,A^k$ demand $\xi(x)\models\neg\svX\tau$. Let
  $\sva_1,\dots,\sva_k$ be the schematic variables appearing in
  $\svA/\svX$; w.l.o.g.\ $\tau(\sva_j)=A^j$ for $j=1,\dots,k$. Then
  $X$ satisfies the first-order sentence $\forall z.(\svA \sigma)$
  where $\sigma(\sva_j)=P^j(y)$. By \L{}o\'{s}'s theorem (in fact already
  by its classical version), there exists $B\in\ultra$ such that
  $X_i\models\forall z.(\svA\sigma)$ and hence $X_i\models\svA\tau_i$
  for all $i\in B$, where $\tau_i(\sva_j)=A^j_i$. By one-step
  soundness of $\svA/\svX$ this implies $TX_i\models\svX\tau_i$ for
  all $i\in B$. But our formulation above that the instances of
  (\ref{eq:ultra}) for $A^1,\dots,A^k$ demand
  $\xi(x)\models\neg\svX\tau$ means more explicitly (and using the
  fact that $\ultra$ is an ultrafilter) that
  $\{i\in I\mid\xi_i(x_i)\models\neg\svX\tau_i\}\in\ultra$, so that we
  have a contradiction. 
\end{proof}

\noindent
From Lemma \ref{lem:loc-finsat-full}, our first existence criterion
for quasi-ultraproducts is immediate: 
\begin{thm}\label{thm:ultra-s1sc}
  If a $\Lambda$-structure is one-step compact
  \bro Definition~\ref{def:os-compact}\brc, then it has quasi-ultraproducts.
\end{thm}
\begin{exa}
  The above criterion applies in particular to all neighbourhood-like
  logics. It thus subsumes Chang's original ultraproduct
  construction~\cite{Chang73}
\end{exa}
\noindent Like for our completeness results, an alternative is to
require bounded operators:
\begin{thm}\label{thm:ultra}
  If a $\Lambda$-structure is finitary one-step compact
  \bro Definition~\ref{def:os-compact}\brc\ and all its operators are bounded,
  then it has quasi-ultraproducts.
\end{thm}

\noindent The proof needs the following lemma.
\begin{lem}\label{lem:loc-finsat-fin}
  Let $(C_i)=(X_i,\xi_i)_{i\in I}$ be a family of $T$-coalgebras, and
  let $\ultra$ be an ultrafilter on $I$. Let $X$ be the ultraproduct
  $\prod_\ultra X_i$, and let $x\in X$. Then the set
  \begin{equation*}
    \Psi=\{ \epsilon\hearts\{y^1,\dots,y^k\}\mid \{i\mid\xi_i(x_i)\models
    \epsilon\hearts\{y^1_i,\dots,y^k_i\}\}\in\ultra\}
  \end{equation*}
  of one-step formulas \bro where $\hearts$ ranges over $\PrLSet$, the
  $y^i$ range over $X$, and $\epsilon$ stands for either negation or
  nothing\brc\ is finitely satisfiable.
\end{lem}
\begin{proof}
  Analogous to Lemma~\ref{lem:loc-finsat-full}, using atoms of the
  form $z=c$ in place of unary predicates. In more detail: if a finite
  subset $\Phi$ of $\Psi$ is one-step unsatisfiable, then by
  Lemma~\ref{lem:os-cutfree-complete} there exist a one-step sound
  rule $\svA/\svX$ and a valuation $\tau$ such that
  $\svX\tau=\neg\Land\Phi$ and $X\models\svA\tau$. Now
  $X\models\svA\tau$ is semantically equivalent to
  $X\models\forall z.\svA_0$ where $\svA_0$ is propositional formula
  over atoms of the form $z=c$, where $c$ ranges over constants
  denoting elements of the involved finite subsets of $X$ in an
  extended first-order structure based on $X$. Then
  $\{i\mid X_i\models\forall z.\,\svA_0\}\in\ultra$ by (the classical
  version of) \L{}o\'s's theorem, where we interpret constants in
  $X_i$ by taking the $i$-th component of the interpretation in
  $X$~(this is just the way the interpretation of constants in the
  factors relates to that in the ultraproduct, classically). Hence
  $\{i\mid X_i\models\svA\sigma_i\}\in\ultra$, where $\sigma_i$
  replaces $\{y^1,\dots,y^k\}$ with $\{y^1_i,\dots,y^k_i\}$, and hence
  $\{i\mid TX_i\models\svW\}\in\ultra$, contradiction as in
  Lemma~\ref{lem:loc-finsat-full}.
\end{proof}
\begin{proof}[Proof (Theorem~\ref{thm:ultra})]
  By Lemma~\ref{lem:loc-finsat-fin} and finitary one-step compactness,
  there exists $\xi(x)$ satisfying the set $\Psi$ from
  Lemma~\ref{lem:loc-finsat-fin}. To show~\ref{eq:ultra} for
  $A\subseteq X$, we regard $A$ as the extension of a unary predicate
  $P$. Then $\xi(x)\models\hearts A$ is equivalent to
  \begin{equation*}
    x\models\exists y^1,\dots,y^k.\,(P(y^1)\land\dots\land P(y^k)\land x
    \hearts\cmh{z}{z=y^1,\dots,z=y^k}).
  \end{equation*}
  Thus it suffices to prove the \L{}o\'s equivalence for formulas
  $x \hearts\cmh{z}{z=y^1,\dots,z=y^k}$. This, however, is exactly
  what satisfaction of $\Psi$ by $\xi(x)$ guarantees.
\end{proof}



\noindent For operators that are $\omega$-bounded but not $k$-bounded
for any $k$, the ultraproduct construction cannot be available, in
consequence of Theorem~\ref{th:noncompact}. However, the downward
L\"owenheim-Skolem theorem does survive under the weaker assumption of
$\omega$-boundedness:

\begin{thm}[Downward L\"owenheim-Skolem Theorem]\label{thm:dls}
  Over $\omega$-bounded finitary one-step compact
  $\Lambda$-structures, $\CPL(\PrLSet,\SoV)$ satisfies the downward
  L\"owenheim-Skolem theorem; that is, every model of infinite
  cardinality $\kappa$ has, for every infinite cardinal
  $\lambda\le\kappa$, an elementary substructure of cardinality
  $\lambda$.
\end{thm}
\noindent Here, we use the term \emph{elementary substructure} in the
usual way to designate first-order substructures whose elements
satisfy the same formulas as they do in the original model; we
explicitly do \emph{not} require that the coalgebra structure on the
substructure forms a subcoalgebra.

 \newcommand{\xhphi}{x\PrL\cmh{y}{\phi}}

The proof needs the following simple lemma.
\begin{lem}\label{lem:prop-fin}
  Let $Y$ be an infinite subset of $X$, $\tau:\SchV\to\Pfin(Y)$ and
  $\svA\in\BSch(\SchV)$. Then $Y\models\svA\tau$ iff
  $X\models\svA\tau$.
\end{lem}

\begin{proof}
  Only finitely many $\sva\in\SchV$ are relevant, so we can, for the
  rest of the proof, assume that $\SchV$ is finite. Define the
  $\tau$-valuation of $x\in X$ as the valuation $\kappa:\SchV\to 2$
  given by $\kappa(\sva)=\top$ iff $x\in\tau(\sva)$. Then the claim of
  the lemma is equivalent to saying that every $\tau$-valuation
  occurring in $X$ occurs also in $Y$. Now if $x\in X\setminus Y$,
  then the $\tau$-valuation of $x$ is everywhere false; this valuation
  occurs also in $Y$, as $\SchV$ and the $\tau(\sva)$ are
  finite. 
\end{proof}

\begin{proof}[Proof (Theorem~\ref{thm:dls})]
  Let $\gM=(C,\gamma,I)$ be a coalgebraic model of cardinality
  $\kappa$. Pick Skolem functions for all formulas $\exists x.\,\phi$
  as usual, and for every formula $x\PrL\cmh{y}{\phi}$ a finitely
  non-deterministic Skolem function
  $f_{\xhphi}:C^{\FV(\xhphi)}\to\Pfin(C)$ with the property that for
  every valuation $\eta\in C^{\FV(\xhphi)}$,
  $f_{\xhphi}(\eta)\finsubset\Sem{\phi}_{C,\eta}$ and
  \begin{equation*}
    C,\eta\models \xhphi\iff \gamma(\eta(x))\models
    \hearts f_{\xhphi}(\eta).
  \end{equation*}
  (Such a function $f_{\xhphi}$ exists because $\hearts$ is
  $\omega$-bounded.) Pick a countably infinite subset $Y_0\subseteq C$
  and let $Y$ be the closure of $Y_0$ under the Skolem functions, in
  the case of the non-deterministic Skolem functions $f_{\xhphi}$ in
  the sense that $f_{\xhphi}[Y]\subseteq Y$. Then $Y$ is countable: it
  consists of the possible values of countably many finitely
  non-deterministic finite Skolem terms.

  It remains to define a coalgebra structure $\zeta$ on $c\in Y$ in
  such a way that
  \begin{equation}\label{eq:ls-coherence}
    \zeta(c)\models \hearts A \iff \gamma(c)\models
    \hearts A
  \end{equation}
  for all $A\finsubset Y$; that is, we have to prove that the set
  \begin{equation*}
   \Psi \deq \{\epsilon\hearts A \mid \gamma(c)\models\epsilon
    \hearts A\}
  \end{equation*}
  of one-step formulas over $\Pfin(Y)$ is satisfiable over $Y$ (where
  $\hearts$ ranges over $\PrLSet$, $A$ ranges over $\Pfin(Y)$, and
  $\epsilon$ ranges over $\{\cdot,\neg\}$). By finitary one-step
  compactness, it suffices to prove that $\Psi$ is finitely
  satisfiable. Assume the contrary; then by
  Lemma~\ref{lem:os-cutfree-complete} there exists a sound one-step
  rule $\svA/\svX$ valuation $\tau:\SchV\to\Pow(Y)$ such that
  $Y\models\svA\tau$ and $\svX\tau$ propositionally contradicts some
  finite subset $\Psi_0$ of $\Psi$. By Lemma~\ref{lem:prop-fin},
  $C\models\svA\tau$, and hence $C\models\svX\tau$; therefore,
  $\Psi_0$ is unsatisfiable over $C$, in contradiction to the fact
  that $\gamma(c)$ satisfies $\Psi$ by construction.

  Since $\Psi$ is satisfiable, we have a coalgebra structure $\zeta$
  satisfying~\eqref{eq:ls-coherence}. It follows by induction over the
  formula structure that for every coalgebraic first-order formula
  $\phi$ and every valuation $v$ in $Y$,
  \begin{equation*}
    \gN,v\models\phi
    \quad\text{iff}\quad \gM,v\models\phi:
  \end{equation*}
  where $\gN=(Y,\zeta,J)$ and $J$ is the induced substructure obtained
  by restricting $I$ to $Y$. The Boolean cases are trivial. The case
  for existential quantification is as in the classical case. The case
  $\xhphi$ is as follows: $\gN,v\models \xhphi$ iff
  $\zeta(v(x))\models\hearts\Sem{\phi}_{\gN,v}
  =\hearts(\Sem{\phi}_{\gM,v}\cap Y)$
  (where the equality holds by induction) iff (by
  $\omega$-boundedness) $\zeta(v(x))\models\hearts A$ for some
  $A\finsubset\Sem{\phi}_{\gM,v}\cap Y$, equivalently
  $\gamma(v(x))\models\hearts A$ by~\eqref{eq:ls-coherence}. The
  latter implies $\gM,v\models \xhphi$ by monotonicity; conversely,
  $\gM,v\models \xhphi$ implies
  $\gamma(v(x))\models\hearts f_{\xhphi}(v)$ by construction,
  and $f_{\xhphi}(v)\finsubset\Sem{\phi}_{\gM,v}\cap Y$. 
\end{proof}


\begin{exa}
  The above version of the downward L\"owenheim-Skolem theorem applies
  to our main bounded examples (relational, graded, and positive
  Presburger modalities) as well as to probabilistic modalities over
  non-standard or zerodimensional subdistributions, respectively,
  which are $\omega$-bounded but not $k$-bounded for any~$k$
  (Examples~\ref{ex:hyperprob} and~\ref{ex:zerodis}).
\end{exa}
\noindent Finally, we note that the downward L\"owenheim-Skolem theorem
holds also for the one-step compact case; this is in mild
generalization of a corresponding result for the neighbourhood case
proved already by Chang~\cite{Chang73}.
\begin{thm}\label{thm:dlsn}
  Over one-step compact $\Lambda$-structures, $\CPL(\PrLSet,\SoV)$
  satisfies the downward L\"owenheim-Skolem theorem \bro in the same
  formulation as in Theorem~\ref{thm:dls}\brc.
\end{thm}

\begin{proof}
  Let $\Phi$ be a set of coalgebraic first-order formulas in
  $\CPL(\PrLSet,\SoV)$, and let $\gM=(C,\gamma,I)$ be such that
  $\gM\models\Phi$. Pick Skolem functions for all formulas
  $\exists x.\,\phi$ as usual, and for every one-step sound one-step
  rule $R=\svA/\svX$ fix a Skolem function that given an element
  $x\in C$ satisfying some instance of $\neg\svX$ picks an element of
  $C$ that satisfies the corresponding instance of $\neg\svA$. More
  precisely: let $\sigma:\FoV\to\CPL(\PrLSet,\SoV)$ be a
  substitution, let $x,y$ be variables with $x$ fresh, let
  $\svX^{x,y}\sigma$ be the formula obtained by replacing in $\svX$
  each modal operator application $\PrL \sva$ with
  $x\PrL\cmh{y}{\sigma(\sva)}$, and let $v$ be a valuation such that
  $\gM,v\models\neg\svX^{x,y}\sigma$. Then there exists a $y$-variant
  $v'$ of $v$ such that $\gM,v'\models\neg\svA\sigma$, and the Skolem
  function $f_{R,\sigma}$ assigns such a $v'(y)$ to
  $v|_{\FV(\svX^{x,y}\sigma)}$. As $\FV(\svX^{x,y}\sigma)$ is finite,
  $f_{R,\sigma}$ is a finitary function, so that closing a given
  subset $Y_0\subseteq C$ of cardinality $|Y_0|=\lambda$ under the
  Skolem functions yields a set $Y\subseteq C$ of the same cardinality
  $|Y|=\lambda$.

  The coalgebra structure $\zeta$ that we are to define on $Y$ has to
  satisfy the \emph{coherence} condition
  \begin{equation*}
    \zeta(c)\models\PrL(\Sem{\rho}^y_v\cap Y)\text{ iff }
    \gamma(c)\models\PrL\Sem{\rho}^y_v
  \end{equation*}
  for all $c\in Y$, all formulas $\rho$, and all valuations $v$ in
  $Y$, where the second condition is by definition equivalent to
  $c\in\Sem{x\PrL\cmh{y}{\rho}}^x_v$. Once this is established, we can
  show as usual that $\gN=(Y,\zeta,J)$, with $J$ interpreting $\Sigma$
  by restricting $I$ to $Y$, is an elementary substructure of
  $(C,\gamma,I)$, and we are done.

  Now assume that $\zeta(c)$ as required fails to exist, which means
  that the set $\Phi$ of constraints of the form
  $\epsilon\PrL(\Sem{\rho}^y_v\cap Y)$ (where $\epsilon$ stands for
  either nothing or negation) on $\zeta(c)$ determined by the
  coherence condition is one-step unsatisfiable. By one-step
  compactness, already some finite subset $\Phi_0$ of $\Phi$ is
  unsatisfiable. By Lemma~\ref{lem:os-cutfree-complete} and the format
  of $\Phi$, there exist a sound one-step rule $R=\svA/\svX$ and a
  substitution $\sigma:\SchV\to\CPL(\PrLSet,\SoV)$ such that
  $Y\models\svA\tau$ and $\svX\tau=\neg\Land\Phi_0$ where
  $\tau(\sva)=\Sem{\sigma(\sva)}_v^y\cap Y$. However, by construction
  of~$\Phi$ we have $\gamma(c)\in\Sem{\neg\svX}\hat\tau$, where
  $\hat\tau$ is the $\Pow(C)$-valuation sending $\sva\in \SchV$ to
  $\Sem{\sigma(\sva)}^y_v$, and hence
  $\gM,v\models\neg\svX^{x,y}\sigma$ where $x$ is fresh and we assume
  w.l.o.g.\ that $v(x)=c$. Then
  $f_{R,\sigma}(v|_{\FV(\svX^{x,y}\sigma)})\in\Sem{\neg\svA}\tau$,
  in contradiction to $Y\models\svA\tau$. 
\end{proof}
\begin{exa}
  Besides the plain neighbourhood case, Theorem~\ref{thm:dlsn} covers
  all instances of CPL defined by imposing rank-1 frame conditions on
  neighbourhood frames, e.g.\ CPL over monotone neighbourhood frames
  and various deontic logics.
\end{exa}







\section{Proof theory} \label{sec:proof}

\subsection{Sequent system for CPL} \label{sec:axiom}

In \S~\ref{sec:completeness}, we have seen 
a complete Hilbert calculus for coalgebraic predicate logic. 
The present goal is a 
 cut-free, complete sequent calculus.
Our basis is the system $\mathsf{G1c}$ of \cite{Troelstra:1996:BPT} that we extend with modal
rules describing the (fixed) $\Lambda$-structure. Our 
 treatment of equality, on the other hand, is inspired by Kanger~\cite{Kanger1957}, Degtyarev and Voronkov~\cite{DegtyarevVoronkov2001} and Seligman \cite{Seligman01:jlc}. 
 In fact, the syntactic cut-elimination proof presented here is based on Seligman's ideas. 

 
 We take \emph{sequents} to be pairs $(\Gamma, \Delta)$, written $\Gamma \To \Delta$ 
where $\Gamma, \Delta \subseteq \Lang$ are finite multisets. 
The sequent calculus for coalgebraic predicate logic contains four types of rules: the standard logical and structural rules for first-order logic, rules for equality and rules for the modal operators. The logical rules are standard as in Table \ref{tab:fo-g-rules}. The formula introduced in the conclusion of a logical rule is called the \emph{principal} formula of the rule. 
 This applies, in particular, to the structural rules in Table \ref{tab:fo-g-rules}: the formula $\phi$ in the conclusion is the principal one. Note that, somewhat counterintuitively, in the equality rules the formula $x = y$ in the conclusion is the \emph{context}, i.e., the only  \emph{non-principal} formula and all the remaining ones are \emph{principal}!  


To account for the modal operators, we incorporate the one-step rules
$\Rules$ into the sequent system. In principle, we just generate a
sequent rule within CPL from every one-step rule in~$\Rules$. Only for
presentational purposes, we factor this process through an alternative
modal rule format where propositional operators are fully dissolved
into sequents:
\begin{defi}
  A rule
  \begin{equation*}
    \vcenter{\infer{\hearts_1 \vec{\sva}_1, \dots, \hearts_n \vec{\sva}_n \To \hearts_{n+1}
	\vec{\sva}_{n+1}, 
	\dots, \hearts_{n+m} \vec{\sva}_{n+m}}{\Gamma_1 \To \Delta_1  \cdots  \Gamma_k \To
	\Delta_k}}
  \end{equation*}
  \emph{represents} a one-step rule $\svA/\svX$ \emph{in sequent
    format} if $\svA$ is propositionally equivalent to
  $\Land_{i=1}^k((\Land\Gamma_i)\to(\Lor\Delta_i)$, and $\svX$ is
  propositionally equivalent to
  $(\Land_{j=1}^n\PrL_j\vec\sva_j)\to(\Lor_{j=n+1}^m\PrL_j\vec\sva_j)$. We
  transfer the existing syntactic restrictions on one-step rules
  according to Definition~\ref{def:hilonestep} to this format by
  requiring that every schematic variable occurring in the premise
  occurs also in the conclusion, and every schematic variable occurs
  at most once in the conclusion.
\end{defi}
\noindent It is clear that every one-step rule can be represented in
sequent format (just transform the premise into conjunctive normal
form and then trivially translate disjunctive clauses into sequents in
both premise and conclusion). Subsequently, we generate a sequent rule
$\Seq(R)$ in CPL syntax, adding weakening contexts $\Sigma$, $\Theta$
to both the conclusion and all the premises:
\[
\infer{\Sigma, z \hearts_1 \boldsymbol{\lcomp x_{1} \col \phi_1 \rcomp}, \dots, z \hearts_n \boldsymbol{\lcomp x_{n} \col \phi_n \rcomp} \To \\ z \hearts_{n+1} \boldsymbol{\lcomp x _{n+1}\col \phi_{n+1} \rcomp}, \dots, z \hearts_{n+m} \boldsymbol{\lcomp x_{n+m} \col \phi_{n+m} \rcomp, \Theta}
}{
\Sigma, \Gamma_1 \sigma_{\boldsymbol{x}}^{y}  \To  \Delta_1 \sigma_{\boldsymbol{x}}^{y}, \Theta &\cdots &\Sigma, \Gamma_k \sigma_{\boldsymbol{x}}^{y}  \To  \Delta_k \sigma_{\boldsymbol{x}}^{y}, \Theta
}
\]
with syntactic details as summarized in Table \ref{tab:fo-g-rules}.
The formulas $z \hearts_i \boldsymbol{ \lcomp x \col \phi_i \rcomp}$
are the principal formulas of $\Seq(R)$.

\newcommand{\yfresh}{\mathbf{\dagger}_y}

\begin{table}[htbp]
\hrule

\small

\caption{Sequent System of Coalgebraic Predicate Logic}
\label{tab:fo-g-rules}

\vspace{\tbskip}

In all the rules below, $\yfresh$ means that $y$ is \textbf{fresh in the conclusion}.

\vspace{\tbskip}

Axioms
\[
\infer[\rul{Ax}]{\phi \To \phi}{}
\quad
\infer[\rul{L\bot}]{\bot \To }{}
\quad
\infer[\rul{R=}]{ \To x = x}{}
\]

\vspace{\tbskip}

Logical Rules
\[
\infer[\rul{R\to}]{\Gamma \To \Delta, \phi \to \psi}{\phi, \Gamma \To \Delta, \psi} \quad 
\infer[\rul{L\to}]{\phi \to \psi, \Gamma \To \Delta}{\Gamma \To \Delta, \phi & \psi, \Gamma \To \Delta}
\]
\[
\infer[\rul{R \forall}\yfresh]{\Gamma \To \Delta, \forall x. \phi}{\Gamma \To \Delta, \phi[y/x]} \quad 
\infer[\rul{L\forall}]{\forall x. \phi, \Gamma \To \Delta}{\phi[z/x], \Gamma \To \Delta}
\]

\vspace{\tbskip}

Equality Rules
\[
\infer[\rul{L=_{1}}]{x = y, \Gamma[y/z] \To \Delta[y/z]}{x = y, \Gamma[x/z] \To \Delta[x/z]} \quad 
\infer[\rul{L=_{2}}]{x = y, \Gamma[x/z] \To \Delta[x/z]}{x = y, \Gamma[y/z] \To \Delta[y/z]}
\]

\vspace{\tbskip}

Modal Rules $\Seq(\Rules)$: for every one-step rule $R \in \Rules$, 
\[
\infer[{\Seq(R)\yfresh}]{\deduce{z \hearts_{n+1} \boldsymbol{\lcomp x_{n+1}
	\col \phi_{n+1} \rcomp}, \dots, z \hearts_{n+m} \boldsymbol{\lcomp x_{n+m}
	\col \phi_{n+m} \rcomp}, \Theta}{\Sigma, z \hearts_1 \boldsymbol{\lcomp
	x_{1} \col \phi_1 \rcomp}, \dots, z \hearts_n \boldsymbol{\lcomp x_{n}
	\col \phi_n \rcomp} \To}}{
\Sigma, \Gamma_{1} \sigma_{\boldsymbol{x}}^{y}  \To \Delta_{1} \sigma_{\boldsymbol{x}}^{y}, \Theta  
&\cdots &\Sigma, \Gamma_{k} \sigma_{\boldsymbol{x}}^{y} \To \Delta_{k} \sigma_{\boldsymbol{x}}^{y}, \Theta } 
\]
where

\vspace{\tbskip}

\begin{itemize}
\item $R$ is represented in sequent format as \, $\vcenter{\infer{\hearts_1 \vec{\sva}_1, \dots, \hearts_n \vec{\sva}_n \To \hearts_{n+1}
	\vec{\sva}_{n+1}, 
	\dots, \hearts_{n+m} \vec{\sva}_{n+m}}{\Gamma_1 \To \Delta_1 & \cdots & \Gamma_k \To
	\Delta_k}}$ \vspace{\tbskip}
\item $\boldsymbol{\lcomp x_{i} \col \phi_i \rcomp} = \lcomp x^{1}_{i} \col \phi_i^1 \rcomp
\dots \lcomp x^{\arty\PrL}_{i} \col \phi_i^{\arty\PrL} \rcomp$ is a finite sequence of comprehension formulas according to $\arty\hearts_i$, and
\item the substitution $\sigma_{\boldsymbol{x}}^{y}$ sends $\sva_{i}^{j}$ to the formula $\phi_{i}^{j}[y/x_{i}^{j}]$ of $\Lang$. 
\end{itemize}

\vspace{\tbskip}

Structural Rules

\vspace{\tbskip}

\[
\infer[\rul{RW}]{\Gamma \To \Delta,  \phi}{\Gamma \To \Delta} \quad 
\infer[\rul{LW}]{\phi, \Gamma \To \Delta}{\Gamma \To \Delta}
\]
\[
\infer[\rul{RC}]{\Gamma \To \Delta, \phi}{\Gamma \To \Delta, \phi, \phi} \quad 
\infer[\rul{LC}]{\phi, \Gamma \To \Delta}{\phi, \phi, \Gamma \To \Delta}
\]

\vspace{\tbskip}

Cut Rule (optional)
\[
\infer[\rul{Cut}]{\Gamma,\Sigma \Rightarrow \Delta, \Theta}{\Gamma \Rightarrow \Delta, \phi & \phi, \Sigma \Rightarrow \Theta}
\]
\hrule
\vspace{1mm}
\end{table}

\begin{exa}
\label{ex:rulesets}
\noindent
Recall from Example \ref{ex:s1sc} that the rule set for the normal modal logic $\mathsf{K}$  consists of rules that are represented in sequent format as
\begin{equation*} 
\infer[\rul{K_n}]{\Diamond \sva \To \Diamond \svb_1, \dots,
\Diamond \svb_n}{\sva \To \svb_1, \dots, \svb_n}
\end{equation*}
for all $n \geqslant 0$. We obtain the following first-order version
\begin{equation*} 
\infer[\rul{\Seq(K_n)}\yfresh]{\Sigma, z \Diamond \lcomp x_{0} \col \phi_0 \rcomp \To z \Diamond
	\lcomp x_{1} \col \phi_1 \rcomp, \dots,  z \Diamond \lcomp x_{n}
	\col \phi_n \rcomp, \Theta}{
\Sigma, \phi_{0}[y/x_{0}] \To \phi_{1}[y/x_{1}],\dots,\phi_{n}[y/x_{n}], \Theta
}
\end{equation*}
(where $y$ is fresh in the conclusion) 
by the previous definition. Recall also  that  modal neighbourhood semantics is
axiomatised by a one-step rule that is represented in sequent format as
\begin{equation*} 
\infer[\rul{C}]{\Box \sva \To \Box \svb}{\sva \To \svb \quad \svb \To \sva},
\end{equation*}
which expresses that $\Box$ is a congruential operator. The first order version of $\rul{C}$  
 then reads
\begin{equation*}
\infer[{\Seq(\rul{C})\yfresh}]{\Sigma, z \Box \lcomp x_{0} \col \phi_0 \rcomp \To z \Box
	\lcomp x_{1} \col \phi_1 \rcomp, \Theta}{
\Sigma, \phi_{0}[y/x_{0}] \To \phi_{1}[y/x_{1}], \Theta & 
\Sigma, \phi_{1}[y/x_{1}] \To \phi_{0}[y/x_{0}], \Theta & }
\end{equation*}
(where $y$ is fresh in the conclusion) 
which 
 provides a complete and, as we are going to see below, cut-free
axiomatisation of Chang's original logic.
\end{exa}


\noindent We write $\Seq\Rules \entails \Gamma \To \Delta$ if
$\Gamma \To \Delta$ can be derived using the logical rules, equality
rules, and axiom rules of Table \ref{tab:fo-g-rules}, together with
the rules $\Seq(R)$ from Table \ref{tab:fo-g-rules} for every rule
$R \in \Rules$.We write $\Seq\Rules\Cut \entails \Gamma \To \Delta$ if
the \emph{cut rule} $\rul{Cut}$ of Table \ref{tab:fo-g-rules} is used
additionally.  If $\mathfrak{M} = (C, \gamma, I)$ is a first-order
model over a $\Lambda$-structure, we write
$\mathfrak{M}, \valu \models \Gamma \To \Delta$ if
$\mathfrak{M}, \valu \models \Land \Gamma \to \Lor \Delta$ and, as
usual $\mathfrak{M} \models \Gamma \To \Delta$ if
$\mathfrak{M}, \valu \models \Gamma \To \Delta$ for all variable
assignments $\valu$ and finally $ \models \Gamma \To \Delta$ if
$\mathfrak{M} \models \Gamma \To \Delta$ for all first-order models
$\mathfrak{M}$ over the corresponding structure, which we elide in
the notation.

\begin{prop}
\label{prop:soundness_seq_rule}
For any one-step rule $R \in \Rules$ and any model $\mathfrak{M}$ = $(C,\gamma,I)$, $\Seq(R)$ preserves the validity on $\mathfrak{M}$. 
\end{prop}

\begin{proof}
Let $R \in \Rules$ be represented in sequent format as $$\infer[]{\hearts_1 \vec{\sva}_1, \dots, \hearts_n \vec{\sva}_n \To \hearts_{n+1} \vec{\sva}_{n+1}, \dots, \hearts_{n+m} \vec{\sva}_{n+m}}{\Gamma_1 \To \Delta_1 & \cdots & \Gamma_k \To \Delta_k}.$$ By Assumption~\ref{conv:1ss}, $R$ is one-step sound. Let $\mathfrak{M}$ = $(C,\gamma, I)$ be a model. To show that $\Seq(R)$ preserves validity, assume that all of $\Sigma, \Gamma_i \sigma_{{\boldsymbol x}}^{y} \To \Delta_i  \sigma_{{\boldsymbol x}}^{y}, \Theta$ ($1\leqslant i \leqslant k$) are valid in $\mathfrak{M}$. Fix any variable assignment $v$ on $C$. To show that the conclusion of $\Seq(R)$ is true at $M,v$, assume that $M,v \models \bigwedge  \Sigma$ and $M,v \not\models \bigvee \Theta$. 
Our goal is to show that $$M,v \models \bigwedge_{1 \leqslant i \leqslant n} z \hearts_i \boldsymbol{\lcomp x_{i} \col \phi_i \rcomp} \to \bigvee_{1 \leqslant j\leqslant m} z \hearts_{n+j}\boldsymbol{\lcomp x_{n+j} \col \phi_{n+j}\rcomp},$$ i.e., 
\begin{center}
if $\gamma(v(z)) \in  \bigcap_{1 \leqslant i \leqslant n} \lsem \hearts_i \rsem_{C,v}(\lsem \phi_i^{1}  \rsem^{x_{i}^{1}}_{C,v}, \ldots,  \lsem \phi_i^{\arty\hearts_{i}}  \rsem^{x_{i}^{\arty\hearts_{i}}}_{C,v})$ \\
then $\gamma(v(z)) \in  \bigcup_{1 \leqslant j \leqslant m} \lsem \hearts_{n+j} \rsem_{C,v}(\lsem \phi_{n+j}^{1}  \rsem^{x_{n+j}^{1}}_{C,v}, \ldots,  \lsem \phi_{n+j}^{\arty\hearts_{n+j}}  \rsem^{x_{n+j}^{\arty \hearts_{n+j}}}_{C,v})$ . 
\end{center}

\noindent Let us define a valuation $\tau: \SchV \to \Pow(C)$ by $\tau(\boldsymbol{p_{i}^{j}})$ = $\lsem \phi_i^{j} [y/x_{i}^{j}] \rsem^{y}_{C,v}$. To show that $C, \tau \models  \bigwedge \Gamma_i \to \bigvee \Delta_i$ for all $1\leqslant i \leqslant k$, let us fix any $c \in C$. 
Since $y$ is fresh in the conclusion of $\Seq(R)$, it follows from $M,v \models \bigwedge  \Sigma$ and $M,v \not\models \bigvee \Theta$ that $M,v[c/y] \models \bigwedge  \Sigma$ and $M,v[c/y] \not\models \bigvee \Theta$. Then from our assumption of the validity of all premises of $\Seq(R)$ on a pair $(M,v)$, we obtain $M,v[c/y] \models \Gamma_i \sigma_{{\boldsymbol x}}^{y} \To \Delta_i  \sigma_{{\boldsymbol x}}^{y}$, which implies $c \in \tau(\bigwedge \Gamma_i \to \bigvee \Delta_i)$, as desired. Since $R$ is one-step sound, we have that $TC, \tau \models {\bigwedge}_{1\leqslant i \leqslant n} \hearts_i \vec{\sva}_i \to {\bigvee}_{1\leqslant j \leqslant m} \hearts_{n+j} \vec{\sva}_{n+j}$. Because $\tau(\boldsymbol{p_{i}^{j}})$ = $\lsem \phi_i^{j} [y/x_{i}^{j}] \rsem^{y}_{C,v}$ = $\lsem \phi_i^{j}  \rsem^{x_{i}^{j}}_{C,v}$ by freshness of $y$, we can conclude our desired implication above. 
\end{proof}

\noindent We show soundness and completeness of the sequent system $\Seq\Rules$ by translating into, and from, the Hilbert system $\Hilb\Rules$ which is known to be (semantically) complete when $\Rules$ is strongly one-step complete. Note that $\Hilb\Rules$ does not include the the \bdpl-axioms. 
Before showing that both systems $\Hilb\Rules$ and $\Seq\Rules\mathsf{Cut}$ have the same deductive power, we note one consequence of the congruence rule (Remark~\ref{rem:cong}) provided that the rules absorb congruence. We introduce the concept of absorption in a slightly more general form which will be used later. 

\begin{defi}
We say that a finite set $\mathbb{S}$ of sequents \emph{covers} a finite set $\mathbb{S}'$ of sequents if each element $\Gamma \Rightarrow \Delta$ of $\mathbb{S}'$ contains some element $\Pi \Rightarrow \Sigma$ of $\mathbb{S}$ in the sense that $\Pi \subseteq \Gamma$ and $\Sigma \subseteq \Delta$.  
We write $\mathbb{S} \rhd \mathbb{S}'$ if $\mathbb{S}$ covers $\mathbb{S}'$ where we identify sequents with singleton sets. A set $\Rules$ of rules \emph{absorbs} a rule $\Sigma_{1} \To \Theta_{1}, \cdots, \Sigma_{m} \To \Theta_{m}/ \Sigma \To \Theta$ if there exists a rule $R = \Gamma_{1}\To \Delta_{1}, \cdots, \Gamma_{n} \To \Delta_{n} / \Gamma_{R} \To \Delta_{R} \in \Rules$ such that 
\[
\{ \Sigma_{1} \To \Theta_{1}, \ldots, \Sigma_{m} \To \Theta_{m} \} \rhd 
\{ \Gamma_{1} \To \Delta_{1}, \ldots, \Gamma_{n} \To \Delta_{n}  \}
\]
and $\Gamma_{R} \To \Delta_{R} \rhd \Sigma \To \Theta$. 
A rule set \emph{absorbs congruence} if it absorbs the rule 
\[
\infer[\rul{\mathsf{Cong}\hearts}]{
\hearts(p_1, \ldots, p_n) \To \hearts(q_1, \ldots, q_n)
}{
p_1 \To q_1 & \cdots & p_n \To q_n &  q_1 \To p_1 & \cdots & q_n \To p_n 
} 
\]
and it 
\emph{absorbs monotonicity of $\hearts$ in the $i$-th argument} if the rule
\[
\infer[\Moni]{\hearts(p_1, \dots, p_n)
\To \hearts(p_1, \dots, p_{i-1}, q_i, p_{i+1}, \dots p_n)}{p_i \To q_i}
\]
is absorbed.
\end{defi}

\begin{lem}
\label{lem:ad_cong}
When $\Rules$ absorbs congruence, the following congruence rule
\begin{equation*}
\infer[\mbox{$\rul{Cong}\hearts$}]{\Sigma, z \hearts \boldsymbol{\lcomp x_{0} : \phi_{0} \rcomp} \To  z \hearts \boldsymbol{\lcomp x_{1} : \phi_{1} \rcomp}, \Theta}{
\{ \Sigma, \phi_{0}^{j}[y/{x_{0}^{j}}]  \To \phi_{1}^{j}[y/{x_{1}^{j}}], \Theta & \Sigma, \phi_{1}^{j}[y/{x_{1}^{j}}] \To \phi_{0}^{j}[y/{x_{0}^{j}}], \Theta \,|\, 1 \leqslant j \leqslant n \} 
}
\end{equation*}
$($where $y$ is fresh in the conclusion and $n$ is the arity of $\heartsuit$$)$ is admissible in $\Seq\Rules$ and $\Seq\Rules\Cut$.
\end{lem}

\begin{proof}
Since $\Rules$ absorbs congruence, we can find a one-step rule $R$ =  $\Gamma_{1}\To \Delta_{1}, \cdots, \Gamma_{m} \To \Delta_{m} / \Gamma_{R} \To \Delta_{R} \in \Rules$ such that $\{ p_{j} \To q_{j}, q_{j} \To p_{j} \,|\, 1 \leqslant j \leqslant n \} \rhd 
\{ \Gamma_{1} \To \Delta_{1}, \ldots, \Gamma_{m} \To \Delta_{m}  \}$ and $\Gamma_{R} \To \Delta_{R} \rhd \hearts(p_1, \ldots, p_n) \To \hearts(q_1, \ldots, q_n)$. Fix such a one-step rule $R$. 
Assume that $\Sigma, \phi_{0}^{j}[y/{x_{0}^{j}}]  \To \phi_{1}^{j}[y/{x_{1}^{j}}], \Theta$ and $\Sigma, \phi_{1}^{j}[y/{x_{1}^{j}}] \To \phi_{0}^{j}[y/{x_{0}^{j}}], \Theta$ are derivable in $\Seq\Rules$ for all $1 \leqslant j \leqslant n$. 
Let us define the substitution $\sigma_{\boldsymbol{x}}^{y}$ which sends each $\boldsymbol{p}^{j}_{i}$ to a formula $\phi^{j}_{i}[y/x_{i}^{j}]$, where $i$ = $0$ or $1$. 
Since $\{ p_{j} \To q_{j}, q_{j} \To p_{j} \,|\, 1 \leqslant j \leqslant n \} \rhd 
\{ \Gamma_{1} \To \Delta_{1}, \ldots, \Gamma_{m} \To \Delta_{m}  \}$, we can obtain the derivability of 
$\Sigma, \Gamma_{k}\sigma_{\boldsymbol{x}}^{y} \To \Delta_{k} \sigma_{\boldsymbol{x}}^{y}, \Theta$ ($1 \leqslant k \leqslant m$) in $\Seq\Rules$ with the help of weakening rules. Since $R \in \Rules$, the covering $\Gamma_{R} \To \Delta_{R} \rhd \hearts(p_1, \ldots, p_n) \To \hearts(q_1, \ldots, q_n)$ and the weakening rules allow us to obtain the derivability of $\Sigma, z \hearts \boldsymbol{\lcomp x_{0} : \phi_{0} \rcomp} \To  z \hearts \boldsymbol{\lcomp x_{1} : \phi_{1} \rcomp}, \Theta$ in $\Seq\Rules$, as required. 
\end{proof}

\noindent By our equality rules, the following lemma is immediate. 

\begin{lem}
\label{lem:arb_repl}
The replacement axiom $x=y, \phi[x/z] \To \phi[y/z]$ is derivable in $\Seq\Rules$. 
\end{lem}

\noindent
One direction of the translation between the two proof systems can
now be given as follows:
\begin{thm}
\label{thm:nb_h2g}
Suppose that $\Rules$ absorbs congruence and let $\Hilb\Rules \entails
\phi$. Then $\Seq\Rules\Cut \entails \To \phi$.
\end{thm}

\begin{proof}
First, we demonstrate admissibility of modus ponens in $\Seq\Rules\Cut$ by
\[
\infer[\rul{Cut}]{\To \psi}{
\To \phi
&
\infer[\rul{Cut}]{\phi \To \psi}{\To \phi \to \psi
&
\phi \to \psi, \phi \To \psi}
}
\]
where the derivability of $\phi \to \psi, \phi \To  \psi$ is easily established by $\rul{L \to}$. Note that this is the only place in this proof where we need $\rul{Cut}$. 
Hence, it suffices to show that all the axioms of $\Hilb\Rules$ 
(recall Table \ref{tab:fo-g-rules}) are derivable in cut-free $\Seq\Rules$. 
All of equality axioms \axref{En5}, \axref{En6}.1 and \axref{En6}.2 are derivable by the equality axiom $\rul{R=}$ and the equality rules $\rul{L=_{i}}$. 
Moreover, since this is easy to show for logical but non-modal axioms, 
we focus on \axref{Cong}, \axref{Alpha} and $\axref{Onestep($\Rules$)}$. Firstly, the derivability of \axref{Cong} follows from Lemma \ref{lem:ad_cong}. 
Secondly, for  \axref{Alpha} we have the following derivation:
\[
\infer[\rul{Cong}\hearts]{
z \hearts {\lcomp x_{0} : \phi_{0} \rcomp} \cdots {\lcomp x_{i} : \phi_{i} \rcomp} \cdots {\lcomp x_{n} : \phi_{n} \rcomp} \To  z \hearts {\lcomp x_{0} : \phi_{0} \rcomp}\cdots {\lcomp u : \phi_{i}[u/x_{i}] \rcomp} \cdots {\lcomp x_{n} : \phi_{n} \rcomp}
}{
\{ \phi_{j}[y/{x_{j}}]  \To \phi_{j} [y/{x_{j}}]  \,|\, j \neq i \} 
&
\phi_{i}[y/{x_{i}}] \To \phi_{i}[u/{x_{i}}][y/u]
&
\phi_{i}[u/{x_{i}}][y/u] \To \phi_{i}[y/{x_{i}}]
}
\]
where we note that $\rul{Cong}\hearts$ is admissible by Lemma \ref{lem:ad_cong}. 
All the premises are axioms since $u$ is assumed to be fresh in $\phi_{i}$.  
Finally, let us move to the provability of $\axref{Onestep($\Rules$)}$. 
Suppose that $R = \Gamma_1 \To \Delta_1, \dots, \Gamma_k \To \Delta_k / \Gamma_{R} \To \Delta_{R}$ is a one-step rule as in Definition \ref{def:hilonestep}. With the help of contraction rules, we note that the following are derivable rules in $\Seq\Rules$: for any finite multiset $\Theta$, 
\[
\infer[\rul{L \land}]{\bigwedge \Theta, \Gamma \To \Delta}{\Theta, \Gamma \To \Delta} \quad
\infer[\rul{R \lor}]{\Gamma \To \Delta, \bigvee \Theta}{\Gamma \To \Delta, \Theta}.
\]
We obtain the following derivation where $N$ = $\{1,...,n \}$, $M$ = $\{n+1,...,n+m\}$ and $\pi_{i}$ is an abbreviation of $(\Land \Gamma_i \to \Lor \Delta_i)\sigma$:
\begin{equation*}
\infer=[\rul{L \land}, \rul{R \land}]{\forall x. (\pi_{1} \land \cdots \land \pi_{n}), \bigwedge \{ x \heartsuit_{i} \boldsymbol{\lcomp x:\phi_{i}\rcomp} \,|\, i \in N  \} \To \bigvee \{ x \heartsuit_{i} \boldsymbol{\lcomp x:\phi_{i}\rcomp} \,|\, i \in M \}}{
\infer[\rul{\Seq(R)}]{\forall x. (\pi_{1} \land \cdots \land \pi_{n}), \{ x \heartsuit_{i} \boldsymbol{\lcomp x:\phi_{i}\rcomp} \,|\, i \in N \} \To \{ x \heartsuit_{i} \boldsymbol{\lcomp x:\phi_{i}\rcomp} \,|\, i \in M \}}{
\infer[\rul{L\forall}]{\{ \forall x. (\pi_{1} \land \cdots \land \pi_{n}), (\Gamma_{i} \sigma)[y/x] \To (\Delta_{i} \sigma)[y/x] \,|\, 1 \leqslant i \leqslant k \}}{
{\{ \pi_{1}[y/x] \land \cdots \land \pi_{n}[y/x], (\Gamma_{i} \sigma)[y/x] \To (\Delta_{i} \sigma)[y/x] \,|\, 1 \leqslant i \leqslant k \}}
}
}
}
\end{equation*}
which shows derivability of the axiom $\axref{Onestep($\Rules$)}$ as the top sequent is readily seen to be derivable in $\Seq\Rules$.
\end{proof}

\noindent
For the converse direction, absorption of congruence is not required.
\begin{thm}
\label{thm:nb_g2h}
Suppose that $\Seq\Rules\Cut \entails \Gamma \To \Delta$. Then $\Hilb\Rules
\entails \Land \Gamma \to \Lor \Delta$.
\end{thm}

\begin{proof}
It suffices to show that all the translations of the axioms and rules of $\Seq\Rules$ are derivable in $\Hilb\Rules$. We can easily handle the cases of the axioms and rules for logical connectives of first-order logic. The provability of the translation of $\rul{L=_{i}}$ follows from the provability of $x = y \to (\phi[x/w] \to \phi[y/w])$. As for $\hearts \in \Lambda$, the provability of the translation of $\Seq(R)$ follows from $\axref{Onestep($\Rules$)}$ and \axref{Alpha}. 
\end{proof}

\noindent
As a corollary, we obtain (for the time being, in a calculus with
cut) both soundness and completeness of the sequent calculus.

\begin{cor}
\label{cor:onestepcompl}
Suppose that $\Rules$ is strongly one-step complete. Then $\Seq\Rules\Cut \entails \Gamma \To \Delta$ iff $\models \Gamma \To \Delta$. 
\end{cor}

\begin{proof}
By Theorems \ref{thm:nb_h2g} and \ref{thm:nb_g2h} in conjunction
with soundness and completeness of $\Hilb\Rules$ (Theorem
\ref{thm:hilb-complete}). The absorption of congruence was shown in \cite[Proposition 5.12]{PattinsonSchroder10}. 
\end{proof}

\noindent A paradigmatic example of a set of rules satisfying the assumptions of Corollary \ref{cor:onestepcompl} is $\mathsf{C}$ and its CPL translation $\Seq(\mathsf{C})$ from Example \ref{ex:rulesets} above. 

As we have seen in \S~\ref{sec:completeness}, 
the assumption of \emph{strongly} one-step complete rule
sets  limits available examples to ``essentially neighbourhood-like'' ones. 
This is why we also gave a complete Hilbert-style
axiomatisation also for \emph{bounded} operators (recall Definition \ref{def:bounded}). 

Note that $k$-boundedness of $i$-th argument of Definition \ref{def:bounded}  implies in particular that $\PrL$ is monotonic in the $i$-th 
argument. Examples of bounded modalities include the standard
$\Diamond$ of relational modal logic interpreted over Kripke frames,
graded modalities over multigraphs and we refer to
\cite{SchroderPattinson10} for more examples. In the Hilbert-calculus, 
boundedness was reflected syntactically by the axiom
\begin{multline*}  \axiom{BdPL}{}_{k,i} \, \yprf
(x\PrL\cmh{y_1}{\phi_1}\dots\cmh{y_n}{\phi_n} \eqvF \exists z_1\dots
z_k.(x\PrL\cmh{y_1}{\phi_1}\dots\cmh{y_{i-1}}{\phi_{i-1}} \\ 
\qquad \cmh{y_i}{y_i \eqF z_1 \vee\dots \vee y_i \eqF
z_k}\cmh{y_{i+1}}{\phi_{i+1}}\dots\cmh{y_n}{\phi_n} \wedge
\bigwedge\limits_{j \leqslant k}\sbst{\phi_i}{z_j}{y_i}))
\end{multline*}
where each $z_{i}$ is fresh for all the $y_i$s and $\phi_{i}$s.
The derivability predicate induced by extending the Hilbert calculus
$\Hilb\Rules$ by the boundedness axiom above gives completeness
under weaker conditions.
\begin{defi}
We write $\Bound\Hilb\Rules \entails \phi$ if $\phi$ is derivable in
$\Hilb\Rules$ where additionally $\axiom{BdPL}_{k, i}$ is used
for every operator that is $k$-bounded in the $i$-th argument. 
\end{defi}
\noindent
Strictly speaking, the derivability predicate $\Bound\Hilb\Rules$
should include information about precisely which operators are
assumed to be $k$-bounded in the $i$-th argument, but this will
always be clear from the context. In the presence of boundedness,
completeness of the Hilbert-calculus has been established under weaker
conditions (see Theorem \ref{thm:hilb-complete}). 

\noindent We can reflect boundedness in the sequent calculus by 
adding a paste rule, similar in spirit to the paste rule of hybrid 
logic \cite[\S~7]{BlackburnEA01} which was generalised to a
coalgebraic setting in \cite{SchroderPattinson10}. In a sequent
setting, this rule takes the form
\begin{equation*}
\inferrule*[Right=\mbox{$\rul{\Pastek}$}]{
	\Gamma, x \hearts {\lcomp x_{1} \col \phi_{1} \rcomp} \cdots {\lcomp x_{i-1} \col \phi_{i-1} \rcomp} \lcomp y: \bigvee_{1 \leqslant j \leqslant k} y = z_{j}  \rcomp{\lcomp x_{i+1} \col \phi_{i+1} \rcomp} \cdots {\lcomp x_{n} \col \phi_{n} \rcomp}, \\\\ 
	\phi[z_{1}/y],..., \phi [z_{k}/y] \To \Delta \\ z_1, \dots, z_k \, \mathrm{fresh}
}{\Gamma, x \hearts {\lcomp x_{1} \col \phi_{1} \rcomp} \cdots {\lcomp x_{i-1} \col \phi_{i-1} \rcomp} \lcomp y:\phi \rcomp {\lcomp x_{i+1} \col \phi_{i+1} \rcomp} \cdots {\lcomp x_{n} \col \phi_{n} \rcomp} \To \Delta},
\end{equation*}
where $z_{1}$, ..., $z_{k}$ are pairwise distinct fresh variables. Additional use of the above paste-rule in the
system $\Seq\Rules$ is denoted by $\Bound\Seq\Rules$, that is, we
write $\Bound\Seq\Rules \entails \Gamma \To \Delta$ if $\Gamma \To
\Delta$ is derivable in $\Seq\Rules$ where $\Pastek$ may
additionally be applied for every modality that is $k$-bounded in
the $i$-th argument.


When $\Rules$ absorbs congruence and monotonicity of all  operators that are $k$-bounded in the $i$-th argument, we note that Lemmas \ref{lem:ad_cong} and \ref{lem:arb_repl} hold also for $\Bound\Seq\Rules$. 

\begin{thm}
\label{thm:b_h2g}
Suppose that $\Rules$ absorbs congruence and monotonicity in the
$i$-th argument of every operator that is $k$-bounded in the $i$-th
argument. 
Then $\Bound\Hilb\Rules \entails \phi$ implies that $\Bound\Seq\Rules\Cut \entails \To \phi$.
\end{thm}

\begin{proof}
First of all, if $\Rules$ absorbs monotonicity in the $i$-th argument of $\heartsuit \in \Lambda$, the rule
\begin{equation*}
\infer[\rul{Mon}_i]{\Sigma, z \heartsuit \boldsymbol{\lcomp x : \phi \rcomp} \To  z \heartsuit \lcomp x_{1}: \phi_{1}\rcomp \dots \lcomp x_{i-1}: \phi_{i-1} \rcomp \lcomp x_{i}: \psi \rcomp \lcomp x_{i+1}: \phi_{i+1} \rcomp \dots \lcomp x_{n}: \phi_{n}\rcomp, \Theta}{\Sigma, \phi_{i}[y/x_{i}] \To \psi[y/x], \Theta}
\end{equation*}
(where $y$ is fresh in the conclusion) is admissible in $\Bound \Seq \Rules$ (and $\Bound \Seq \Rules \Cut$).  
Almost all the arguments are the same as the proof of Theorem \ref{thm:nb_h2g}, except that we need to show the provability of  \axiom{BdPL} by $\rul{\mathsf{Paste}}$ (note that the only place we need the cut rule is the derivability of Modus Ponens). 
More precisely, we can show the left-to-right
implication of \axiom{BdPL} by means of $\rul{\Pastek}$ 
and $\rul{{Mon}_i}$ gives the reverse direction. For example, when 
$\heartsuit$ is unary and 1-bounded, the derivability of the right-to-left direction of \axiom{BdPL}
is demonstrated as follows. 
\begin{equation*}
\infer[\rul{L\exists}]{\exists z. (x \heartsuit \lcomp  y: y=z \rcomp  \land \phi[z/y])  \To x \heartsuit \lcomp y:\phi \rcomp
}{
\infer[\rul{L \land}]{x \heartsuit \lcomp  y: y = w \rcomp  \land \phi[w/y])  \To x \heartsuit \lcomp y:\phi \rcomp
}{
\infer[\rul{\mathsf{Mon}}]{x \heartsuit \lcomp  y: y = w \rcomp, \phi[w/y] \To x \heartsuit \lcomp y:\phi \rcomp
}{
v = w, \phi[w/y] \To \phi[v/y]
}
}
},
\end{equation*}
where the top sequent is the replacement axiom, which is derivable by Lemma \ref{lem:arb_repl}.
\end{proof}

\noindent
The reverse direction of Theorem \ref{thm:b_h2g} is established
analogously to Theorem \ref{thm:nb_g2h} and again absorption
properties are not needed.
\begin{thm}
\label{thm:b_g2h}
$\Bound\Seq\Rules\Cut \entails \Gamma \To \Delta$ only if
$\Bound\Hilb\Rules \entails \Land \Gamma \to \Lor \Delta$.
\end{thm}

\begin{proof}
The only difference from the proof of Theorem \ref{thm:nb_h2g} is to need to care about the translation of $\rul{Paste}$. However, we can easily establish this by the axiom $\mathrm{BDLP}$. 
\end{proof}

\noindent
As in the non-bounded case we obtain semantic soundness and
completeness, but under weaker coherence conditions.
\begin{cor}
Suppose that $\Rules$ is strongly finitary one-step complete. Then $\Bound\Seq\Rules\Cut \entails \Gamma \To \Delta$ iff $\models \Gamma \To \Delta$.
\end{cor}

\begin{proof}
By Theorems \ref{thm:b_h2g} and \ref{thm:b_g2h} in conjunction
with soundness and completeness of $\Bound\Hilb\Rules$ (Theorem
\ref{thm:hilb-complete}). 
Note that absorption of congruence and monotonicity follows from (strong,
finitary) one-step completeness as in \cite[Proposition 5.12]{PattinsonSchroder10}.
\end{proof}

\noindent
A canonical example of a rule set satisfying the assumptions of the above corollary can be obtained by taking $\mathsf{K}$ of Example \ref{ex:rulesets} and extending it with $\Pastek$ for $i = k = n =1$.

\subsection{Admissibility of Cut}
\label{sec:cutelim}

When we try to prove the admissibility of $\rul{Cut}$ in first-order logic (or $\Seq \Rules$), we encounter difficulties with the rules of contraction. That is, the following derivation: 
\[
\mathcal{D} =
{
\infer[\rul{Cut}]{\Gamma,\Sigma \To \Delta, \Theta}{
\infer[\rul{RC}]{\Gamma \To \Delta, \phi}{\deduce{\Gamma \To \Delta, \phi,\phi}{\mathcal{D}'}} 
& 
\deduce{\phi,\Sigma \To \Theta}{\mathcal{D}''}
},
}
\]
may be transformed into:
\[
\infer=[\rul{LC}, \rul{RC}]{\Gamma,\Sigma \To \Delta, \Theta}{
\infer[\rul{Cut}]{\Gamma,\Sigma,\Sigma \To \Delta, \Theta, \Theta}{
\infer[\rul{Cut}]{\Gamma,\Sigma \To \Delta, \Theta,\phi}{
{\deduce{\Gamma \To \Delta, \phi,\phi}{\mathcal{D}'}} 
& 
\deduce{\phi,\Sigma \To \Theta}{\mathcal{D}''}
}
&
\deduce{\phi,\Sigma \To \Theta}{\mathcal{D}''}
}
},
\]
but this derivation does not provide us with a reduction in terms of the number of sequents above the application of $\rul{Cut}$ in $\mathcal{D}$. This is why Gentzen introduced the following generalized form of $\rul{Cut}$: 
\[
\infer[\rul{Mcut}]{\Gamma,\Sigma \To \Delta, \Theta}{\Gamma \To \Delta, \phi^{m} & \phi^{n},\Sigma \To \Theta}
\]
where $n$, $m \geqslant 1$ and $\phi^{k}$ stands for $k$ copies of $\phi$ and ``$Mcut$'' is a shorthand of ``multi-cut'' (sometimes also called ``mix''). Since $\rul{Cut}$ is a special case of the new rule of $\rul{Mcut}$, it suffices for us to prove the admissibility of $\rul{Mcut}$ in a given sequent system to obtain the admissibility of $\rul{Cut}$ in the system. 

Moreover, we note that \emph{a priori} we cannot expect that an application of $\rul{Mcut}$ between two instances of modal rules can be moved up without changing the conclusion: 
the set $\Rules$ of one-step rules can possibly consist of a single rule, 
and an application of $\rul{Mcut}$ between this rule and itself may not be derivable. 
We therefore need to impose an additional requirement to deal with this case. 

\begin{defi}
Let $\mathbb{S}$ be a finite set of sequents. 
The set of all sequents that can be derived from premises in $\mathbb{S}$ using (only) {\em one application} of $\rul{Mcut}$ is denoted by $\MCut(\mathbb{S})$.  
A rule set $\Rules$ \emph{absorbs multicut}, if for all pairs $(R_{1},R_{2})$ of rules in $\Rules$:
\[
\infer[{R_{1}}]{\Gamma_{R_{1}} \To \Delta_{R_{1}}, (\hearts \vec{\sva}\,)^{m}}{\Gamma_{11} \To \Delta_{11} &\cdots & \Gamma_{1r_{1}} \To \Delta_{1r_{1}}}
\quad
\infer[{R_{2}}]{(\hearts \vec{\sva}\,)^{n}, \Gamma_{R_{2}} \To \Delta_{R_{2}}}{\Gamma_{21} \To \Delta_{21} &\cdots & \Gamma_{2r_{2}} \To \Delta_{2r_{2}}}
\]
there is a rule $R$ = $\Gamma_{1} \To \Delta_{1}, \cdots, \Gamma_{k} \To \Delta_{k} / \Gamma_{R} \To \Delta_{R}  \in \Rules$ such that:
\[
\MCut(\Gamma_{11} \To \Delta_{11}, \ldots, \Gamma_{1r_{1}} \To \Delta_{1r_{1}},\Gamma_{21} \To \Delta_{21}, \ldots, \Gamma_{2r_{2}} \To \Delta_{2r_{2}} ) \rhd \{ \Gamma_{1} \To \Delta_{1}, \ldots, \Gamma_{k}\To \Delta_{k} \}
\]
and $\Gamma_{R} \To \Delta_{R} \rhd \Gamma_{R_{1}}, \Gamma_{R_{2}} \To \Delta_{R_{1}}, \Delta_{R_{2}}$.
\end{defi}

\begin{lem}
\label{lem:rename}
If $\Seq\Rules \vdash \Gamma \To \Delta$ and $y$ is fresh in $\Gamma$ and $\Delta$, then $\Seq\Rules \vdash \Gamma[y/x] \To \Delta[y/x]$ with the same height of derivation. 
\end{lem}


\begin{lem}[Hauptsatz]
\label{lem:hauptsatz}
Let $\mathcal{D}$ be a derivation in the system $\Seq\Rules$ extended with $\rul{Mcut}$ in  the following form:
\[
{
\infer[\rul{Mcut}]{\Gamma, \Sigma \To \Delta, \Theta}{\deduce[\mathcal{D}_{\texttt{L}}]{\Gamma \To \Delta, \phi^{m}}{} & \deduce[\mathcal{D}_{\texttt{R}}]{\phi^{n}, \Sigma \To \Theta}{}},}
\]
where $\mathcal{D}_{\texttt{L}}$ and $\mathcal{D}_{\texttt{R}}$ contain no application of $\rul{Mcut}$, the last rule of $\mathcal{D}$ is the only application of $\rul{Mcut}$ in $\mathcal{D}$. Then $\Gamma,\Sigma\To \Delta,\Theta$ is derivable in $\Seq \Rules$. 
\end{lem}

\begin{proof}
First of all, we introduce some terminology used only in this proof. Let $\mathcal{D}$ be the derivation in question. 
 We say that $\phi$ is a {\em cut formula} of $\mathcal{D}$, and we define the complexity $\texttt{c}(\mathcal{D})$ as the complexity of the cut formula $\phi$, i.e., the length or the number of connectives including the logical and modal connectives. Moreover, we define $\texttt{w}(\mathcal{D})$ as the total number of sequents in $\mathcal{D}_{\texttt{L}}$ and $\mathcal{D}_{\texttt{R}}$. Our proof of the statement of the claim is shown by the double induction on $(\texttt{c}(\mathcal{D}), \texttt{w}(\mathcal{D}))$ (note that $\texttt{c}(\mathcal{D}) \geqslant 0$ and $\texttt{w}(\mathcal{D})\geqslant 2$). Let us denote the last applied rule (or axiom, possibly) of a derivation $\mathcal{E}$ by $\texttt{rule}(\mathcal{E})$. We divide our argument into the following (exhaustive) cases:
\begin{enumerate}
\item One of $\texttt{rule}(\mathcal{D}_{\texttt{L}})$ and $\texttt{rule}(\mathcal{D}_{\texttt{R}})$ is an axiom. \item One of $\texttt{rule}(\mathcal{D}_{\texttt{L}})$ and $\texttt{rule}(\mathcal{D}_{\texttt{R}})$ is a structural rule. 
\item One of $\texttt{rule}(\mathcal{D}_{\texttt{L}})$ and $\texttt{rule}(\mathcal{D}_{\texttt{R}})$ is a logical rule or a modal rule and the cut formula is not principal in the rule.
\item Both $\texttt{rule}(\mathcal{D}_{\texttt{L}})$ and $\texttt{rule}(\mathcal{D}_{\texttt{R}})$ are logical rules for the same logical connective and the cut formula is principal in each of the rules. 
\item Both $\texttt{rule}(\mathcal{D}_{\texttt{L}})$ and $\texttt{rule}(\mathcal{D}_{\texttt{R}})$ are modal rules and the cut formula is principal in each of the rules. 
\item One of $\texttt{rule}(\mathcal{D}_{\texttt{L}})$ and $\texttt{rule}(\mathcal{D}_{\texttt{R}})$ is an equality rule. 
\end{enumerate}
Let us check each case one by one. 
\begin{enumerate}
\item One of $\texttt{rule}(\mathcal{D}_{\texttt{L}})$ and $\texttt{rule}(\mathcal{D}_{\texttt{R}})$ is an axiom:
We have four cases since it is impossible that $\texttt{rule}(\mathcal{D}_{\texttt{L}})$ is $\rul{L \bot}$ or $\texttt{rule}(\mathcal{D}_{\texttt{L}})$ is $\rul{R=}$.  Firstly, when $\texttt{rule}(\mathcal{D}_{\texttt{L}})$ is $\rul{Ax}$, let the derivation be
\[
\infer[\rul{Mcut}]{\phi, \Sigma \To \Theta}{
\infer[\rul{Ax}]{\phi \To \phi}{}
&
\deduce[\mathcal{D}_{\texttt{R}}]{\phi^{n}, \Sigma \To \Theta}{}}.
\]
When $n$ = $1$, we already obtain the derivability of $\phi, \Sigma \To \Theta$ in $\Seq \Rules$.  When $n \geqslant 2$, $\phi, \Sigma \To \Theta$ is obtained from $\phi^{n}, \Sigma \To \Theta$ by finitely many applications of $\rul{LC}$. 

Secondly, when $\texttt{rule}(\mathcal{D}_{\texttt{R}})$ is $\rul{Ax}$, the argument is similar to the previous case where $\texttt{rule}(\mathcal{D}_{\texttt{L}})$ is $\rul{Ax}$. 

Thirdly, when $\texttt{rule}(\mathcal{D}_{\texttt{L}})$ is $\rul{R=}$, we need to look at what the last rule $\texttt{rule}(\mathcal{D}_{\texttt{R}})$ is, where $\mathcal{D}$ is of the following form:
\[
{
\infer[\rul{Mcut}]{\Sigma \To \Theta}{\infer[\rul{R=}]{\To x=x}{} & \deduce[\mathcal{D}_{\texttt{R}}]{(x=x)^{n}, \Sigma \To \Theta}{}}.
}
\]
If $\texttt{rule}(\mathcal{D}_{\texttt{R}})$ is an axiom, then it should be $\rul{Ax}$ and we have already checked this case in our second case of this item. Otherwise, $\texttt{rule}(\mathcal{D}_{\texttt{R}})$ is a structural rule, a logical rule, a modal rule or an equality rule. These cases will be discussed below (especially (2), (3) and (6)), so we leave them out for now. 
It is, however, noted that the cut formula $x=x$ is not principal in the case (3). 

Fourthly, when $\texttt{rule}(\mathcal{D}_{\texttt{R}})$ is $\rul{L\bot}$, then we need to look at the last rule $\texttt{rule}(\mathcal{D}_{\texttt{L}})$, where $\mathcal{D}$ is 
\[
{
\infer[\rul{Mcut}]{\Gamma \To \Delta }{\deduce[\mathcal{D}_{\texttt{L}}]{\Gamma \To \Delta, \bot^{m}}{} & \infer[\rul{L \bot}]{\bot \To}{}}
}.
\]
If $\texttt{rule}(\mathcal{D}_{\texttt{L}})$ is an axiom, it should be $\rul{Ax}$ and we have already checked such case in the first case of this item. Otherwise, $\texttt{rule}(\mathcal{D}_{\texttt{L}})$ must be a structural rule, a logical rule, an modal rule or an equality rule. Again these cases will be discussed below (especially (2), (3) and (6), where  it is noted that the cut formula $\bot$ is not principal in the case (3)), so we leave them out for now. 

\item One of $\texttt{rule}(\mathcal{D}_{\texttt{L}})$ and $\texttt{rule}(\mathcal{D}_{\texttt{R}})$ is a structural rule: all arguments for this case are standard, so we deal only with the case where $\texttt{rule}(\mathcal{D}_{\texttt{L}})$ is $\rul{RC}$, i.e., $\mathcal{D}$ is of the following form:
\[
{
\infer[\rul{Mcut}]{\Gamma, \Sigma \To \Delta, \Theta}{\infer[\rul{RC}]{\Gamma \To \Delta, \phi^{m}}{\deduce[\mathcal{D}_{\texttt{L}}']{\Gamma \To \Theta, \phi^{m+1}}{}} & \deduce[\mathcal{D}_{\texttt{R}}]{\phi^{n}, \Sigma \To \Theta}{}},
}
\]
since multicut plays an essential role. This derivation is transformed into: 
\[
\infer[\rul{Mcut}]{\Gamma,\Sigma \To \Delta, \Theta}{
\deduce[\mathcal{D}_{\texttt{L}}']{\Gamma \To \Delta, \phi^{m+1}}{}
&
\deduce[\mathcal{D}_{\texttt{R}}]{\phi^{n}, \Sigma \To \Theta}{}
}
\]
where the application of $\rul{Mcut}$ is eliminable since the complexity of the derivation is the same as $\texttt{c}(\mathcal{D})$ and  the weight of the derivation is smaller than $\texttt{w}(\mathcal{D})$. 
\item One of $\texttt{rule}(\mathcal{D}_{\texttt{L}})$ and $\texttt{rule}(\mathcal{D}_{\texttt{R}})$ is a logical rule or a modal rule and the cut formula is not principal in the rule: Our argument for logical rules are standard, so we focus on the case where one of the rules is a modal rule $\Seq(R)$. Let $\texttt{rule}(\mathcal{D}_{\texttt{L}})$ is $\Seq(R)$. Then our derivation $\mathcal{D}$ is of the following form:
\[
\scalebox{0.83}{\infer[\rul{Mcut}]{\deduce{z \hearts_{n+1} \boldsymbol{\lcomp x
	\col \phi_{n+1} \rcomp}, \dots, z \hearts_{n+m} \boldsymbol{\lcomp x
	\col \phi_{n+m} \rcomp}, \Theta', \Theta}{\Sigma, \Sigma', z \hearts_1 \boldsymbol{\lcomp
	x \col \phi_1 \rcomp}, \dots, z \hearts_n \boldsymbol{\lcomp x
	\col \phi_n \rcomp} \To }}{
\infer[{\Seq}(R)\yfresh]{\deduce{
	z \hearts_{n+1} \boldsymbol{\lcomp x
	\col \phi_{n+1} \rcomp}, \dots, z \hearts_{n+m} \boldsymbol{\lcomp x
	\col \phi_{n+m} \rcomp}, \Theta', \phi^{m}}{\Sigma', z \hearts_1 \boldsymbol{\lcomp
	x \col \phi_1 \rcomp}, \dots, z \hearts_n \boldsymbol{\lcomp x
	\col \phi_n \rcomp} \To }}{
\deduce[\mathcal{D}_{\texttt{L}_{1}}]{\Sigma', (\Gamma_{1} \sigma) [y/x] \To (\Delta_{1} \sigma) [y/x], \Theta', \phi^{m}}{}
\cdots \deduce[\mathcal{D}_{\texttt{L}_{k}}]{\Sigma', (\Gamma_{k} \sigma) [y/x] \To (\Delta_{k} \sigma) [y/x], \Theta',\phi^{m}}{}   } 
& \deduce[\mathcal{D}_{\texttt{R}}]{\phi^{n}, \Sigma \To \Theta}{}}
}\]
where $\yfresh$ in the application of $\Seq(R)$ means that $y$ is fresh in the conclusion. For each $\mathcal{D}_{\texttt{L}_{i}}$, we apply height-preserving substitution $[z/y]$ for a fresh variable $z$ in the conclusion of $\mathcal{D}$ and we obtain the following derivation:
\[
\infer[\rul{Mcut}]{\Sigma, \Sigma', (\Gamma_{i} \sigma) [z/x] \To (\Delta_{i} \sigma) [z/x], \Theta',\Theta}{
\deduce[\text{$\mathcal{D}_{\texttt{L}_{i}} [z/y]$} ]{\Sigma', (\Gamma_{i} \sigma) [z/x] \To (\Delta_{i} \sigma) [z/x], \Theta', \phi^{m}}{}
& 
\deduce[\mathcal{D}_{\texttt{R}}]{\phi^{n}, \Sigma \To \Theta}{}}
.
\]
We can eliminate the last application of $\rul{Mcut}$ since the complexity of the derivation is the same as $\texttt{c}(\mathcal{D})$ and  the weight of the derivation is smaller than $\texttt{w}(\mathcal{D})$. Finally we apply the same rule $\Seq(R)$ to obtain the desired conclusion. When $\texttt{rule}(\mathcal{D}_{\texttt{R}})$ be $\Seq(R)$, the argument is similar to the case just discussed. 

\item Both $\texttt{rule}(\mathcal{D}_{\texttt{L}})$ and $\texttt{rule}(\mathcal{D}_{\texttt{R}})$ are logical rules for the same logical connective and the cut formula is principal in each of the rules: We have two cases, i.e., two cases where the cut formula is of the form $\phi \to \psi$ or of the form $\forall x.\phi$. Here we only deal with the case where the cut formula is of the form $\forall x.\phi$. Then the derivation $\mathcal{D}$ is of the following form: 
\[
{
\infer[\rul{Mcut}]{\Gamma, \Sigma \To \Delta, \Theta}{
\infer[\rul{R\forall}\yfresh]{\Gamma \To \Delta, (\forall x.\phi)^{m}}
{
\deduce[\mathcal{D_{\texttt{L}}'}]{\Gamma \To \Delta, (\forall x.\phi)^{m-1}, \phi[y/x]}{}
} 
&\infer[\rul{L\forall}]{
(\forall x.\phi)^{n}, \Sigma \To \Theta
}{
\deduce[\mathcal{D}_{\texttt{R}}']{\phi[z/x], (\forall x.\phi)^{n-1}, \Sigma \To \Theta
}{}
}
}
.
}
\]
With the help of our height-preserving substitution, 
we can consider a multicut between $\mathcal{D}_{\texttt{L}}'$ and $\mathcal{D}_{\texttt{R}}$:
\[
\infer[\rul{Mcut}]{\Gamma,\Sigma \To \Delta,\Theta, \phi[z/x]}{
\deduce[\text{$\mathcal{D_{\texttt{L}}'}[z/y]$}]{\Gamma \To \Delta, (\forall x.\phi)^{m-1}, \phi[z/x]}{}
&
\deduce[\mathcal{D}_{\texttt{R}}]{(\forall x.\phi)^{n}, \Sigma \To \Theta}{}
},
\]
and then by induction hypothesis (the complexity of this derivation is the same as $\mathcal{D}$ but the weight is smaller than the original $\mathcal{D}$) we now know that $\Gamma,\Sigma \To \Delta,\Theta, \phi[z/x]$ is derivable in $\Seq\Rules$ without multicuts by a derivation $\mathcal{E}_{1}$. Let us also consider a multicut between $\mathcal{D}_{\texttt{L}}$ and $\mathcal{D}_{\texttt{R}}'$:
\[
\infer[\rul{Mcut}]{\phi[z/x],\Gamma,\Sigma \To \Delta,\Theta}{
\deduce[\mathcal{D}_{\texttt{L}}]{\Gamma \To \Delta, (\forall x.\phi)^{m}}{}
&
\deduce[\mathcal{D}_{\texttt{R}}']{\phi[z/x], (\forall x.\phi)^{n-1}, \Sigma \To \Theta}{}
},
\]
and then by induction hypothesis (the complexity of this derivation is the same as $\mathcal{D}$ but the weight is smaller than the original $\mathcal{D}$) we now know that $\phi[z/x],\Gamma,\Sigma \To \Delta,\Theta$ is derivable in $\Seq\Rules$ without multicuts by a derivation $\mathcal{E}_{2}$.  Now let us take a cut between $\mathcal{E}_{1}$ and $\mathcal{E}_{2}$: 
\[
\infer[\rul{Mcut}]{
\Gamma,\Gamma,\Sigma,\Sigma  \To \Delta,\Delta,\Theta,\Theta
}{
\deduce[\mathcal{E}_{1}]{\Gamma,\Sigma \To \Delta,\Theta, \phi[z/x]}{}
&
\deduce[\mathcal{E}_{2}]{\phi[z/x],\Gamma,\Sigma \To \Delta,\Theta}{}
}
\]
and the conclusion of this derivation is derivable in $\Seq\Rules$ without multicuts by induction hypothesis because the complexity of this derivation (i.e., the length of $\phi[z/x]$) is strictly smaller than $\texttt{c}(\mathcal{D})$. Finally, finitely many applications of contraction rules enables us to obtain the derivability of $\Gamma,\Sigma \To \Delta,\Theta$ in $\Seq\Rules$, as desired. 
\item Both $\texttt{rule}(\mathcal{D}_{\texttt{L}})$ and $\texttt{rule}(\mathcal{D}_{\texttt{R}})$ are modal rules and the cut formula is principal in each of the rules: 
Let $\texttt{rule}(\mathcal{D}_{\texttt{L}})$ = $\Seq(R_{1})$ and $\texttt{rule}(\mathcal{D}_{\texttt{R}})$ = $\Seq(R_{2})$ where we can assume: 
\[
R_{1} = 
\vcenter{
\infer[]{\hearts_1 \vec{\sva}_1, \dots, \hearts_a \vec{\sva}_a \To \hearts_{a+1}
	\vec{\svb}_1, \dots, \hearts_{a+b} \vec{\svb}_b, (\hearts \vec{\sva})^{n}}{\Gamma_{11} \To \Delta_{11} & \cdots & \Gamma_{1k}\To \Delta_{1k}}},
\]
\[
R_{2} = 
\vcenter{
\infer[]{(\hearts \vec{\sva})^{m}, \spadesuit_1 \vec{\sva}_1', \dots, \spadesuit_c \vec{\sva}_c' \To \spadesuit_{c+1}
	\vec{\svb}_1', \dots, \spadesuit_{c+d} \vec{\svb}_d'}{\Gamma_{21} \To \Delta_{21} & \cdots & \Gamma_{2l}\To \Delta_{2l}}
},
\]
because the cut formula is principal in both rules. In what follows, we assume that all of $\vec{\sva}_{i}$, $\vec{\svb}_{j}$,  $\vec{\sva}_{i}'$, $\vec{\svb}_{j}'$ are distinct. 
So $\mathcal{D}_{\texttt{L}}$ is of the following form:
\[
\infer[{\Seq(R_{1})}]{\deduce{ z \hearts_{a+1} \boldsymbol{\lcomp x_{a+1}
	\col \phi_{a+1} \rcomp}, \dots, z \hearts_{a+b} \boldsymbol{\lcomp x_{a+b}
	\col \phi_{a+b} \rcomp}, (z \hearts \boldsymbol{\lcomp x
	\col \phi \rcomp})^{m}, \Theta_{1}}{\Sigma_{1}, z \hearts_1 \boldsymbol{\lcomp
	x_{1} \col \phi_1 \rcomp}, \dots, z \hearts_{a} \boldsymbol{\lcomp x_{a}
	\col \phi_a \rcomp} \To}}{
\deduce[\mathcal{D}_{\texttt{L}_{1}}']{\Sigma_{1}, \Gamma_{11} \sigma_{\boldsymbol{x}}^{y_{1}} \To \Delta_{11} \sigma_{\boldsymbol{x}}^{y_{1}}, \Theta_{1}}{}  
&\cdots &\deduce[\mathcal{D}_{\texttt{L}_{k}}']{\Sigma_{1}, \Gamma_{1k}\sigma_{\boldsymbol{x}}^{y_{1}}  \To \Delta_{1k} \sigma_{\boldsymbol{x}}^{y_{1}}, \Theta_{1}}{}
}
\]
\noindent and $\mathcal{D}_{\texttt{R}}$ is of the following form:
\[
\infer[{\Seq(R_{2})}]{\deduce{z \spadesuit_{c+1} \boldsymbol{\lcomp x_{c+1}
	\col \psi_{c+1} \rcomp}, \dots, z \spadesuit_{c+d} \boldsymbol{\lcomp x_{c+d}
	\col \psi_{c+d} \rcomp}, \Theta_{2}}{\Sigma_{2},  (z \hearts \boldsymbol{\lcomp x
	\col \phi \rcomp})^{n}, z \spadesuit_1 \boldsymbol{\lcomp
	x_1 \col \psi_1 \rcomp}, \dots, z \spadesuit_c \boldsymbol{\lcomp x_c
	\col \psi_c \rcomp} \To}}{
\deduce[\mathcal{D}_{\texttt{R}_{1}}']{ \Sigma_{2}, \Gamma_{21} \tau_{\boldsymbol{x}}^{y_{2}}  \To \Delta_{21} \tau_{\boldsymbol{x}}^{y_{2}}, \Theta_{2}}{}  
&\cdots &\deduce[\mathcal{D}_{\texttt{R}_{l}}']{\Sigma_{2}, \Gamma_{2l} \tau_{\boldsymbol{x}}^{y_{2}} \To \Delta_{2l} \tau_{\boldsymbol{x}}^{y_{2}}, \Theta_{2}}{}
}.
\]

\medskip

We also note that the conclusion of $\mathcal{D}$ is:
\begin{align*}
\Sigma_{1},\Sigma_{2},  & \{ z \hearts_i \boldsymbol{\lcomp x_{i} \col \phi_i \rcomp}  \}_{1 \leqslant i \leqslant a}, \{ z \spadesuit_{j} \boldsymbol{\lcomp x_{j} \col \psi_{j} \rcomp} \}_{1 \leqslant j \leqslant c} \To \\
 & \quad \{ z \hearts_{a+i} \boldsymbol{\lcomp x_{a+i} \col \phi_{a+i} \rcomp} \}_{1 \leqslant i \leqslant b}, \{z \spadesuit_{c+j} \boldsymbol{\lcomp x_{c+j} \col \psi_{c+j} \rcomp} \}_{1 \leqslant j \leqslant d}, \Theta_{1}, \Theta_{2}.
\end{align*}
Let $y$ be a fresh variable not occurring in this conclusion. By height-preserving substitution, we can obtain derivations $\mathcal{D}'_{\texttt{L}_{i}}[z/y_{1}]$ and $\mathcal{D}'_{\texttt{R}_{j}}[z/y_{2}]$ ($1 \leqslant i \leqslant k$ and $1 \leqslant j \leqslant l$). Since $\Rules$ absorbs multicut, we can find a rule $R = \Gamma_{1} \To \Delta_{1}, \cdots, \Gamma_{e} \To \Delta_{e}/\Gamma_{R} \To \Delta_{R} \in \Rules$ such that 
\begin{itemize}
\item[($\ast_1$)]  $\MCut(\{ \Gamma_{1i} \To \Delta_{1i}\}_{1 \leqslant i \leqslant k},  \{\Gamma_{2j} \To \Delta_{2j}\}_{1 \leqslant j \leqslant l}) \rhd \{ \Gamma_{1} \To \Delta_{1}, \ldots, \Gamma_{e}\To \Delta_{e} \}$ and 
\item[($\ast_2$)] $\Gamma_{R} \To \Delta_{R} \rhd  \{ \hearts_i \vec{\sva}_i \}_{1 \leqslant i \leqslant a}, \{ \spadesuit_j \vec{\sva}_j'\}_{1\leqslant  j \leqslant c} \To \{ \hearts_{a+i} \vec{\svb}_i \}_{1 \leqslant i \leqslant b}, \{ \spadesuit_{c+j} \vec{\svb}_j'\}_{1 \leqslant j \leqslant d}$. 
\end{itemize}
By the clause $(\ast_{1})$ and our derivations $\mathcal{D}'_{\texttt{L}_{i}}[z/y_{1}]$, $\mathcal{D}'_{\texttt{R}_{j}}[z/y_{2}]$, we now use the induction hypothesis (the complexity is the same but the weight becomes smaller than that of $\mathcal{D}$) and weakening rules to obtain the derivability in $\Seq\Rules$ (without multicuts) of 
\[
\Gamma_{i} \sigma \To \Delta_{i} \sigma \quad (1 \leqslant i \leqslant e)
\]
where $\sigma$ is a substitution which is the union of $\sigma_{\boldsymbol{x}}^{y_{1}}$ and $\tau_{\boldsymbol{x}}^{y_{2}}$. 
It follows from the rule $\Seq(R)$, the clause $(\ast_{2})$ and weakening rules that the conclusion of $\mathcal{D}$ is derivable in $\Seq\Rules$ without multicuts, as desired. 

\item One of $\texttt{rule}(\mathcal{D}_{\texttt{L}})$ and $\texttt{rule}(\mathcal{D}_{\texttt{R}})$ is an equality rule: There are three cases that we need to consider. In the first case, $\texttt{rule}(\mathcal{D}_{\texttt{R}})$ is $\rul{L=_{i}}$ where at least one occurrence of the cut formulas is not principal in $\rul{L=_{i}}$ and so the cut formula is of the form $x=y$. In the second case, $\texttt{rule}(\mathcal{D}_{\texttt{L}})$ is $\rul{L=_{i}}$ but all occurrences of the cut formula are principal. In the third case, $\texttt{rule}(\mathcal{D}_{\texttt{R}})$ is $\rul{L=_{i}}$ and all occurrences of the cut formula are principal. Since our argument for the third case is almost similar to the one for the second case, we focus on the first and the second cases in what follows. 

Firstly, consider the case when $\texttt{rule}(\mathcal{D}_{\texttt{R}})$ is $\rul{L=_{i}}$ and at least one occurrence of the cut formulas is not principal in $\rul{L=_{i}}$. Without loss of generality, we assume that $i$ = 1. Then our derivation $\mathcal{D}$ is of the following form:
\[
{
\infer[\rul{Mcut}]{\Gamma, \Sigma \To \Delta, \Theta}{\deduce[\mathcal{D}_{\texttt{L}}]{\Gamma \To \Delta, (x=y)^{m}}{} & 
\infer[\rul{L=_{1}}]{x = y, \phi'_{1}[y/w], \ldots, \phi'_{n-1}[y/w], \Sigma'[y/w] \To \Theta'[y/w]}{
\deduce[\mathcal{D}_{\texttt{R}}']{x = y, \phi'_{1}[x/w], \ldots, \phi'_{n-1}[x/w], \Sigma'[x/w] \To \Theta'[x/w]}{}}
}
}
\]
where $\Sigma'[y/w]$ = $\Sigma$, $\Theta'[y/w]$ = $\Theta$, $\phi_{i}'[y/w]$ is $x = y$ and so $x = y, \phi'_{1}[y/w], \ldots, \phi'_{n-1}[y/w]$ is the same as $(x=y)^{n}$. In this case we need to check what is the last rule $\texttt{rule}(\mathcal{D}_{\texttt{L}})$. If $\texttt{rule}(\mathcal{D}_{\texttt{L}})$ is $\rul{Ax}$ (it cannot be $\rul{L\bot}$), or a structural rule, or a logical or modal rule, we can use the same argument in the items (1), (2), (3). If $\texttt{rule}(\mathcal{D}_{\texttt{L}})$ is an equality rule $\rul{L=_{i}}$, then our argument is the same as in the second case below. 
 The remaining case is $\texttt{rule}(\mathcal{D}_{\texttt{L}})$ is an axiom $\rul{R=}$. Then our derivation above $\mathcal{D}$ has the following form:
\[
{
\infer[\rul{Mcut}]{\Sigma \To\Theta}{\infer[\rul{R=}]{\To x=x}{} & 
\infer[\rul{L=_{1}}]{x = x, \phi'_{1}[x/w], \ldots, \phi'_{n-1}[x/w], \Sigma'[x/w] \To \Theta'[x/w]}{
\deduce[\mathcal{D}_{\texttt{R}}']{x = x, \phi'_{1}[x/w], \ldots, \phi'_{n-1}[x/w], \Sigma'[x/w] \To \Theta'[x/w]}{}}.
}
}
\]
Then this derivation is transformed into:
\[
\infer[\rul{Mcut}]{\Sigma \To \Theta}{\infer[\rul{R=}]{\To x=x}{} & 
\deduce[\mathcal{D}_{\texttt{R}}']{x = x, \phi'_{1}[x/w], \ldots, \phi'_{n-1}[x/w], \Sigma'[x/w] \To \Theta'[x/w]}{}
}
\]
and this last application of multicut is eliminable since the complexity is the same as that of $\mathcal{D}$ but the weight becomes smaller. 

Secondly, let $\texttt{rule}(\mathcal{D}_{\texttt{L}})$ be $\rul{L=_{i}}$ and assume that all occurrences of the cut formula are principal. In this case the derivation $\mathcal{D}$ is of the following form:
\[
{
\infer[\rul{Mcut}]{x=y, \Gamma', \Sigma \To \Delta, \Theta}{
\infer[\rul{L=_{1}}]{x= y, \Gamma''[y/w] \To \Delta''[y/w], \phi_{1}'[y/w], \ldots, \phi_{m}'[y/w] }{
\deduce[\mathcal{D}_{\texttt{L}}']{x= y, \Gamma''[x/w] \To \Delta''[x/w], \phi_{1}'[x/w], \ldots, \phi_{m}'[x/w]}{}
} & 
\deduce[\mathcal{D}_{\texttt{R}}]{\phi^{n}, \Sigma \To \Theta}{}}
}
\]
where $\Gamma''[y/w]$ = $\Gamma'$, $\Delta'[y/w]$ = $\Delta$ and $\phi_{i}'[y/w]$ = $\phi$ ($1 \leqslant i \leqslant m$). 
Before transforming this derivation into a multicut-free derivation, we remark that $\phi_{i}'[x/w][y/x]$ = $\phi_{i}'[y/w][y/x]$ = $\phi[y/x]$, 
$\Gamma''[x/w][y/x]$ = $\Gamma''[y/w][y/x]$ = $\Gamma'[y/x]$, 
$\Delta'[x/w][y/x]$ = $\Delta'[y/w][y/x]$ = $\Delta[y/x]$. 
With the help of this remark, the derivation $\mathcal{D}$ is transformed into:
\[
\infer[\rul{LC}]{x=y, \Gamma',\Sigma \To \Delta,\Theta}{
\infer[\rul{L=_{2}}]{x=y, x= y, \Gamma',\Sigma \To \Delta,\Theta}{
\infer[\rul{RW}]{x=y, y= y, \Gamma'[y/x],\Sigma[y/x] \To \Delta[y/x],\Theta[y/x]}{
\infer[\rul{Mcut}]{y= y, \Gamma'[y/x],\Sigma[y/x] \To \Delta[y/x],\Theta[y/x]}{
\deduce[\text{$\mathcal{D}_{\texttt{L}}'[y/x]$}]{y= y, \Gamma'[y/x] \To \Delta[y/x], (\phi[y/x])^{m}}{}
&
\deduce[\text{$\mathcal{D}_{\texttt{R}}[y/x]$}]{(\phi[y/x])^{n}, \Sigma[y/x] \To \Theta[y/x]}{}}
}
}
}
\]
where we note that the first application of multicut is eliminable since the complexity is the same as that of $\mathcal{D}$ but the weight is smaller than that of $\mathcal{D}$ by height-preserving substitution $[y/x]$.  \qedhere
\end{enumerate}
\end{proof}

\begin{thm}[Cut Elimination]
\label{thm:cut-elim-strongly}
Suppose that $\Rules$ absorbs multicut. Then the rule $\rul{Mcut}$ is admissible in $\Seq \Rules$. Therefore, $\rul{Cut}$ is also admissible in $\Seq \Rules$. 
\end{thm}

\begin{proof}
Suppose that a sequent is derivable in the system $\Seq\Rules$ extended with $\rul{MCut}$. Let $\mathcal{E}$ be  such a derivation. Then we focus on one of the topmost applications of $\rul{MCut}$ to show that such application of $\rul{MCut}$ is eliminable, i.e., we show that the derivation whose last applied rule is such multicut can be replaced with a multicut-free derivation of $\Seq \Rules$. This is done using Lemma \ref{lem:hauptsatz}. Once we eliminate one of the topmost applications of $\rul{MCut}$, we repeat the same argument for the remaining topmost applications with the help of Lemma \ref{lem:hauptsatz} to get rid of all applications of $\rul{MCut}$ in the original derivation $\mathcal{E}$. 
\end{proof}

\noindent In what follows, we introduce the notion of \emph{absorption of contraction and cut} and show that jointly they provide a sufficient condition of absorption of multicut. 

\begin{defi} 
Let $\mathbb{S}$ be a finite set of sequents. 
The set of sequents that can be derived from premises $\mathbb{S}$ using  (only) the \emph{contraction rules} is denoted by $\Con(\mathbb{S})$. Similarly,  the set of all sequents that can be derived from premises in $\mathbb{S}$ using (only) {\em one application} of the \emph{cut rule} is denoted by $\Cut(\mathbb{S})$. A rule set $\Rules$ \emph{absorbs contraction} if, for all rules $R$ = $\Gamma_{1} \To \Delta_{1}, \cdots, \Gamma_{k} \To \Delta_{k} / \Gamma_{R} \To \Delta_{R}  \in \Rules$ and all $\Gamma' \To \Delta' \in \Con(\Gamma_{R} \To \Delta_{R})$ there exists a rule $S = \Sigma_{1} \To \Theta_{1}, \cdots, \Sigma_{l} \To \Theta_{l} / \Gamma_{S} \To \Delta_{S} \in \Rules$ such that 
\[
\Con(\{ \Gamma_{1} \To \Delta_{1}, \ldots, \Gamma_{k} \To \Delta_{k} \}) \rhd \{ \Sigma_{1} \To \Theta_{1}, \ldots, \Sigma_{l} \To \Theta_{l} \}
\]
and $\Gamma_{S} \To \Delta_{S} \rhd \Gamma' \To \Delta'$. A rule set $\Rules$ \emph{absorbs cut}, if for all pairs $(R_{1},R_{2})$ of rules in $\Rules$:
\[
\infer[{R_{1}}]{\Gamma_{R_{1}} \To \Delta_{R_{1}}, \hearts \vec{\sva}}{\Gamma_{11} \To \Delta_{11} &\cdots & \Gamma_{1r_{1}} \To \Delta_{1r_{1}}}
\quad
\infer[{R_{2}}]{\hearts \vec{\sva}, \Gamma_{R_{2}} \To \Delta_{R_{2}}}{\Gamma_{21} \To \Delta_{21} &\cdots & \Gamma_{2r_{2}} \To \Delta_{2r_{2}}}
\]
there is a rule $R$ = $\Gamma_{1} \To \Delta_{1}, \cdots, \Gamma_{k} \To \Delta_{k} / \Gamma_{R} \To \Delta_{R}  \in \Rules$ such that:
\[
\Cut(\Gamma_{11} \To \Delta_{11}, \ldots, \Gamma_{1r_{1}} \To \Delta_{1r_{1}},\Gamma_{21} \To \Delta_{21}, \ldots, \Gamma_{2r_{2}} \To \Delta_{2r_{2}} ) \rhd \{ \Gamma_{1} \To \Delta_{1}, \ldots, \Gamma_{k}\To \Delta_{k} \}
\]
and $\Gamma_{R} \To \Delta_{R} \rhd \Gamma_{R_{1}}, \Gamma_{R_{2}} \To \Delta_{R_{1}}, \Delta_{R_{2}}$.
\end{defi}

\noindent
Informally, absorption of cut and contraction of a rule set allows us 
to replace an application of cut or contraction to the conclusions 
of rules in $\Rules$ by a possibly different rule with possibly 
weaker premises and stronger conclusion. 
While these definitions are purely syntactic, a semantic characterisation has been given in
\cite{PattinsonSchroder10} in terms of \emph{one-step cut-free
completeness}. 
For many $\Lambda$-structures
including those for 
 the modal logic $K$ and the logic of (monotone)
neighbourhood frames, one-step cut-free complete rule sets are
known. In particular, these rule sets satisfy absorption of cut, 
contraction and congruence~\cite[\S~5]{PattinsonSchroder10}. 

\begin{lem}
\label{lem:absorb_cont_cut2absorb_mcut}
If the rule set $\Rules$ absorbs contraction and cut then $\Rules$ also absorbs multicut. 
\end{lem}

\noindent By Theorem \ref{thm:cut-elim-strongly} and Lemma \ref{lem:absorb_cont_cut2absorb_mcut}, we obtain the following. 

\begin{cor}
Suppose that $\Rules$ absorbs contraction and cut. Then $\rul{Cut}$ is also admissible in $\Seq \Rules$.\end{cor}

\noindent As an immediate corollary, we obtain completeness of the cut-free
calculus assuming that $\Rules$ is \emph{strongly} one-step
complete:

\begin{cor}
Suppose that $\Rules$ is strongly one-step complete. 
 Then $\models \Gamma \To \Delta$ iff
$\Seq\Rules \entails \Gamma \To \Delta$.
\end{cor}
\begin{proof}
This follows from Theorem \ref{thm:cut-elim-strongly} with the help
of Proposition 5.11 and 5.12 of \cite{PattinsonSchroder10}, the
latter asserting precisely the absorption of cut and congruence.
\end{proof}

\noindent The situation is more complex in presence of bounded
operators where completeness of the Hilbert calculus is only
guaranteed in presence of \axiom{BdPL}, and completeness of the
associated sequent calculus relies on $\rul{\Pastek}$. The
difficulty in a proof of cut-elimination is a cut-end derivation
where a cut is performed on $x \hearts \lcomp y_1 \col \phi_1
\rcomp \dots \lcomp y_n \col \phi_n \rcomp$ which is introduced by
$\Pastek$ and a (one-step) rule where the same formula
is principal. We leave this as an open problem:

\begin{prob}
\label{prob:pastecut}
Is there a way to modify the rules of
$\Bound\Seq\Rules$ so that completeness with respect to
$\Bound\Hilb\Rules$ holds and cut is admissible?
\end{prob}

\section{Conclusions and Further Work} \label{sec:conclusions}

We have introduced \emph{coalgebraic predicate logic}, a natural
first-order formalism that incorporates coalgebraic modalities and
thus serves as an expressive language for coalgebras. As instances, it
subsumes both standard relational first-order logic and Chang's
first-order logic of neighbourhood systems~\cite{Chang73}; other
instances include a first-order logic of nonmonotone conditionals as well
as first-order logics of integer-weighted relations that include
weighted or (positive) Presburger modalities. We have provided the foundations for proof theory and model theory of  CPL. 

In terms of future research, a promising avenue appears to be
coalgebraic finite model theory; in fact, the first result in this
direction is the existing finite version of the coalgebraic van
Benthem-Rosen
theorem~\cite{SchroderP10fossacs,LitakPSS12:icalp,SchroderPL15:jlc}.
It is worth observing that van Benthem-Rosen is a rare instance of a
model-theoretic characterization of a fragment of first-order
predicate logic that remains valid over finite models. The only other
major result of this type we are aware of is the characterization of
existential-positive formulas as exactly those preserved under
homomorphisms 
\cite{Rossman08:acm}. The result is relevant to constraint
satisfaction problems and to database theory, as existential-positive
formulas correspond to unions of conjunctive queries. Interestingly,
the proof of Rossman's result relies on Gaifman graphs, which also
play a central role in the proof of the coalgebraic Rosen theorem.

Embedding modal operators into a first-order syntax opens up the
possibility of applying modalities to predicates of arity greater
than~$1$; operators of this type are found, e.g., in Halpern's
Type-$1$ probabilistic first-order logic~\cite{Halpern90}. We leave
the ramifications of this option to future investigation.

Possible directions in coalgebraic model theory over unrestricted
models include generalizations of standard results of classical model
theory like Beth definability or interpolation and the Keisler-Shelah
characterization theorem. 


It remains to be seen which results of \emph{modal model theory} building upon the interplay between modal and predicate languages can be generalized. Specific potential examples include Sahlqvist-type results for suitably well-behaved structures and analogues of results by Fine (does
elementary generation imply canonicity, at least
wherever the coalgebraic J\'{o}nsson-Tarski theorem \cite{KupkeEA05} obtains?)\footnote{Recently, first results in this direction have been announced by Kentar\^o Yamamoto, UC Berkeley.} or Hodkinson \cite{Hodkinson2006:ndjfl}
(is there an algorithm generating a CML axiomatization
for CPL-definable classes of coalgebras?).  

Finally, a natural direction of investigation will be to study models
based on coalgebras for endofunctors on categories other than $\Set$
and corresponding variants of CPL with non-Boolean propositional
bases.



\section*{Acknowledgment}

As usual with matters of (non)compactness and (in)completeness, our gratitude to Erwin R. Catesbeiana is unbounded. We are very obliged to a particularly astute referee who spotted and made us iron out a number of issues. 
 We are also grateful to the editors of the Festschrift honoring Ji\v{r}\'i Ad\'amek for making us return to this material, and to  Ji\v{r}\'i himself for decades of broadening the frontiers  of category theory and coalgebra,  in computer science and elsewhere, without which there would be no reason for this Festschrift.



\bibliographystyle{alpha}
\bibliography{coalfol,coalgml,dirk}

\end{document}